\documentclass[
10pt,a4paper]{article}

\usepackage[utf8]{inputenc}%
\usepackage{amsmath}%
\usepackage{amsfonts}%
\usepackage{amssymb}%
\usepackage{makeidx} 
\usepackage{tikz}%
\usepackage{graphicx}%
\usepackage{accents}%
\usepackage{centernot}%
\usepackage{lmodern}
\usepackage{xcolor}%
\usepackage{float}%
\usepackage[left=1.6cm,right=1.6cm,top=2cm,bottom=2.5cm]{geometry}%
\usepackage[style=phys,biblabel=brackets,sorting=none,backend=biber, doi=false, url=false, isbn=true,eprint=true]{biblatex}%
\usepackage[colorlinks=true, linkcolor=blue, urlcolor=blue,citecolor=blue]{hyperref}%
\usepackage{multicol}%
\usepackage[capitalise,noabbrev,
nameinlink]{cleveref}
\usepackage{amsthm} 
\usepackage{xspace}
\usepackage{pdfcomment}

\usepackage{tcolorbox}
\tcbset{size=fbox,boxsep=3mm,colback=white,colframe=black}

\crefname{sub}{subroutine}{subroutines}
\creflabelformat{sub}{#2{\upshape(#1)}#3} 

\graphicspath{{./}{Graphics/}}
\usepackage[labelfont=bf]{caption} 
\renewenvironment{figure}[1][1]{\trivlist\item}{\endtrivlist}

\newtheoremstyle{my}
{3ex}
{1.5ex}
{\itshape}
{}
{\bfseries}
{}
{1ex}
{\thmname{\textbf{#1}}\thmnumber{ \textbf{#2}}\textbf{:}\thmnote{ (#3)}}%

\newtheoremstyle{my2}
{2ex}
{1.5ex}
{}
{}
{\bfseries}
{}
{1ex}
{\thmname{\textbf{#1}}\thmnumber{ \textbf{#2}}\textbf{:}\thmnote{ (#3)}}%

\theoremstyle{my}
\newtheorem{definition}{Definition}[section]

\crefname{observation}{observation}{observations}

\newtheorem{theorem}{Theorem}[section]
\newtheorem{lemma}{Lemma}[section]

\theoremstyle{my2}

\newtheorem{algorithm}{Algorithm}[section]
\newtheorem{subroutine}{Subroutine}[section]

\newcommand{\I}{\ensuremath{\mathcal{I}}\xspace}
\newcommand{\X}{\ensuremath{\mathcal{X}}\xspace}
\newcommand{\Y}{\ensuremath{\mathcal{Y}}\xspace}
\newcommand{\Z}{\ensuremath{\mathcal{Z}}\xspace}
\renewcommand{\H}{\ensuremath{\mathcal{H}}\xspace}
\renewcommand{\S}{\ensuremath{\mathcal{S}}\xspace}
\renewcommand{\O}{\ensuremath{\mathcal{O}}\xspace}
\newcommand{\U}{\ensuremath{\mathcal{U}}\xspace}
\newcommand{\V}{\ensuremath{\mathcal{V}}\xspace}
\newcommand{\F}{\ensuremath{\mathcal{F}}\xspace}
\newcommand{\G}{\ensuremath{\mathcal{G}}\xspace}
\newcommand{\T}{\ensuremath{\mathcal{T}}\xspace}
\newcommand{\M}{\ensuremath{\mathcal{M}}\xspace}
\newcommand{\CNOT}{\ensuremath{\mathcal{CNOT}}\xspace}
\newcommand{\CZ}{\ensuremath{\mathcal{CZ}}\xspace}
\newcommand{\CV}{\ensuremath{\mathcal{CV}}\xspace}
\newcommand{\SWAP}{\ensuremath{\mathcal{SWAP}}\xspace}
\newcommand{\MAJ}{\ensuremath{\mathcal{MAJ}}\xspace}

\title{\textbf{Quantum Simulation Logic, Oracles, and the Quantum Advantage}}
\author{Niklas Johansson,$^*$ Jan-Åke Larsson$^{\dagger}$\\\small
\textit{Institutionen för Systemteknik, Linköpings Universitet, 581 83
Linköping, SWEDEN}} \date{}

\definecolor{Green}{rgb}{.1,.5,.1}

\newcommand{\ket}[1]{\ensuremath{\left\vert{#1}\right\rangle}}
\newcommand{\bra}[1]{\ensuremath{\left\langle{#1}\right\vert}}
\newcommand{\inner}[2]{\ensuremath{\left\langle{#1}\middle\vert{#2}\right\rangle}}
\newcommand{\average}[2]{\ensuremath{\left\langle{#1}\middle\vert{#2}\middle\vert{#1}\right\rangle}}
\renewcommand{\outer}[2]{\ensuremath{\left\vert{#1}\middle\rangle\middle\langle{#2}\right\vert}}

\DeclareMathOperator{\hypen}{-}

\begin{document}

\maketitle
\begin{quote}
	Query complexity is a common tool for comparing quantum and classical computation, and it has produced many examples of how quantum algorithms differ from classical ones. 
	Here we investigate in detail the role that oracles play for the advantage of quantum algorithms. 
	We do so by using a simulation framework, Quantum Simulation Logic (QSL), to construct oracles and algorithms that solve some problems with the same success probability and number of queries as the quantum algorithms. 
	The framework can be simulated using only classical resources at a constant overhead as compared to the quantum resources used in quantum computation. 
	Our results clarify the assumptions made and the conditions needed when using quantum oracles. 
	Using the same assumptions on oracles within the simulation framework we show that for some specific algorithms, like the \textsc{Deutsch-Jozsa} and \textsc{Simon's} algorithms, there simply is no advantage in terms of query complexity. 
  This does not detract from the fact that quantum query complexity provides examples of how a quantum computer can be expected to behave, which in turn has proved useful for finding new quantum algorithms outside of the oracle paradigm, where the most prominent example is \textsc{Shor's} algorithm for integer factorization.

	\hfill PACS number(s): 03.67.Ac, 03.67.Lx
\end{quote}

\tableofcontents

\vspace{2cm}
\begin{multicols}{2}

\noindent One of the main problems in quantum information theory is to
understand the resources needed for quantum computation. Here, a resource is a
property used to enable the computational speed-up, usually described as
something consumed by the computation, but it could also be a static resource:
a property present in one computational model but absent in another. Some
properties that have been proposed as resources for quantum computation are
interference~\cite{Feynman1982}, entanglement~\cite{Einstein1935},
nonlocality~\cite{Bell1964},
contextuality~\cite{Kochen1967,Kleinmann2011,Howard2014}, and even
coherence~\cite{Baumgratz2014,Hillery2016}. 
It could also be that different algorithms makes use of
different resources, or different combinations of resources, which further
motivates resource studies in the context of specific algorithms.

From our perspective contextuality is clearly a top contender. If contextuality
is the intrinsic property that gives quantum computers their advantage one would
expect to find evidence of this even in simple algorithms. Our first attempt to
study this started with the \textsc{Deutsch-Jozsa} algorithm
\cite{Deutsch1985,Deutsch1992}, that determines whether a given function is
constant or balanced. This was the first proposed quantum algorithm that showed
an exponential speed-up compared to the best-known classical algorithm, although
this speed-up only holds when requiring the solution to be deterministic. There
are existing results on resources for this algorithm, for example,
\textcite{Collins1998} shows that for one and two qubits there is one
implementation where entanglement is not needed. While this does not rule out
implementations where entanglement is present, it does show that entanglement is
not \emph{strictly} needed for the algorithm to work. The question remains if
contextuality is needed. Notable is that for one- and two-qubit input, the
oracles can be implemented with a stabilizer circuit, and while stabilizer
circuits can be efficiently simulated \cite{Gottesman1998}, they can also
produce phenomena like nonlocal and, in particular, contextual correlations. A
similar framework, closely related to the stabilizer subtheory, is
Spekkens' model~\cite{Spekkens2007} which contains similar transformations, but
does not produce nonlocal and contextual correlations. Spekkens' model gives the
same behavior as the stabilizer subtheory of quantum mechanics for one- and
two-qubit input functions.

Our question at this point was simply: what needs to be added to Spekkens' model
to enable three-qubit inputs (or larger). To our surprise, what was needed was
to add a new gate to Spekkens' model \parencite{Johansson2017}, without allowing
for nonlocal correlations and without making it contextual. This addition
enables all balanced functions on an arbitrarily large input, not only a subset
thereof. In our paper, we show that for each problem instance there exists a
realization in this new extended model, so that all balanced and constant
functions at any size of the input can be reproduced. In addition, we show that
the standard Deutsch-Jozsa quantum algorithm, using the gates of the model,
solves the problem with a single query.

Since nonlocal and contextual correlations are manifestly absent from the
framework, the possibility of stating and solving each problem instance within
the model, rules out these properties from being enabling properties. As before,
this does not rule out implementations where entanglement or contextuality is
present, but shows that they are not \emph{strictly} needed.

Further, the model can be efficiently simulated on a classical probabilistic
Turing machine, i.e., the algorithm is efficiently simulatable in a classical
probabilistic Turing machine --- with constant overhead --- meaning that there
is no speed-up when comparing to the algorithm run on a quantum Turing machine.

Having established this framework, we turned to a second algorithm,
\textsc{Simon's} algorithm \cite{Simon1994,Simon1997} that finds the generator
of a hidden order-two subgroup. This algorithm is usually portrayed as the
poster-child for quantum \emph{exponential} speed-up, and is regarded as
stronger evidence since the speed-up remains even if the solution is accepted
with a bounded error-probability. Also here our framework reproduces the
behavior of \textsc{Simon's} quantum algorithm  \cite{Johansson2017}. Thus,
contextuality is not needed, and the algorithm is efficiently simulatable in a
classical probabilistic Turing machine.

Both algorithms assume
access to the function through a black-box oracle, meaning that the solver only
has access to the function's input and output, not the internal structure of the
oracle.

The framework used in \cite{Johansson2017} is limited and puts hard structural
constraints on the setup. In this paper, we extend the framework, allowing for a
less strict setup resulting in an approximation that gives predictions closer to
those of quantum theory. This builds further on the main idea of
\cite{Johansson2017} that the resources	needed is the ability \emph{to choose}
to store, process, and retrieve information from an additional
information-carrying degree-of-freedom of the physical system.

We begin with some preliminaries (\cref{sec:Preliminaries}) and a thorough
introduction to our framework (\cref{sec:QSL}). We then start with a simple
oracle problem called the \textsc{Bernstein-Vazirani} problem (\cref{sec:Bernstein-Vazirani}), for which our model completely reproduces the results of the
quantum algorithm. We also show that our approach works even when
\textsc{Deutsch-Jozsa} is considered both as a promise problem and as a decision problem  (see
\cref{sec:Deutsch-Jozsa}). For \textsc{Grover Search}
there is a speed-up, but not as much as in the quantum case
(\cref{sec:Grover's}), and also for \textsc{Simon's} problem our model completely reproduces the results of the
quantum algorithm (\cref{sec:Simon's}).

A more pressing question that these results point to, is whether the oracle model really can produce conclusive evidence for a separation between quantum and classical computation. 
We tend to believe it cannot, and definitely not for \textsc{Deutsch-Jozsa} and
\textsc{Simon's} problems, see the discussion in \cref{sec:Oracles}.	
However, quantum algorithms that produce bounds for quantum query complexity do provide us with examples of how a quantum computer could be expected to behave, and inspire us to invent new approaches to non-relativized computational problems. \textsc{Shor's} algorithm\cite {Shor1994,Shor1999} is one such example.

We do point out that there is no principal reason prohibiting our model (or a similar one) to have an effect outside of the oracle paradigm, and that our model is best compared with a construction suffering from systematic errors. In this spirit, we compare our framework with a current state-of-the-art implementation of \textsc{Shor's} algorithm (\cref{sec:Shor's}). 

\section{Preliminaries}\label{sec:Preliminaries}

In this section we go through some relevant theory helpful to understand the content of this paper, namely the concepts of Turing machines and oracles, and also a brief introduction to quantum computation. 

\subsection{Turing Machines}

The concepts of Probabilistic Turing Machines (PTMs) and
Quantum Turing Machines (QTMs) are essential here. While it is possible to use 
the formal notion of a Turing Machine that manipulates classical or quantum 
symbols on a strip of tape according to a table of rules, the comparison is 
made simpler by using the presentation of \textcite{Bernstein1993,Bernstein1997}.

In their physics-like view of a PTM, they describe its time evolution as a
sequence of probability distributions. At each time step, the distribution
describes the likelihood of the machine's configuration, and the probability
assigned to any specific configuration is a real number in $[0,1]$, to within a
precision of $2^{-m}$ for some integer $m$. With a machine able to do coin-flips, these numbers are reachable with
a deterministic algorithm in time polynomial in $m$. The probability
distribution over all configurations at time $t$ is represented by a vector
$\vec{v}_t$. The transition function taking the machine into the next
configuration can be thought of as a stochastic matrix $M$ whose rows and
columns index configurations, entries are probabilities, and the columns sum
to one. 

\citeauthor{Bernstein1997} point out that stochastic matrices obtained from probabilistic TM are finitely specified and map each configuration by making only local changes to it, but even so, the support of the probability distribution can be exponential in the running time of the machine.
Observing the PTM (or part of it) at a time step will yield a
configuration (or partial configuration) sampled from the probability
distribution, and the distribution is updated conditioned on the value
observed, i.e., sampling shrinks the support of the distribution to only cover
the observed configuration. This update does not change the future behavior of the machine,
because the involved probabilities are additive. It is therefore only necessary to
keep track of which configuration the machine is in, at each time step. So,
even though the distribution may have a support that grows exponentially with
running time, it is only necessary to keep track of a constant amount of
information to trace the behavior of the machine as time progresses.

A similar physics-like view of a QTM assigns complex numbers called amplitudes, instead of probabilities, 
to the possible configurations. The resulting vector
of complex amplitudes describes the QTM's quantum state at each time step, as a
linear combination of configurations known as a superposition, and the
probability of observing a specific configuration is given by the absolute
square of its amplitude~\cite{Born1954}. Further, observation of a measurement
outcome enforces the state of the machine to be updated to be consistent with
the outcome according to \emph{Lüders rule}~\cite{Luders2006}. The quantum time
evolution needs to keep track of the amplitude and phase of all the
configurations because the involved quantities are not only added to each
other; the components may cancel each other because of the phase. This
phenomenon is known as interference, and is needed to reproduce the quantum
behavior. Observation would restrict the configuration so that the interference
is prohibited and therefore, in a QTM, observation may disturb its state and its
later behavior. The result is that the possible exponential growth of the
superposition support needs to be retained to reproduce the complete quantum
behavior. In other words, it may be necessary to keep track of an exponentially
increasing amount of information to trace the behavior of the machine as time progresses.

More details and the correspondence with the formal Turing Machine model can
be found in \textcite{Bernstein1993,Bernstein1997}. In what follows the time
evolution will be divided into chunks known as gates as we will be using the
circuit model of quantum computation. This has been shown to be polynomially
equivalent to the QTM model by \textcite{Yao1993}.

\subsection{Oracle Turing Machines and oracle notions}
\label{sec:OracleDefinition}

A brief deviation into the standard definition of an oracle TM is needed, before we translate it into the physics-like description.
A formal definition can be found in \textcite[73]{Arora2009a}, but here we will adopt the more suitable description of \cite{Bennett1997}.
An oracle TM has a special \textit{query tape} and two distinguished internal states: a pre-query state 
and a post-query state. 
A~query is executed whenever the machine enters the pre-query state, and causes the oracle $f$ to evaluate the query string present on the query tape (if no query is present the oracle performs a no-op).
The query string should be written in the form $x||b$, i.e., the input to the oracle $x$ concatenated with a target bit $b$, commonly initialized to 0, onto which the oracle's answer is added modulo two.
(Except for the query tape and change of internal state from pre- to post-query state, other parts of the oracle TM do not change during the query.)
The target bit is strictly speaking not needed for a classical oracle TM, since the answer could be overwritten onto the space used for the query string, but it can be added without loss of generality.
More importantly, the target bit enables reversion of the oracle, sometimes known as ``uncomputation'' \cite{Bennett1973}, since the input equals $x$ and the target bit is already equal to $f(x)$, calling the oracle again will reset the target bit to $0$. 
Extending this to several target bits is simple and can be done at a linear cost corresponding to one evaluation per target bit, most presentations condense this to a single evaluation per target bitstring.

Informally~\cite{Bennett1997} an oracle is a
device that lets the machine evaluate some function at unit cost. Effectively, using
an oracle the question becomes: \emph{if we could compute this function
efficiently, what else could we then compute?} With this in mind, oracles can
be thought of as a tool used to calculate a lower bound on the resource
requirement.

Reformulating this in the already mentioned physics-like description, we describe the evaluation of the oracle as a change in a single time-step, so that the evaluation has unit cost.
We will later implement this physics-like description in terms of quantum or classical reversible  gates that specify the map in each time-step.
In this circuit model, an oracle constitutes a part of the circuit that counts as one single gate when analyzing its resource requirements. 
Normally, a circuit implementation of a function $f$ uses a query register from which the input is read, and an answer register where the result is added, bitwise modulo two (see \cref{fig:reversible_f}). 
The query register remains unchanged by the gate; this enables reversibility of the circuit implementation even for non-bijective functions, similar to the oracle TM model above. 
A simple extension of the argument of \textcite{Yao1993} can be used to show that this model is polynomially equivalent with the oracle TM model, taking into account that each output bit needs to be handled by a separate oracle in the formal model.

\begin{figure}[t]
	\centering
	\includegraphics[scale=.8]{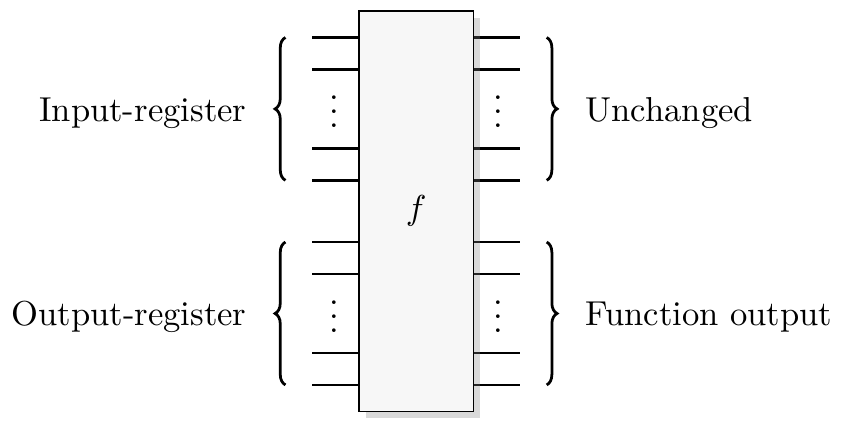}
	\captionof{figure}{A classical reversible logical function construction for 		a function $f$.} \label{fig:reversible_f}
\end{figure}

In \emph{oracle query complexity} the relevant quantity is the number of calls an algorithm
makes to the oracle. Each call is counted as one operation, or one unit of
time, and this is the measure that will be used throughout this paper. As mentioned
above, most of the problems that we will consider are \emph{decision} and \emph{promise problems}. The solver of a promise problem is guaranteed that the function is one of a few possible types,
 and the problem is to determine which it is.

It is conventional in computational complexity to consider the oracle as a
black box. Sometimes however, and especially in the related field of cryptography,
a wider definition is used to better capture the problem. 
In the context of analyzing the security of a cryptographic algorithm, an
attack model or security model needs to be specified, stating
assumptions on the power, knowledge, and access to the system available to an adversary. \textcite{Chow2002} explains that in the black-box model the adversary is
restricted to observe input and output from the cryptographic algorithm, in
contrast to the white-box model where he or she is assumed to have complete
access to both the specific software implementation, and to the running environment of the algorithm. A gray-box model is a mixture between these two extremes,
and applicable whenever the adversary has access to some internal details
of the algorithm, but not all.

As an example, consider that we know (are promised) the mapping of a
cryptographic function and the problem is to determine an input parameter (a
key) based on the output of the function. If access is given to the function in
the form of a black-box, then the problem is usually intractable. However, if the
function runs on some hardware, then a more realistic model is that we can
monitor its power consumption. This additional information may open up for a
\emph{side-channel} attack known as power analysis; information about the key
leaks out through the power consumption, and retrieving the key becomes
tractable (see e.g., \cite{Fransson2015}). 
Side channels naturally appear when the protocol is implemented in a
physical system. This is an example of a gray-box model since it
only requires having an incomplete description of the machinery, we only need to
have knowledge about the side effect resulting from the physical implementation,
and not the full description as in the white-box model.

Another example where these notions fit well is computational modeling and
system identification~\cite[43]{Bosch1994}, where the goal is instead to model
or identify a process. In the white-box setting, we have complete knowledge
about the process, and a model can be built from first principles. In the
black-box setting we have no prior knowledge about the underlying process, and
to our help is only the statistics generated by the input-output process. A
gray-box model corresponds to something in between, when we have partial
knowledge about the process.

In the circuit model the different black- and non-black-box models will translate into the following

\begin{itemize}
	\item The white-box model assumes that we know everything that we can know about implementation of the oracle. That is, we know the specific circuit that performs the function, and even how the computing machinery is built and operates.

	\item The black-box model assumes that we know nothing about the implementation of the oracles. We can only access inputs to and outputs from the circuit.

	\item A gray-box is some specified mixture of black and white, for instance, it could be that we have more knowledge about how the function is implemented, but not complete knowledge.
\end{itemize}		

\subsection{Quantum Computation}

A QTM computer is a machine that operates on quantum bits, or qubits, rather than bits. 
Qubits are the elementary information carriers in quantum theory. 
They are two-level systems that can be parameterized by two complex numbers $a$ and $b$ through the expression
\begin{equation}
	a\ket{0} +b\ket{1},
	\label{eq:qubit}
\end{equation}
where $\ket{0}$ and $\ket{1}$ are the two orthogonal eigenstates to the Pauli-Z operator,
\begin{equation}
	Z\ket{0}=\ket{0}, \qquad Z\ket{1}=-\ket{1}.
\end{equation} 
These states are usually referred to as the computational basis states. 

We will use two kinds of operations on these qubits, \textit{reversible transformations} and \textit{measurements}. 
\textit{Transformations} relate to the classical bit operations and are described by unitary operators
\begin{equation}
	AA^\dagger=A^\dagger A=I,
\end{equation}
where ($^\dagger$) is the Hermitian conjugate, and $I$ the identity operator.
These operators are also referred to as quantum gates.
\textit{Measurements} are used to retrieve information from these systems, and of
particular use are projective measurements. In such a measurement the state is projected onto
an eigenstate of an observable (a Hermitian operator), and the information retrieved by an
observer from such a measurement is the eigenvalue corresponding
to that eigenstate. The probability of obtaining a specific outcome from a
specific observable depends on the state before measurement, and is given by
the \emph{Born rule}~\cite{Born1954}. The probability is given by the absolute square of the component (amplitude) in the direction of the eigenstate
related to the outcome. As an example, the probability of finding the qubit
represented by expression (\ref{eq:qubit}) in state $\ket{0}$, from
measuring the $Z$ observable, is $|a|^2$.	 
Furthermore, the state of the physical system is changed in a projective measurement to the eigenstate in question (or a vector in the eigenspace, in the degenerate case). The state is projected onto the eigenspace, hence the term projective measurement. The requirement that the probabilities of the different outcomes sums to one translates into normalization of the state so that $|\psi|^2=1$.

There exists a third operation: \textit{preparation}, that can be thought of as measurement on an unknown state followed by a unitary transformation that depends on the measurement outcome, such that the output eigenstate is rotated to the desired state. 
For a more thorough explanation of the primitives (and their generalizations that we do not  treat here), see \textcite{Nielsen2010}.

An important concept in quantum information theory is that of \emph{mutually
unbiased bases}. Take two sets of states $\{\ket{e_i}\}$ and $\{\ket{f_j}\}$,
both being bases for the quantum system. These two bases are said to be
mutually unbiased if any pair of two states formed between the sets satisfy
the condition
\begin{equation}
	\big|\inner{e_i}{f_j}\big|^2=\frac{1}{d},
\end{equation} 
where $d$ is the dimension of the system. Importantly, the
left-hand side is constant, independent of bases and states, so that
information retrieved from a projective measurement along one basis is
completely unrelated to the information retrieved from a projective
measurement in the other basis. \cite{Bengtsson2007}

As an example, the bases spanned by the eigenstates of the Pauli operators 
$X,Y$ and $Z$ are 2-dimensional and mutually unbiased.
\begin{equation}
	\{\ket{0},\ket{1}\}_Z, \quad \{\ket{+},\ket{-}\}_X, \quad
	\{\ket{+i},\ket{-i}\}_Y
\end{equation} 

A transition between the computational basis and the \emph{phase basis},
spanned by the eigenstates of $X$, can be performed by the Hadamard gate $H=H^{\dagger}$ by
\begin{equation}
	H\ket0=\ket{+},\quad H\ket1=\ket{-}.
\end{equation}
It also transform the Pauli operators according to
\begin{equation}
	HZH=X,\quad HXH=Z,\quad HYH=-Y.
\end{equation}

Composition of systems of several qubits uses the tensor product, for example
\begin{equation}
\ket{\psi}_{AB}=\ket{\phi}_A \otimes \ket{\varphi}_B=\ket{\phi}_A\ket{\varphi}_B,
\label{eq:separable}
\end{equation} 
where the subscripts $A$ and $B$ indicates that the states describe two subsystems, but these are often dropped from the notation. 
The dimension of the tensor product vector space is the product of the constituent spaces, and therefore, computations involving $n$ qubits will have states described by a $2^n$-dimensional Hilbert space.
States that can be written on the form \eqref{eq:separable} are called product states. 
The tensor product creates a Hilbert space (here, finite-dimensional complex vector space) of all linear combinations of the possible separable states.
Some of these linear combinations cannot be written as a product state as in  \cref{eq:separable}, and these are called entangled states. 

The most well-known two-qubit gate is the quantum \textit{CNOT} gate. 
A \textit{CNOT} is a controlled-X operation; it applies $X$ to the target qubit conditioned on the control qubit. 
The graphical presentation is a dot on the control qubit, connected with a control line to the target gate, here an $X$.
For some examples of gate arrangements, see \cref{fig:cnot_identity}.
With $\ket{c}$ as control and $\ket{t}$ as target it has the effect
\begin{equation}\label{eq:quantum_CNOT}
CNOT\ket{c}\ket{t}=\ket{c}\ket{t\oplus c},
\end{equation}
where ``$\oplus$'' is the exclusive-OR, or addition modulo 2 so that
\begin{equation}
\begin{split}
&\ket{00}\mapsto\ket{00}\\
&\ket{01}\mapsto\ket{01}\\
&\ket{10}\mapsto\ket{11}\\
&\ket{11}\mapsto\ket{10}.
\end{split}
\label{eq:cnotcomputational}
\end{equation}
This is the effect as seen from the computational basis.
From the point of view of the phase basis, transforming the operator via Hadamard transforms, we obtain a different map (see \cref{fig:cnot_identity}A), and it is easy to verify that, for example
\begin{equation}
\begin{split}
\ket{+-}=\tfrac12&\ket{00}+\tfrac12\ket{01}+\tfrac12\ket{10}-\tfrac12\ket{11}\\
\mapsto  \tfrac12&\ket{00}+\tfrac12\ket{01}-\tfrac12\ket{10}+\tfrac12\ket{11}=\ket{--},
\end{split}\end{equation}
in total
\begin{equation}
\begin{split}
&\ket{++}\mapsto\ket{++}\\
&\ket{+-}\mapsto\ket{--}\\
&\ket{-+}\mapsto\ket{-+}\\
&\ket{--}\mapsto\ket{+-},
\end{split}
\label{eq:cnotphase}
\end{equation}
that is,
\begin{equation}
H^{\otimes 2}~CNOT~H^{\otimes 2}\ket{c}\ket{t}=\ket{c\oplus t}\ket{t}.
\end{equation}
This effect in the phase basis is sometimes called phase
kickback~\cite{Cleve1998}. 

A controlled phase flip (Controlled-$Z$ or $CZ$) and a qubit swap (\textit{SWAP}), that corresponds to classically swapping the qubits \cite{Beckman1996}, can also be constructed directly from the \textit{CNOT} (see \cref{fig:cnot_identity}A). The $CZ$ and \textit{SWAP} behaviors correspond to small modifications of \cref{eq:cnotcomputational,eq:cnotphase}.

\begin{figure}
	\centering
	\raisebox{2cm}{\textbf{A.}}\includegraphics[scale=01]{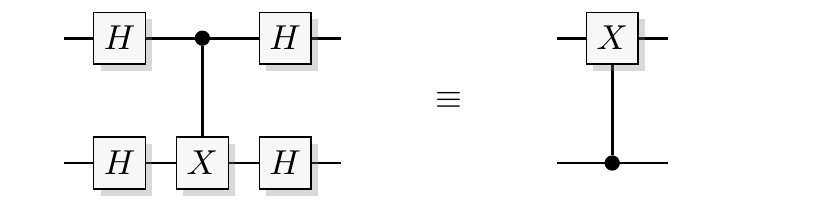}\\\bigskip
	\raisebox{2cm}{\textbf{B.}}\includegraphics[scale=01]{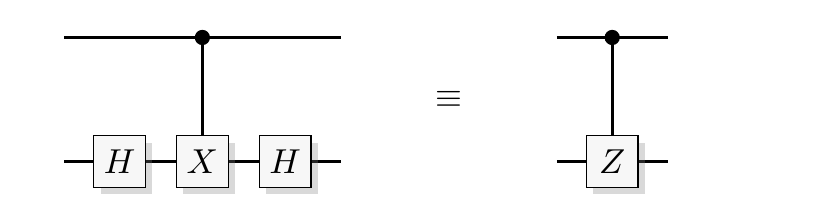}\\\bigskip
	\raisebox{2cm}{\textbf{C.}}\includegraphics[scale=01]{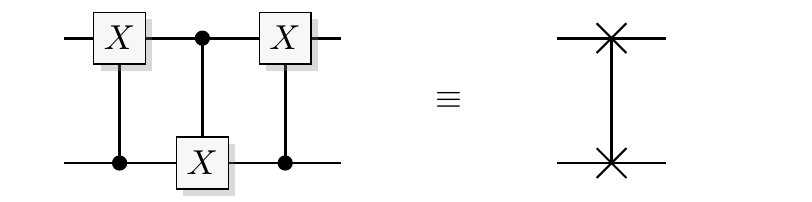}
	\captionof{figure}{Identities for two qubit quantum gates. \textbf{A.} The
effect of a quantum CNOT in the Hadamard basis. \textbf{B.} relation
between CNOT and CZ. \textbf{C.} construction of a SWAP gate from three CNOT
gates.} \label{fig:cnot_identity}
\end{figure}

Toffoli and Fredkin gates are controlled versions of the CNOT and SWAP, respectively.
A Toffoli gives the map
\begin{equation}\label{eq:quantum_Toffoli}
\textit{Toffoli}\ket{c}\ket{c'}\ket{t}=\ket{c}\ket{c'}\ket{t\oplus cc'},
\end{equation}
while a Fredkin gives
\begin{equation}\label{eq:quantum_Toffoli}
\textit{Fredkin}\ket{c}\ket{t}\ket{t'}=\ket{c}\ket{t\oplus c(t\oplus t')}\ket{t' \oplus c (t\oplus t')},
\end{equation}

A Toffoli gate with $ n $ controlling qubits is called an $ n $-Toffoli. Also, inverted (white) controls enables the gate if the control system is in state $ \ket{0} $, and can be implemented by applying an $X$ before and after the control.

Since functions in quantum computation are constructed from unitary transformations, any function implemented needs to be reversible. 
The standard approach is to use the circuit model wherein the computation is described by a reversible circuit, and then replace the classical reversible gates by their quantum equivalents. 
Then, the circuit constitutes a unitary operation with the same
mapping of the computational basis states as the classical analogue of a
reversible circuit operating on bits. In contrast to the classical reversible function, there may be additional information to retrieve from
the output of the query register. There is at least one additional
information carrying degree-of-freedom that we can choose to
retrieve information from, instead of retrieving the result of the computation. Of course, in classical reversible logic, we cannot expect to retrieve any additional information about the computation from the query register after such a transformation,
see \cref{fig:additional_degree}.
\begin{figure}
	\centering
	\includegraphics[scale=.8]{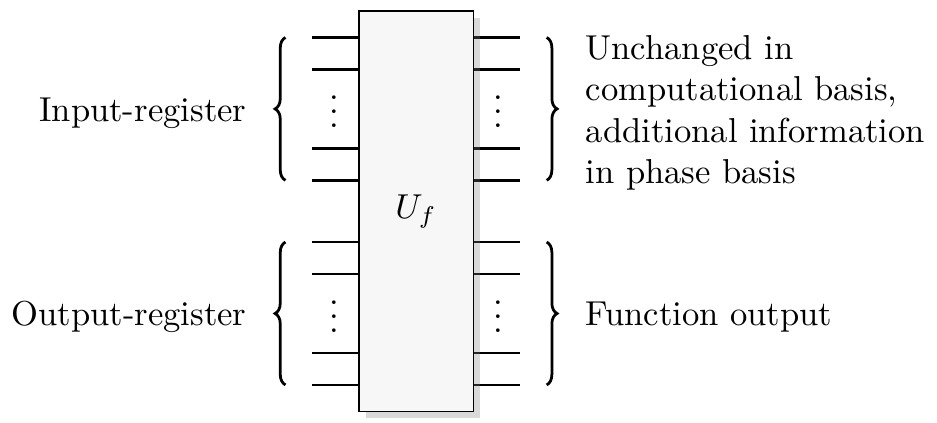}
\captionof{figure}{A reversible logical function.  In the quantum case
we can choose to retrieve function output or some additional information
from the output of the query register.} \label{fig:additional_degree}
\end{figure}

A quantum oracle is an oracle assumed to be a specific unitary
transformation acting on qubits (or other quantum systems) rather than on
classical bits. (For a formal definition in the QTM model, see
\textcite{Bennett1997}.) Such an oracle is a unitary transformation, so it is
necessarily a reversible map from qubits to qubits. A quantum oracle is
usually described by a classical reversible logical function encoded into
the computational basis of a unitary operator. Having access to such an oracle a quantum computer can
exploit the ability to sample from some distributions related to the
function, rather than just having access to the function itself.
This will be further discussed in \cref{sec:Oracles}. 

\section{Quantum Simulation Logic}\label{sec:QSL}

In what follows we will use Quantum Simulation Logic (QSL), which extends an earlier model for simulation of quantum mechanical systems known as Spekkens' toy model~\cite{Spekkens2007}. 
That model views the quantum state as a state of the observer's knowledge, an \emph{epistemic state}, represented as probability distributions over a set of \emph{ontic states}, and the ontic state is hidden away from the observer. 
The word hidden is key here, because the model restricts the knowledge an observer can have about the ontic state of the system; \textit{the knowledge balance principle} states that the amount of knowledge an observer has about the system is at most equally large as the amount of knowledge an observer lacks about the system. 
This implies that there is a restriction to the allowed probability distributions that are used to describe states in the model.

Spekkens' model associates the ontic states of a qubit system with the four points of the corresponding phase space. QSL indexes these four points using two bits of information, so that the ontic state is represented through two classical bits, details will follow below.
Since the main part of the paper is about computation, the words ontic and epistemic will not be used below, but those acquainted with Spekkens' model may find it useful to remember that the basic bit values in the QSL model are associated with Spekkens' ontic states, and a probabilistic combination of different bit-values correspond to Spekkens' epistemic states, or indeed, correspond directly to quantum states of the simulated quantum system.

In this section we will see several examples of how QSL is similar to quantum theory. Some, but not all, of these examples can be stated directly in Spekkens' toy model, and a few are even present in Ref.~\cite{Spekkens2007}.  
We include these to show how easy it is to follow the time evolution of a QSL system, protocol, or algorithm, as compared to the standard representation of Spekkens' model, and to illustrate and get a feeling of how the QSL framework works.
One should be aware that QSL extends beyond Spekkens' toy model (through the Toffoli as described below), so some properties of Spekkens' model do not hold in QSL, but QSL can in turn reproduce phenomena not reproducible in Spekkens' model.

\subsection{Elementary Systems}\label{sec:QSL_Elementary_Systems}

Each elementary QSL system is constructed from two classical information
carriers that each can carry one bit of information. 
Yet, preparation and measurement only allow for storage of one bit of information or retrieval of one bit of information, while the other is destroyed through randomization, just as in Spekkens' model. 
This gives the following description.

\subsubsection{States}

The state of an elementary system is represented by a tuple
\begin{equation}
(x_0,p_0),
\end{equation}
where $x_0$ and $p_0$ represent the bit values of the two classical carriers.

Preparing an elementary system to have a definite value in the $x_0$ bit corresponds to preparing a quantum system in the computational basis state $\ket{x_0}$, an eigenstate of the $Z$ observable. 
For this reason, $x_0$ will be referred to as the \emph{computational bit}.
Since the elementary system only can carry one bit of information, the second bit $p_0$ is in this case necessarily described by an equally weighted random variable $R\in \{0,1\}$.	
In other words, to simulate the eigenstates of the $Z$ observable we have
\begin{equation}
(0,R)\sim\ket{0}, \qquad (1,R)\sim\ket{1}.
\label{eq:QSLZbasis}
\end{equation}
Similarly, preparing an elementary system to have a definite value of the $p_0$ bit corresponds to preparing a quantum system in a phase basis state,
\begin{equation}
(R,0)\sim \ket+ = \frac{\ket{0}+\ket{1}}{\sqrt{2}}, \qquad (R,1)\sim
\ket-=\frac{\ket{0}-\ket{1}}{\sqrt{2}}.
\label{eq:QSLXbasis}
\end{equation}
These give us a simulation of eigenstates of the $X$ observable. For this reason, $p_0$ will be referred to as the \emph{phase bit}.
Also here, the other bit, in this case the computational bit,
must be random. 

Note that the $\ket1$ state is associated with $x_0=1$ and the $-1$ eigenvalue of the $Z$ observable, and similarly, the $\ket-$ state is associated with $p_0=1$ and the $-1$ eigenvalue of the $X$ observable. In what follows, the $X$ observable will be used interchangeably with the projection $(I-X)/2=\ket-\bra-$, differing only in eigenvalues while retaining eigenvectors.

This brings us to the first analogy with quantum theory, namely the
\emph{uncertainty principle}. An observer cannot simultaneously make a value
assignment to the complementary observables $X$ and $Z$, and having perfect knowledge about one
implies having none about the other. This is true in both QSL and in quantum
theory.

The information can also be stored in the correlation between the
computational bit and phase bit. They are then either correlated or
anti-correlated, and this simulates the eigenstates of the $Y$ observable.
\begin{equation}
(R,R)\sim \frac{\ket{0}+i\ket{1}}{\sqrt{2}}, \qquad (R,R\oplus 1)\sim
\frac{\ket{0}-i\ket{1}}{\sqrt{2}}
\label{eq:QSLYbasis}
\end{equation}
Also here, each value of the bit XOR is associated with the eigenspaces of the $(I-Y)/2$ observable. 

These six are all the pure qubit states that we can simulate with a single elementary QSL-system. In
addition, we can represent the maximally mixed state that encodes that we
have no knowledge at all about the system. 
The natural way to do that is
to have both the computational and phase bit represented by two independent
random variables.
\begin{equation}
(R_x,R_p)\sim \frac{I}{2}
\end{equation}
These seven QSL states are equivalent to Spekkens' epistemic states \cite{Spekkens2007}: the six pure states, and the maximally mixed state.

An elementary QSL system has a sample space of four discrete points,
corresponding to the four ontic states in Spekkens' model (the description with the four boxes $\square\square\square\square$ appears in \cite{Spekkens2007}),
\begin{equation}
\begin{split}
(x_0,p_0)=\ &\underset{ 0 \hspace{.9ex} 1\hspace{.9ex} 2\hspace{.9ex} 3}{\square\square\square\square}\sim\ket{\psi}\\
(0,R)=\ &\blacksquare\blacksquare\square\square\sim\ket{0}\\
(1,R)=\ &\square\square\blacksquare\blacksquare\sim\ket{1}\\
(R,0)=\ &\blacksquare\square\blacksquare\square\sim\ket{+}\\
(R,1)=\ &\square\blacksquare\square\blacksquare\sim\ket{-}\\
(R,R\oplus 1)=\ &\square\blacksquare\blacksquare\square\sim\ket{+i}\\
(R,R)=\ &\blacksquare\square\square\blacksquare\sim\ket{-i}\\
(R_x,R_p)=\ &\blacksquare\blacksquare\blacksquare\blacksquare\sim\frac{I}{2},
\end{split}
\label{eq:spekkensstates}
\end{equation}
and with $p_0$ as the least significant bit (LSB) they have the canonical labeling
\begin{equation}
\begin{split}
(x_0,p_0)=\ &\underset{ 0 \hspace{.9ex}1\hspace{.9ex} 2\hspace{.9ex} 3}{\square\square\square\square}\\
(0,0)=\ &0\\
(0,1)=\ &1\\
(1,0)=\ &2\\
(1,1)=\ &3.\\
\end{split}\label{eq:canonical_labeling}
\end{equation}

Perhaps the most used method of calculating the similarity between two
quantum states $\ket{\psi}$ and $\ket{\phi}$ is the absolute square of
their inner product, $|\inner{\psi}{\phi}|^2$. This is called the
\emph{statistical overlap} or \emph{transition
probability}~\cite{Uhlmann1976}, and take on values between 0 and 1. If
the states are parallel, then $|\inner{\psi}{\phi}|^2=1$, and if they are
orthogonal $|\inner{\psi}{\phi}|^2=0$.

In QSL this quantity is the standard statistical overlap, or the complement of the \emph{Kolmogorov distance} 
\begin{equation}
F^2(P,Q)=1-\delta(P,Q) = 1-\frac{1}{2}\sum_{x\in\Omega} |P(x)-Q(x)|,
\end{equation}
where $P$ and $Q$ are distributions describing the states. States with
disjoint support are perfectly distinguishable in the same way as
orthogonal states are perfectly distinguishable, while states with
overlapping support are not.

Another popular measure of similarity is the \emph{fidelity}, which in QSL
is given as the square root of the statistical overlap,
\begin{equation}
F(P,Q)=\sqrt{1-\delta(P,Q)}.
\end{equation}
The quantum fidelity between two mixed states $\rho$ and $\sigma$ is given for general states by 
\begin{equation}
F(\rho,\sigma)=\mathrm{Tr}(\sqrt{\rho^{1/2}\sigma \rho^{1/2}}),
\end{equation} 
where Tr denotes the trace of the operator.
If $\rho$ is pure this simplifies to $\sqrt{\average{\rho}{\sigma}}$, and if
both states are pure to $|\inner{\rho}{\sigma}|$ (see~\parencite{Jozsa1994}). For an
elementary system the fidelity between QSL states is completely equivalent
with that between quantum states. As an example, the fidelity between
$(0,R)$ and $ (0,R') $ is $1$, between $(0,R)$ and $(1,R')$ is $0$, and
between $(0,R)$ and any other state is $1/\sqrt{2}$ (including the
maximally mixed).

Note that the fidelity measure in QSL is not the
Bhattacharyya coefficient
\begin{equation}\label{eq:Bhattacharyya}
B(P,Q)=\sum_{x\in\Omega} \sqrt{P(x)Q(x)}
\end{equation} that is usually taken as the classical analogue to
quantum fidelity~\cite{Fuchs1999}.

The quantum states of a qubit have a geometrical representation called the
\emph{Bloch sphere}. It is a unit ball where all pairs of antipodal points on
the surface correspond to orthogonal pure states, and inside the surface
resides the mixed states.
\begin{figure}
	\centering
	\includegraphics[scale=1]{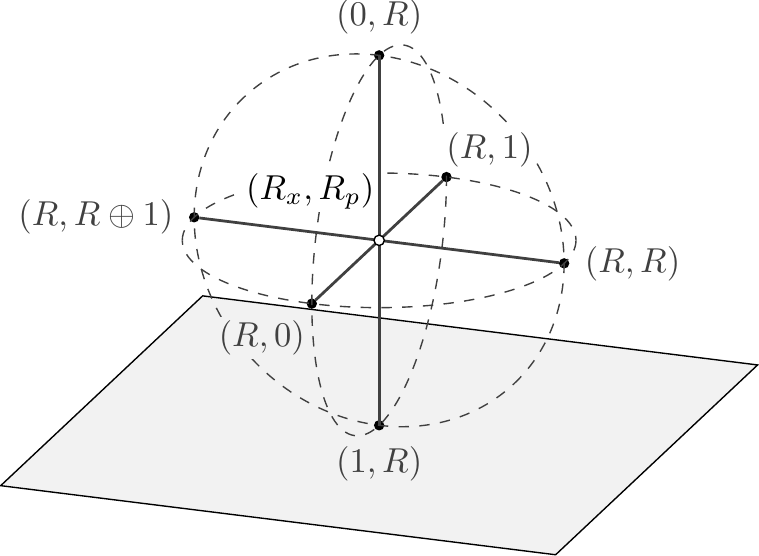}
\captionof{figure}{Geometric representation of the seven states of
an elementary QSL system and their position relative to the Bloch sphere.} 
\label{fig:geometric_rep_QSL}
\end{figure}
The QSL states in relation to the Bloch sphere is shown on
\cref{fig:geometric_rep_QSL}. 
Their positions are motivated by the relations of \cref{eq:QSLXbasis,eq:QSLYbasis,eq:QSLZbasis} to mutually unbiased pairs of orthogonal quantum states. 
The states in each pair have disjoint support, and each state has a statistical overlap $F^2=1/2$ with any state in any other pair; they form mutually unbiased partitions of the state space.

\subsubsection{Transformations}

Unitary transformations are reversible maps on quantum states. It is
therefore natural to simulate these with operations from \emph{classical
reversible logic} on the bits composing a QSL state.

In identifying which transformations that correspond to which quantum gates, we start with the $X$ transformation. 
This operation inverts the states of the computational basis and leaves the phase basis unchanged.
\begin{equation}
X\ket{0} = \ket 1,\quad X\ket{1} = \ket 0,\quad X\ket{\pm} = \ket{\pm}
\end{equation}
Therefore the corresponding QSL gate flips the computational bit and
leaves the phase bit unchanged (see \cref{fig:QSL_Pauli_gates}),
\begin{equation}
\X(x_0,p_0) = (x_0\oplus 1,p_0)
\end{equation}
or as the permutation $(02)(13)$ in the cyclic notation.

In the same way, the $Z$ gate inverts the phase basis and leaves the
computational basis unchanged.
\begin{equation}
Z\ket{0} = \ket 0,\quad Z\ket{1} = \ket 1,\quad Z\ket{\pm} = \ket{\mp}.
\end{equation}
Hence, the QSL Z-gate flips the phase bit instead of the computational bit
\begin{equation}
\Z(x_0,p_0) = (x_0,p_0\oplus 1)
\end{equation}
or $(01)(23)$.

In QSL the states that correspond to eigenstates of $Y$ can be singled out
as those that preserve the parity between the computational and phase bit,
$x_0 \oplus p_0$. The only way a transformation will preserve that while
also inverting the computational and phase basis, is to flip both bits.
\begin{equation}
\Y(x_0,p_0)= (x_0\oplus 1, p_0 \oplus 1)
\end{equation}
This corresponds to the permutation $(03)(12)$. \cref{fig:QSL_Pauli_gates}
shows a graphical representation of these transformations. Blue line
segments represent the computational bit and red represents the phase bit.

As in quantum theory these transformations uphold the identities
\begin{equation}
\X^2=\Y^2=\Z^2=I,
\end{equation}
but composition of these gates also shows a difference between QSL and
quantum theory. The identity 
\begin{equation}
XZ=-iY
\end{equation}
is not upheld, but instead
\begin{equation}
\X\Z=\Y.
\end{equation}

\begin{figure}
	\centering
	\includegraphics[scale=1]{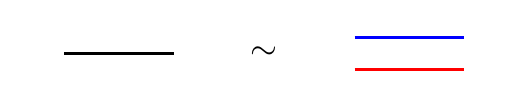}
	\includegraphics[scale=1]{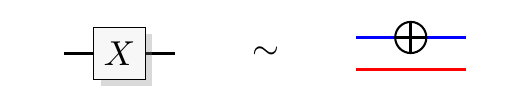}
	\includegraphics[scale=1]{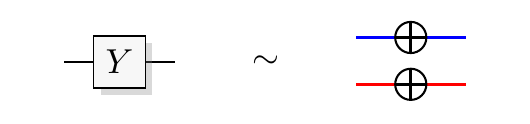}
	\includegraphics[scale=1]{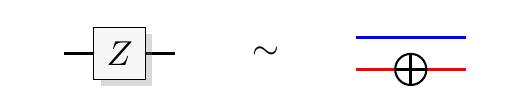}
	\captionof{figure}{Graphical representation of the simulation of the Pauli gates $I,X,Y$ and $Z$. Blue line segments represent the computational bit and red represents the phase bit.} \label{fig:QSL_Pauli_gates}
\end{figure}	

Another useful transformation is the Hadamard gate. 
This quantum gate obeys
\begin{equation}
HZH=X, \quad \text{and} \quad HXH=Z
\end{equation}
while also being an involution, $H^2=I$. In QSL the mapping that simulates
this is the one that swaps the computational and phase bit.
\begin{equation}
\H(x_0,p_0)=(p_0,x_0).
\end{equation}
This permutation, which in cyclic notation is given by $(12)$, is identified as the analogue of the Hadamard
gate in~\cite{Pusey2012}, and indeed,
\begin{equation}
\H^2=\I,\quad \H\Z\H=\X, \quad \text{and} \quad \H\X\H=\Z.
\end{equation}
Note that the analogy is not complete, because in quantum theory we have 
\begin{equation}
HYH=-Y,
\end{equation}
while in QSL
\begin{equation}
\H\Y\H=\Y.
\end{equation}

The last single qubit gate that we describe is the $S$-gate. This is the
square root of the Z-gate,
\begin{equation}
S^2=Z,
\label{eq:S^2}
\end{equation}
which has an order of four. By defining it as 
\begin{equation}
\S(x_0,p_0)=(x_0\oplus 1,p_0\oplus x_0),
\end{equation}
which is equivalent with the permutation $(0213)$, and with
\begin{equation}
\S^{-1}(x_0,p_0)=(x_0\oplus 1,p_0\oplus x_0\oplus 1),
\end{equation} 
we see that \cref{eq:S^2} and the quantum identities
\begin{equation}
SXS^\dagger =Y \quad \text{and} \quad SZS^\dagger=Z
\end{equation} 
are obeyed since
\begin{equation}
\S^2=\Z,\quad \S\X\S^{-1} =\Y, \quad \text{and} \quad \S\Z\S^{-1}=\Z,
\end{equation} 
but again, while we expect 
\begin{equation}
SYS^\dagger=-X
\end{equation}
QSL instead gives
\begin{equation}
\S\Y\S^{-1}=\X.
\end{equation}
\cref{fig:QSL_H_och_S} shows a graphical representation of the simulation
of the Hadamard, $S$, and $S^\dagger$ gates.

\begin{figure}
	\centering
	\includegraphics[scale=1]{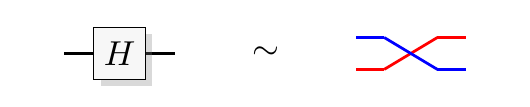}
	\includegraphics[scale=1]{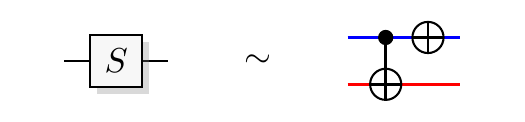}
	\includegraphics[scale=1]{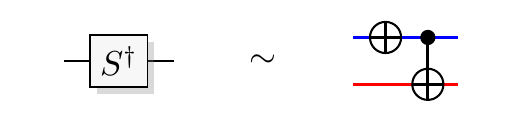}
	\captionof{figure}{Graphical representation for the simulation of Hadamard, $S$, and $S^\dagger$ gates.}
	\label{fig:QSL_H_och_S}
\end{figure}	

All of these transformations are generated by $\H$ and~$\S$ (${\Z=\S}^2$, ${\X=\H\Z\H}$, ${\Y=\X\Z}$, $\S^{-1}={\X\S\X}$), which also generate the group of transformations that Spekkens allows for in his model, the symmetric group on four symbols $S_4$~\cite[68]{Judson2016}. 
This is because $\S$ gives the 4-cycle $(0213)$ and $\H$ gives the 2-cycle $(12)$ of adjacent elements in that 4-cycle.

Another way to view this~\cite{Spekkens2007} is through the Bloch representation. 
\cref{fig:QSL_rotations} shows the effect of \X, \Z, \H, and \S, on the six QSL states.
	\begin{figure}
\centering
\raisebox{3.5cm}{\textbf{A}}\includegraphics[scale=.8]{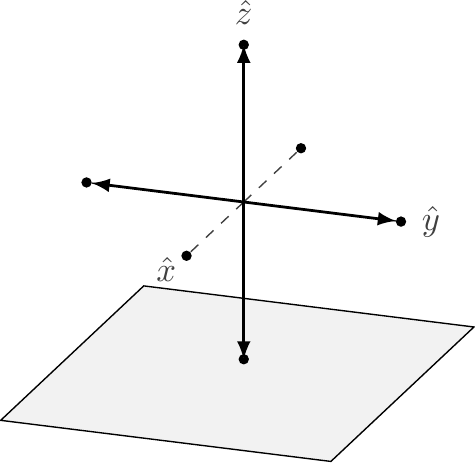}
\raisebox{3.5cm}{\textbf{B}}\includegraphics[scale=.8]{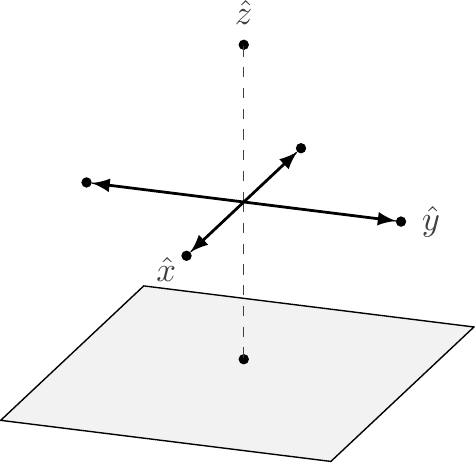}\\[1em]
\raisebox{3.5cm}{\textbf{C}}\includegraphics[scale=.8]{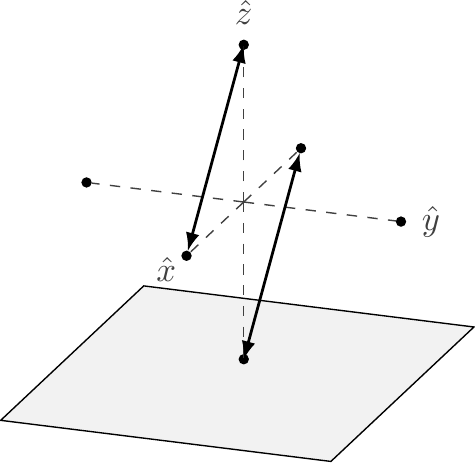}
\raisebox{3.5cm}{\textbf{D}}\includegraphics[scale=.8]{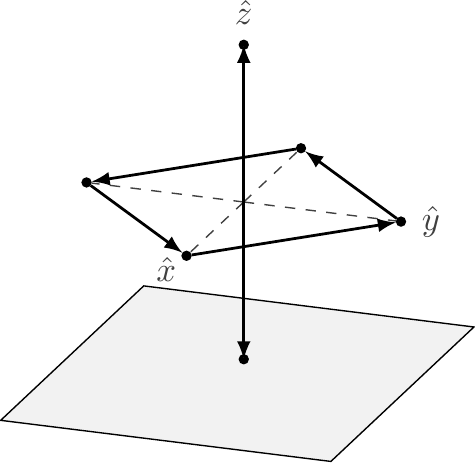}

\captionof{figure}{Effect of the transformations on the six QSL states
in the Bloch representation. \textbf{A-D} shows the \X, \Z, \H, and
\S gates respectively. } \label{fig:QSL_rotations}
\end{figure}
We can see that the \X and \Z transformations have the expected effect,
considering that in the Bloch sphere representation they can be viewed as
$180^\circ$ rotations around the $\hat{x}$ and $\hat{z}$ directions respectively (see
\cref{fig:QSL_rotations}A and \cref{fig:QSL_rotations}B).
However, the expected effect of the Hadamard gate is to also flip the states
along $\hat{y}$, but this is not the case in QSL (see \cref{fig:QSL_rotations}C). 
Finally, the quantum $S$-gate correspond to a $90^\circ$ plane rotation of the Bloch sphere, while the QSL simulation also flips the states along $\hat{z}$ (see \cref{fig:QSL_rotations}C).

Spekkens recognizes this and shows that some transformations, like those we
relate to $H$ and $S$ gates, correspond to antiunitary maps in the Bloch
sphere representation. \textcite{Skotiniotis2008} connected the subgroup of
12 even permutations $A_4$ to unitary maps, while the odd permutations correspond to
antiunitary maps.

It is surprising that these faulty identities, or systematic errors, exist at the level of simulating qubits, but that the simulation can reproduce many phenomena even when the system size grows. There will be many examples of this in the rest of this paper.

\subsubsection{No Universal Spin-1/2 Inverter}

A universal state inverter is a transformation that takes an arbitrary state
$\rho$ and produces an orthogonal state~$\rho^\perp$. Such a device is not
allowed by quantum theory~\cite{Rungta2001}. 
For a pure qubit state the operation inverting the spin is given by a
composition of $X$, $Z$, and complex conjugation
$\ket{\psi^\perp}=XZ\ket{\psi^*}$ \cite{Wootters1998}. Here $\ket{\psi^*}$
represents the complex conjugated state. In the Bloch sphere representation
the spin flip operation is such that
\begin{equation}
\rho=\frac{I+\vec{r}\,\vec{\sigma}}{2} \mapsto \rho^\perp= \frac{I-\vec{r}\,\vec{\sigma}}{2},
\end{equation}
where $\vec{r}$ is the Bloch vector and
$\vec{\sigma}=X\hat{x}+Y\hat{y}+Z\hat{z}$ is the Pauli vector. In general
\begin{equation}
\rho^\perp=XZ\rho^*ZX,
\end{equation}
where $Z$ flips the sign of the $\hat{x}$-component in the Bloch vector, $X$
flips the sign of the $\hat{z}$-component, and conjugation flips the sign of
the $\hat{y}$-component. 
This is an anti-unitary map and cannot be implemented with the dynamics of a quantum system.

In QSL, because of the existence of transformations that correspond to anti-unitary maps in the Bloch sphere representation, one may be tempted to think that a universal spin-1/2 inverter could be possible.
This is not the case, because a device that takes an arbitrary QSL state and produces the disjoint state must obey the following three conditions.
\begin{enumerate}
\item Map the eigenstates of $\X$ to each other. This is done by flipping the phase bit while doing nothing to the computational bit.
\item Map the eigenstates of $\Z$ to each other. This is done by flipping the computational bit while doing nothing to the phase bit.
\item Map the eigenstates of $\Y$ to each other. This requires changing the parity between the computational and phase bit.
\end{enumerate}
The third condition cannot be obeyed without violating the other two. 

\subsubsection{Measurement}

A measurement in QSL consists of information retrieval followed by a state update to ensure that only one bit of information can be known about an elementary system. 
This is completely equivalent to Spekkens' model, and simulates measurement disturbance, or the update of a quantum state after a measurement (sometimes called collapse). 
The specifics are as follows.

A projective measurement of $\Z$ in QSL returns the computational bit
and randomizes the phase bit, while measuring the $\X$ observable instead
reads the phase bit and randomizes the computational bit. Remember that the phase-bit value 0 is associated with a positive relative phase between the computational states, and the phase-bit value 1 is associated with a negative relative phase.
Finally, measuring
$\Y$ returns the parity of the computational and phase bit and then
randomizes both the computational and phase bits while preserving parity,
that is, if measuring $\Y$ yields $x_0\oplus p_0=0$, then the system is updated randomly to either $(0,0)$ or $(1,1)$, and if $x_0\oplus p_0=1$, the system is updated randomly to either $(0,1)$ or $(1,0)$.

This ensures repeatability of measurements, i.e., if a projective measurement is repeated, the second measurement will yield the same outcome as the first. It also ensures measurement disturbance, prohibiting that the bit values not measured can be predicted after the measurement.
The connection with mutual unbiasedness is clear. 
Information retrieved from one of these measurements is unrelated to the information retrieved from the other. 
That is, we can retrieve either the computational bit, the phase bit, or the parity between them, never two of them, or indeed three.

\subsubsection{Preparation}

To prepare a system that simulates $\ket{0}$ or $\ket{1}$, we create a system
with the computational bit $x_0=0$ or $x_0=1$ respectively, and the phase bit
is chosen randomly. In the same way we can prepare QSL states simulating
$\ket{+},\ket{-}, \ket{+i}$ and $\ket{-j}$ by distributing $p_0$ and $x_0$
according to the respective distribution.

Preparation by measurement can also be performed, by taking a maximally mixed state (the state of \emph{no knowledge}, with both computational and the phase bit randomized), measuring some observable to obtain information on the state, and then transforming the now known state into the desired state.

\subsubsection{Non-commutativity of Measurements}

In quantum theory the order of measurements performed is of
importance to their outcomes. Consider for instance that we have prepared a
quantum state $\ket{0}$ and then we measure the observables $X$ and $Z$. If $Z$
is measured first it will produce the outcome 0 with certainty,
and then $X$ will randomly yield the outcome 0 (for $\ket{+}$) or 1 (for $\ket{-}$) with
equal probability. If the order of the measurements is reversed, both measurements of $X$
and $Z$ will instead produce independent random outcomes.

In QSL the outcomes are, for this particular protocol, in complete compliance with quantum theory. 
To simulate the above procedure we start by preparing the QSL state
\begin{equation}
(0,R).
\end{equation}
Measuring $\Z$ will return the computational bit which will always be $0$, and randomize the already random phase bit, so that the state changes from $(0,R)$ to $(0,R')$. Measurement of $\X$ will then return the new random value $R'$ of the phase bit and randomize the computational bit. 
If we instead start by measuring $\X$ the outcome will be the random value $R$ while the computational bit is randomized so that the state changes from $(0,R)$ to $(R'',R)$. A subsequent measurement of $Z$ will retrieve $R''$.

\subsubsection{QKD --- BB84} 

Having built up the model for simulating single qubit systems, let us now consider a cryptographic protocol that only requires a single qubit system.
In 1984, \textcite{Bennett1984} presented a protocol for distributing a cryptographic key, with the property that an eavesdropper can in principle be detected, and then the compromised key can be discarded.

The protocol is shown in \cref{fig:BB84} and proceeds as follows. 
In one round of the protocol Alice wants to send a random bit $b$ to Bob. 
Alice encodes this bit in a one of two bases, \{\ket{0},\ket{1}\} or \{\ket{+},\ket{-}\}. 
The bit-value $b=0$ is encoded in either $\ket{0}$ or $\ket{+}$, and $b=1$ into
$\ket{1}$ or $\ket{-}$. 
Alice chooses the encoding randomly, and Bob chooses to perform a measurement in one of these two bases, also randomly.
If Bob chooses the same encoding as Alice, then he receives the bit $b$, otherwise he retrieves $0$ or $1$ with equal probability.

\begin{figure}
\centering
\includegraphics[scale=1]{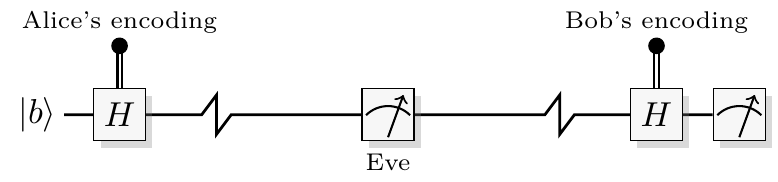}

\captionof{figure}{The \textsc{BB84} protocol with an eavesdropper present.}
\label{fig:BB84}
\end{figure}	

This procedure is now repeated a number of times, and after a sufficiently long sequence they stop and share the encoding they used. 
Both Alice and Bob can now deduce in which rounds of the protocol their encoding were different, and discard data from these rounds. 
This process is known as sifting.
The remaining data of Alice is now equal to that of Bob (in the ideal case); they each have a copy of the same random bit-string.

If someone (usually denoted Eve) is trying to eavesdrop during this process, the measurement disturbance will introduce noise in the key. 
Eve is assumed to have no knowledge about the random encoding used by Alice, so Eve's best strategy is to guess the encoding at random.
If Eve's guess is correct, the quantum state is unaffected by her measurement, 
but if Eve's guess is wrong, she only gets random data and her measurement device outputs a state in the wrong basis.
For example, if Alice sends the state $\ket0$, and Eve picks the phase basis $\{\ket{+},\ket{-}\}$, then Eve's measurement outcome $e$ is random, and her measurement will update the quantum state into
\begin{equation}
\frac{\ket{0}+(-1)^e\ket{1}}{\sqrt{2}}.
\end{equation}
When Bob then measures in the computational basis there is a 50\% chance of
creating a bit-error in the key.

An eavesdropper can now be detected if Alice and Bob sacrifice part of the key and compare it between themselves. 
If the error rate is too high (nonzero in the ideal case), the whole key should be discarded because Eve has been eavesdropping.
If the error rate is sufficiently low (zero in the ideal case), they can be confident that no one have listened in on their communication.
In a real system with noise and lost qubits this is more complicated, and there are also ways to handle limited eavesdropping, but here we will restrict ourselves to the ideal case.

The simulation of the protocol in QSL can be exemplified as follows. 
If Alice chooses the computational basis state, her bit $b$ is encoded into the QSL state
\begin{equation}
(b,R).
\end{equation}
If she instead chooses the phase basis, her bit $b$ is encoded in the QSL state
\begin{equation}
(R,b).
\end{equation}
If Bob chooses to measure in the same basis as Alice has used for encoding, he will retrieve the bit $b$, otherwise he will retrieve the random value $R$, which will then be discarded in the sifting step.

If Eve is eavesdropping, there is a 50\% chance that she chooses a different encoding than Alice, in which case she retrieves the random value $R$, and through the measurement disturbance mechanism, she also randomizes the bit-value $b$. 
In this case Bob will receive a random bit value (independent of $b$), irrespective of in which encoding he measures. 
This reproduces the quantum behavior of BB84.
In other words, the BB84 protocol including the described attack, known as the intercept-resend attack, can be faithfully simulated in QSL.

It is important to note that the security fails if there exists a measurement that retrieves information without disturbing the system. 
In QSL we have included a restriction on the allowed measurements in order to simulate measurement disturbance.
In an actual implementation of QSL, using physical bits, this restriction could be ignored by an adversary, and a measurement that does not disturb the system be performed. 
In quantum theory there is no such measurement~\cite{Busch2009}, but if an actual implementation of BB84 does not use an ideal qubit there may be such a measurement. 
A well-known example is when failing to make a single-photon source so that in each round, the information is encoded onto several photons. 
In this case, an eavesdropper could split off two of the many photons, and measure separately on them, in the two possible encodings.
This would reveal the information present in one of the encodings (Eve will learn which is correct in the sifting phase), while not disturbing the remaining photons at all, so that no noise is present in the generated key. 
This attack, known as the photon-number-splitting attack, would be equivalent to an adversary ignoring the restriction of QSL. 

\subsection{Pairs of Elementary Systems}

In this section we are at most considering pairs of qubits,
\emph{i.e.}, $2\times2$ dimensional systems.

QSL-systems compose under the Cartesian product, and for two elementary systems we write
\begin{equation}
(x_1,p_1)\times(x_0,p_0)=(x_1,p_1)\;(x_0,p_0),
\end{equation}%
the latter notation for brevity, which we could equivalently write
\begin{equation}
(x,p)=(2x_1+x_0,2p_1+p_0)
\end{equation}
so that $x$ and $p$ can take the values $\{0,1,2,3\}$. 
This gives a pair of QSL-systems a sample space of $16$ discrete points. 

States that cannot be described as convex combinations of other states
corresponds to pure states, and the maximally mixed state is still represented
by a uniform distribution over all 16 points.

\subsubsection{Transformations} 

For pairs of elementary systems we also have gates corresponding to two-qubit gates.
The relations in \cref{fig:cnot_identity} motivates the following construction
of a QSL analogue of the quantum CNOT, and one can verify that this construction
is equivalent with the construction in \cite{Spekkens2007}.
\begin{equation}
{\CNOT}(x_1,p_1)\;(x_0,p_0)=(x_1,p_1\oplus p_0)\;(x_0 \oplus
x_1,p_0),
\end{equation} 
see \cref{fig:QSL_2_system_gates}A for a graphical representation.	
\begin{figure}
\centering
\raisebox{2cm}{\textbf{A.}}\includegraphics[scale=1]{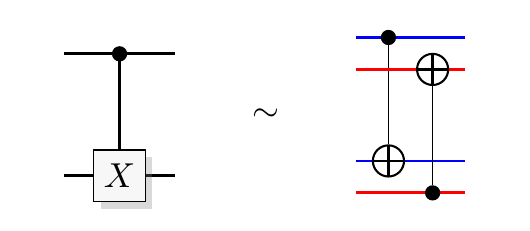}\\\bigskip
\raisebox{2cm}{\textbf{B.}}\includegraphics[scale=1]{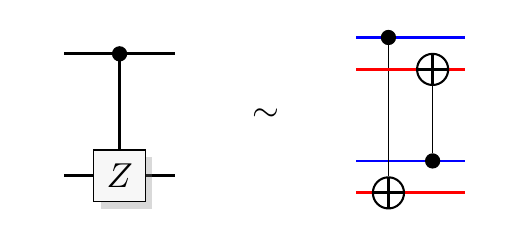}\\\bigskip
\raisebox{2cm}{\textbf{C.}}\includegraphics[scale=1]{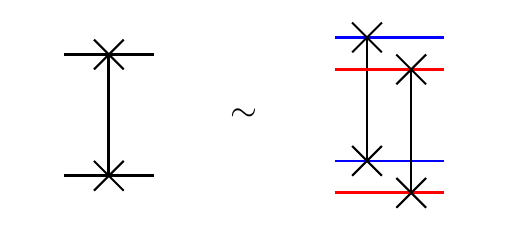}

\captionof{figure}{Graphical representation of the QSL analogue of two qubit
quantum gates. \textbf{A.} \textit{CNOT}, \textbf{B.} controlled-$Z$, and
\textbf{C.} \textit{SWAP}, all constructed to uphold the identities in
\cref{fig:cnot_identity}.} \label{fig:QSL_2_system_gates}
\end{figure}	

Controlled-Z ($CZ$) and \textit{SWAP} follow directly from the identities
\begin{equation}
CZ=(I\otimes H)CNOT(I\otimes H)
\end{equation}
\begin{equation}
SWAP=CNOT(H\otimes H)CNOT(H\otimes H)CNOT
\end{equation}
(see \cref{fig:QSL_2_system_gates}B,C). This implies the QSL maps
\begin{equation}
{\CZ}(x_1,p_1)\;(x_0,p_0)=(x_1,p_1\oplus x_0)\;(x_0,p_0\oplus x_1)
\end{equation}
and
\begin{equation}
{\SWAP}(x_1,p_1)\;(x_0,p_0)=(x_0,p_0)\;(x_1,p_1).
\end{equation}

\subsubsection{Entanglement} 

The QSL analogue to product states are states where the information about one subsystem is independent of the information about the other subsystem. 
In contrast, in entangled states the information is stored in the correlations between the pair of QSL-systems.

In quantum theory, a measure of the amount of entanglement in a bipartite system is the entropy of entanglement.
This is obtained by calculating the von Neumann entropy of the reduced density operator $\rho_A=\mathrm{tr}_B(\rho_{AB})$, where $\rho_{AB}$ is the state of the composite system $A\otimes B$, and $\mathrm{tr}_B$ is the partial trace taken over the subspace related to system $B$. 
As an example, consider the Bell state $1/\sqrt{2}(\ket{01}+\ket{10})$ with
density operator
\begin{equation}\label{eq:Bell_state}
\rho=\frac{\outer{01}{01}+\outer{01}{10}+\outer{10}{01}+\outer{10}{10}}{2}.
\end{equation}
The reduced density operator for system $A$ is the maximally mixed state,
\begin{equation}
\begin{split}
\rho_A=&\frac{\outer{0}{0}\inner{1}{1}+\outer{0}{1}\inner{1}{0}+\outer{1}{0}\inner{0}{1}+\outer{1}{1}\inner{0}{0}}{2}\\
=&\frac{\outer{0}{0}+\outer{1}{1}}{2}= \frac{I}{2}.
\end{split}
\end{equation}

This makes the entropy of entanglement equal to 1, and means that even though we have a pure bipartite state (of entropy 0) we obtain a maximally mixed state (of entropy 1)  for the subsystem $A$. 
The entropy of entanglement is symmetric in its
arguments, so the uncertainty about the state of subsystem $A$ equals that
of subsystem $B$, $S(\rho_A)=S(\rho_B)$. So we have maximal information about the whole system but no information about the individual parts. Thus, in a sense all information is
stored in the correlations of the pair. States with this property are called maximally entangled. The Bell states
\begin{equation}\label{eq:Bell_states}
\ket{\Psi^\pm}=\frac{\ket{00}\pm\ket{11}}{\sqrt{2}} \quad \text{and}
\quad\ket{\Phi^\pm}=\frac{\ket{01}\pm\ket{10}}{\sqrt{2}}, 
\end{equation}
are maximally entangled states. These can be generated from the following recipe (see \cref{fig:Bell_generation}).
\begin{enumerate}
\item Prepare a pair of qubits in the state $\ket{a}\ket{b}$, where $ab$ is 00 for $\ket{\Psi^+}$, 01 for $\ket{\Phi^+}$, 10 for $\ket{\Psi^-}$, and 11 for~$\ket{\Phi^-}$.
\item Apply a Hadamard gate to the first qubit.
\item Apply a \textit{CNOT} with the first qubit as control and the second as
target.
\end{enumerate}
\begin{figure}
	\centering
		\includegraphics[scale=1]{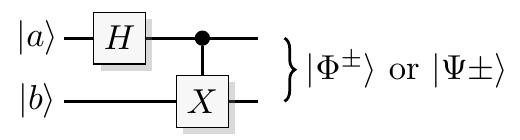}
	
	\captionof{figure}{Circuit for generating the four Bell states.}
	\label{fig:Bell_generation}
\end{figure}	
We can use this recipe also in QSL to produce the states simulating the Bell
states in \cref{eq:Bell_states}. As an example, the state simulating $\Phi^+$
is
\begin{equation}
\begin{split}
(0,R)\;(1,R') \xrightarrow{\H\times \I}\ &(R,0)\;(1,R')\\
\xrightarrow{{\CNOT}}\ &(R,R')\;(\overline R,R').
\end{split}
\end{equation}
For the whole state there are 2 bits of uncertainty, and is therefore a
state of maximal knowledge since it is composed by two elementary
subsystems. Viewing one of the individual subsystems, and ignoring the
other, there is also two bits of uncertainty and therefore, viewed individually, they are
maximally mixed.

The four QSL states related to Bell states are
\begin{equation}
\begin{split}
&\ket{\Psi^+}\sim(R,R')\;(R,R')\\
&\ket{\Psi^-}\sim(R,\overline{ R'})\;(R,R')\\
&\ket{\Phi^+}\sim(R,R')\;(\overline R,R')\\
&\ket{\Phi^-}\sim(R,\overline{ R'})\;(\overline R,R'),
\end{split} 
\label{QSL_Bell_states}
\end{equation}
where $R$ and $R'$ are independent randomly generated bits. These states are
equivalent with those used by Spekkens, but represented differently.
Spekkens represents the state of a pair of elementary systems by their joint
probability distribution. As an example,
\begin{equation}
\begin{split}
&\square \square \square \blacksquare\\[-1.63ex]
&\square \square \blacksquare \square\\[-1.63ex]
&\square \blacksquare \square \square\\[-1.63ex]
&\blacksquare \square \square \square\\
\end{split}
\end{equation}
shows the joint probability distribution corresponding to $\ket{\Psi^+}$.
The marginal probability distributions for each elementary
system are uniform, corresponding to a maximally mixed states; in their marginal
distributions there is no information about the state.

\subsubsection{Remote Steering}	
Suppose that Alice and Bob each possess one qubit of a pair in the state
\begin{equation}
\frac{\ket{00}+\ket{11}}{\sqrt{2}}=\frac{\ket{++}+\ket{--}}{\sqrt{2}}.
\end{equation}
This is sometimes expressed as Alice and Bob sharing an entangled state.
If Bob measures his qubit in the computational basis $\{\ket{0},\ket{1}\}$,
then he obtains the outcome for $\ket{0}$ and $ \ket{1}$ with equal
probability. If he obtains $\ket{0}$, the state updates according to
\begin{equation}
\frac{\ket{00}+\ket{11}}{\sqrt{2}} \to \ket{00}
\end{equation} 
and if he gets $\ket{1}$ it updates according to 
\begin{equation}
\frac{\ket{00}+\ket{11}}{\sqrt{2}} \to \ket{11}.
\end{equation} 
In both cases he knows the state of Alice's qubit, but if Bob instead chooses
to measure in the phase basis, he will still get outcomes for
$\{\ket{+},\ket{-}\}$ at random and with equal probability. If the outcome
is $\ket{+}$ or $\ket{-}$ the state updates to
\begin{equation}
\frac{\ket{00}+\ket{11}}{\sqrt{2}} \to \ket{++}\quad\text{or}\quad		
\frac{\ket{00}+\ket{11}}{\sqrt{2}} \to \ket{--}
\end{equation}
respectively. In either case Bob will learn the state of Alice's system, but
he also notes that it will be different if he chose to measure in the
computational or phase basis. This phenomenon, that Bob apparently can
influence Alice's state by the choice of his measurement basis, is known as
remote steering.

In QSL this is simulated as follows. Alice and Bob each get one elementary
system from a pair initiated in the state
\begin{equation}
\ket{\Psi^+}\sim(R,R')\;(R,R').
\end{equation} 

If Bob measures $\Z$, he learns the value of $R$, and thus also the
value of the computational bit in Alice's system. If $R$ was 1 (there is a
$50\%$ probability of that being the case) the state updates to
\begin{equation}
(1,R'')\;(1,R'),
\end{equation}	 
and accordingly if the outcome is 0, here $R''$ is a new i.i.d random bit generated by the measurement. Note
that this also terminates the last bit of correlation between the two
systems so that there is no correlation between them after the measurement.
If Bob instead measures $\X$, he learns the value of $R'$, and thus
the value of the phase bit in Alice's system. If the outcome is 1 the state
updates to
\begin{equation}
(R'',1)\;(R,1).
\end{equation}	 
and accordingly if the outcome is 0. In QSL, and in Spekkens' model, the
measurement choice of Bob is not influencing Alice's system --- the measurement choice is
influencing Bob's knowledge about Alice's system.

\subsubsection{Anti-correlation in Spin-measurements of the Singlet}

Another closely related example that instead does not work in Spekkens' model or QSL is complete anticorrelation within the singlet. 
If a single system spin measurement is preformed in any direction, $\vec{r}\,\vec{\sigma}$, on both qubits in the singlet state
\begin{equation}
\ket{\Phi^-}=\frac{\ket{01}-\ket{10}}{\sqrt{2}},
\end{equation}
they will always output opposite values. 
The state simulating the singlet is
\begin{equation}
(R,\overline{ R'})\;(\overline R,R'),
\end{equation}
where the overline denotes the Boolean complement.
We can see that measuring $\Z,\X$ and $\Y$ on
both systems will return
\begin{equation}
\begin{split}
R \quad &\text{and} \quad \overline R\\
\overline {R'} \quad &\text{and} \quad R'\\
R \oplus \overline{ R'}\quad &\text{and} \quad \overline R \oplus R'
\label{eq:YY}
\end{split}
\end{equation}
respectively. The first two measurements $\Z$ and $\X$
will produce opposite outcomes, but the last, $\Y$, will produce
equal outcomes (see the last line of expression (\ref{eq:YY})). 
Thus, even though we are
restricted to only three measurements, QSL does \emph{not} reproduce this
phenomenon.

There are three more examples like this; one for each Bell state, with
different relations between the correlations that can be seen from
measurements of $X$, $Y$ and $Z$ on both qubits. In QSL and Spekkens' however, measurement of
$\Y$ will always show opposite correlations as those in quantum
theory, just as in the above example (see further \cite{Spekkens2007}).

\subsubsection{No-cloning}
No-cloning is the no-go theorem stating that there is no unitary
transformation that takes an arbitrary state $\ket{\psi}$ and an auxiliary qubit
$\ket{a}$, and returns both systems in the state
$\ket{\psi}$~\cite{Wootters1982}
\begin{equation}
\ket{\psi}\ket{a} \overset{U}{\nrightarrow} \ket{\psi}\ket{\psi}.
\end{equation}
A unitary transformation preserves the inner product, so if we take two
quantum states $\ket{\psi}$ and $\ket{\psi'}$ and apply a hypothetical cloner,
\begin{equation}
\begin{split}
\ket{\psi}\ket{a} &\rightarrow \ket{\psi}\ket{\psi}\\
\ket{\psi'}\ket{a} &\rightarrow \ket{\psi'}\ket{\psi'}.
\end{split}
\end{equation}
Requiring the inner product to be preserved we have
\begin{equation}
\inner{\psi}{\psi'}\inner{a}{a}=(\inner{\psi}{\psi'})^2.
\end{equation}
Since $\inner{a}{a}=1$ this implies that the norm of these two scalars is equal,
\begin{equation}\label{eq:no_cloning}
|\inner{\psi}{\psi'}|=|\inner{\psi}{\psi'}|^2,
\end{equation}
which only happens if the states are parallel or orthogonal. Therefore, there
is no single unitary that performs the task for an arbitrary state.

In QSL we have two elementary systems in pure states $T,W$, and an auxiliary
system in a pure state $A$. The transformation that clones these states is
\begin{equation}
\begin{split}
T\times A &\mapsto T\times T\\
W\times A &\mapsto W\times W.
\end{split}
\end{equation}

For this to be reversible, the number of elements in the joint support of $ T\times A $ and $ W\times A $ needs to be preserved by the transformation, because if the support grows or shrinks it is not a bijection. Before the transformation, the number of elements in the joint support is 
\begin{equation}
	\begin{split}
		|S(T\times A)\cup S(W\times A)|\\
		=|S(T\times A)|+|S(W\times A)|-|S(T\times A\cap W\times A)|\\
		=|S(T)| |S(A)|+|S(W)| |S(A)|-|S(T\times A\cap W\times A)|\\
		=8-|S(T\times A\cap W\times A)|\\
		=8-|S((T\cap W)\times A)|\\
		=8-|S(T\cap W)||S(A)|\\
		=8-2|S(T\cap W)|,
	\end{split}
\end{equation}
where we have used that the support of a pure state has 2 elements, i.e., $|S(A)|=2$. The number of elements in the joint support after the transformation is
\begin{equation}
	\begin{split}
		|S(T\times T)\cup S(W\times W)|\\
		=|S(T\times T)|+| S(W\times W)|-|S(T\times T\cap W\times W)|\\
		=|S(T)|^2+| S(W)|^2-|S(T\times T\cap W\times W)|\\
		=8-|S(T\times T\cap W\times W)|\\
		=8-|S(T\cap W\times T\cap W)|\\
		=8-|S(T\cap W)||S(T\cap W)|.
	\end{split}
\end{equation}
This gives us the condition (compare with \cref{eq:no_cloning})
\begin{equation}
2|S(T\cap W)|=|S(T\cap W)|^2,
\end{equation}
but this is only fulfilled when $T$ and $W$ are disjoint (orthogonal) or completely overlap (parallel).
Therefore, in a setup with two QSL systems, there is no single QSL transformation that performs the task for an arbitrary state.

\subsubsection{Interference}\label{sec:Interference}

Examples of interference does not require a bipartite system, but it makes for
a good example of how it is used in quantum information processing.

	\begin{figure}
\centering
\includegraphics[scale=1]{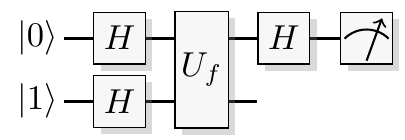}

\captionof{figure}{Circuit for the \textsc{Deutsch} algorithm used to illustrate
interference.} \label{fig:interference}
\end{figure}
Consider the protocol in \cref{fig:interference}. Prepare a two qubit system
in the state $\ket{01}$, and apply a Hadamard gate to each of them
\begin{equation}
\ket{01} \xrightarrow{H\otimes H} \frac{\ket{0}(\ket{0}-\ket{1})+\ket{1}(\ket{0}-\ket{1})}{2}.
\end{equation}
Then a unitary operation that encodes a reversible function in the
computational basis is applied.
\begin{equation}
\frac{\ket{0}\Big(\ket{f(0)}-\big|\overline{f(0)}\big\rangle\Big)+\ket{1}\Big(\ket{f(1)}-\big|\overline{f(1)}\big\rangle\Big)}{2}.
\label{eq:parallelism_second_qubit}
\end{equation}
Even though the unitary encoding the function is used only once,  information
about more than a single function values (in this case two) appears in the state. This phenomenon is
sometimes called quantum parallelism. The state in expression
(\ref{eq:parallelism_second_qubit}) also equates to
\begin{equation}
\frac{(-1)^{f(0)}\ket{0}\ket{-}+(-1)^{f(1)}\ket{1}\ket{-}}{\sqrt{2}},
\end{equation} 
and after the final Hadamard it becomes
\begin{equation}
\frac{(-1)^{f(0)}\ket{+}\ket{-}+(-1)^{f(1)}\ket{-}\ket{-}}{\sqrt{2}}.
\end{equation}
Measuring the observable $\outer{0}{0}\otimes I$, that is, testing whether
the first qubit is still in the initial state $\ket{0}$, will test positive
with probability
\begin{equation}
\left|\frac{(-1)^{f(0)}+(-1)^{f(1)}}{2}\right|^2.
\end{equation}
So, if $f(0)\neq f(1)$ amplitudes will interfere destructively, and the test will be positive with probability zero. If $f(0)=f(1)$ they interfere
constructively, and the test will be positive with unit probability. This is
known as \textsc{Deutsch} algorithm~\cite{Deutsch1985}, and
considered by many the starting point of quantum algorithm research.

There are four functions over 1 bit: $f(x)=0,1,x,$ and $x\oplus 1$. The
canonical way of constructing those with reversible logic follows. The first
two, $f(x)=0$ and $f(x)=1$, are constructed by doing nothing at all respectively
applying a $NOT$ gate to the output. For the last two, $f(x)=x$ and
$f(x)=x\oplus 1$, an additional \textit{CNOT} is applied between the input and
output. The quantum unitary $U_f$ would then be realized in four ways, as the
identity, an $X$ gate on the second qubit only, a \textit{CNOT}, or a
\textit{CNOT} followed by an $X$ gate on the second qubit.

In QSL, the protocol proceeds as follows. Prepare two systems simulating $\ket{01}$, and apply $\H$ to each system
\begin{equation}
(0,R)\;(1,R') \xrightarrow{{\H \times \H}} (R,0)\;(R',1).
\label{eq:81}
\end{equation}
The effect of the above discussed construction of $\U_f$ in QSL will map the
function from the first to the second system over the computational bits. It
will also add the phase of the second system into the first system if the
\textit{CNOT} is present (that is, if $f(0)\oplus f(1)=1$). This simulates the
phenomenon sometimes called phase kickback~\cite{Cleve1998}. For the other two
functions (when $f(0)\oplus f(1)=0$) there will be no phase kickback. So, not
knowing which function that is implemented, the natural way of describing the
effect of the phase kickback is to use the parity of the function. This gives
the map
\begin{equation}\label{eq:QSL_parallelism}
	(x,p)\;(y,z)\mapsto\big(x,p\oplus z(f(0)\oplus
	f(1)\,)\big)\;\big(y\oplus f(x),z\big).
\end{equation}
Applying the function to the result of the map in expression (\ref{eq:81}) we get
\begin{equation}
(R,0)\;(R',1) \xrightarrow{\U_f} \big(R,f(0)\oplus
f(1)\big)\;\big(R'\oplus f(R),1\big),
\end{equation}
and then applying the last transformation ${\H\otimes \I}$ gives
\begin{equation}
\big(f(0)\oplus f(1),R\big)\;\big(R'\oplus f(R),1\big).
\label{eq:84}
\end{equation}
A simulation of measuring the observable $\outer{0}{0}\otimes I$ is to test
whether the computational bit of the first system is zero or not. From
the state in expression (\ref{eq:84}) this translates to testing 
\begin{equation}
f(0)\oplus
f(1)\overset{?}{=}0,
\end{equation} which is true only when $f(0)$ and $f(1)$ are equal,
and false when they are not, just as in quantum theory.

\begin{figure}
	\centering
		\includegraphics[scale=1]{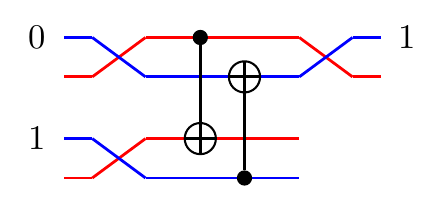}
	
	\captionof{figure}{Example of QSL simulating interference in \textsc{Deutsch}
		algorithm when $f(x)=x$.} \label{fig:interference_with_cnot}
\end{figure}	
\cref{fig:interference_with_cnot} shows the explicit instance of protocol in
QSL when $f(x)=x$.

This is a faithful simulation of a protocol showing interference and quantum parallelism. 
Even in QSL there is an apparent effect of information from both function values being added to the result, while only one query is made (see the QSL map $\mathcal{U}_f$ in \cref{eq:QSL_parallelism} ). 
Now,
whether a quantum computer access all function values with only one query, or
whether the equations only describe the structure of additional information
that can be retrieved, is an unsettled philosophical debate. For QSL the
latter is certainly the case, but we should also stress that QSL is \emph{not} quantum theory. 
Yet, since one process faithfully simulates the other, these processes are
operationally equivalent, and in that spirit a machinery running this particular QSL
recipe should not be distinguished from one running the quantum recipe.

We need to stress that in QSL we do \emph{not} calculate the parity of the
function $f(0)\oplus f(1)$ by accessing the function twice, and then have the oracle
signal that information. The expression using the parity of the function just turns out to be a good
description for the additional information available when we do
not know which of the four functions that was applied. With the constructions of
$f(x)=x$ and $f(x)=x \oplus 1$, the phase bit of the first system is
flipped, but not for the constructions for $f(x)=0$ and $f(x)=1$, i.e., the
output depends on the choice of function made in the construction.

To further clarify,  with ``operationally equivalent'' we mean that processes, or
theories are equivalent when only judged by their input/output behavior. So, a
machine running \textsc{Deutsch} algorithm in QSL is operationally
equivalent to the corresponding process on a quantum machine, but comparing
QSL as a theory against quantum theory is not. In fact, we have already seen
examples that QSL is operationally different, for instance in the
correlations seen in the outcomes from Pauli measurement over Bell states. However,
our main goal here is not to fuel the philosophical debate, nor to perfectly
simulate the whole of quantum theory, but only to simulate it accurate enough to solve the computational problem.

\subsubsection{Measurements}

For one elementary system we have defined three measurements that we relate to
the Pauli observables. We also have six observables for the orthogonal
projections onto their eigenstates, as exemplified in the previous section, and we
can ask whether the computational bit is $0$ or not (corresponding to the observable $\outer{0}{0}$), rather than asking if it
is $0$ or $1$ (Corresponding to the observable $Z$). When measuring a single elementary system there is no distinction
between asking if it is $0$ or not, or $0$ or $1$, but for higher dimensional systems there is a distinction. 
For instance, say that the dimension is four and that we measure the observable $ \outer{0}{0} $, with the outcome that the system is not in $\ket{0}$, we cannot infer
whether it is $\ket{1}$, $\ket{2}$, or $\ket{3}$. In QSL a measurement of $ \outer{0}{0} $ corresponds to asking the system if all the computational bits are zero or not. A measurement that distinguishes between $\ket{0}$, $\ket{1}$, $\ket{2}$, or $\ket{3}$ corresponds to asking the system whether the computational bits encode $0$, $1$, $2$, or $3$ respectively.

For systems composed of two elementary systems, we have seen that single
system measurements act locally. In the Pauli group of observables there
are also joint measurements of the kind $\sigma_i \otimes \sigma_j$, and in
QSL these are simulated as follows.

A measurement of the ${\Z \times \Z}$ observable returns the correlation between the computational bits of both systems, $x_1 \oplus x_0$, and redistributes according to the outcome. 
That is, if $x_1 \oplus x_0=0$ the computational bits of both systems will be random but completely correlated after the measurement, and if $x_1 \oplus x_0=1$ they will be random but anti-correlated. 
A measurement represented by ${\X\times \X}$ does the same, but for the phase bits. 
The procedure is similar for ${\Z\times \X}$, for the computational bit of the first system and with the phase bit of the second, and so on.

These joint measurements retrieve 1 bit of information even though the whole system is composed of two elementary systems, and we should therefore be able to retrieve two bits, implying that the joint measurements are non-maximally informative measurements. 
In quantum theory this relates to when an
observable has degenerate eigenvalues. For instance, the observable quantity $X\otimes X$ can take on two different values, and knowing this value gives us 1 bit of information.

A Bell state measurement is a measurement that distinguish the four Bell states. This
is a maximally informative measurement since it takes the maximally mixed
state into a pure state (one of the Bell states). In QSL, an analogue to this measurement returns the correlation between the two computational bits \emph{and} between the
two phase bits. An example of this will be given in \cref{sec:Superdense}.

This is not to be confused with measurements whose outcome violate Bell
inequalities --- measurements included in a Bell test. For instance the
CHSH inequality has a classical bound of 2, and quantum theory can violate
this to a maximal value of $2\sqrt{2}$, called the Tsirelson bound.
An experiment that maximally violates the inequality can be constructed from
measuring the observables
\begin{equation}
X \text{ or } Z
\label{eq:85}
\end{equation}
on one qubit, and
\begin{equation}
\frac{Z-X}{\sqrt{2}}\text{ or }\frac{-Z-X}{\sqrt{2}}
\label{eq:86}
\end{equation}
on the other, starting with a system in the singlet state. For a complete
description see \textcite[111]{Nielsen2010}. In Hilbert space the
observables in expression (\ref{eq:86}) has eigenbases that are offset with
$22.5^\circ$ from the eigenbases of those in expression (\ref{eq:85}). In QSL we do not
have measurements with this relationship between each other.

So, in QSL we have maximally entangled states, but cannot violate Bell inequalities. 
Obtaining a violation from QSL, which is a local realist model, would contradict Bell's theorem~\cite{Bell1964}.

Another phenomenon that QSL cannot reproduce is the contextual correlations in
the Peres-Mermin square, which is an explicit example of the Kochen-Spekker
theorem~\cite{Kochen1967}. This was shown by~\textcite{Pusey2012}, not in
the QSL representation, but with his stabilizer representation of Spekkens'
model. However, \textcite{Kleinmann2011} initiated work on extending
Spekkens' model to reproduce the contextual correlations that can be seen
from Pauli group measurements. Their work was later concluded by
\textcite{Harrysson2016}. It is also interesting to note that correlations
seen from measurement sequences in the Peres-Mermin square have an efficient
simulation according to the Gottesman-Knill
theorem~\cite{Gottesman1998,Aaronson2004}.

\subsubsection{Superdense Coding}\label{sec:Superdense}

Superdense coding is the name of a protocol that in a sense allows one party,
Alice, to convey two bits of information $m_1,m_0$ to the other party, Bob, by
only interacting with one qubit. It is however a two qubit protocol as shown
by \cref{fig:superdense_coding}.
	\begin{figure}
\centering
\includegraphics[scale=1]{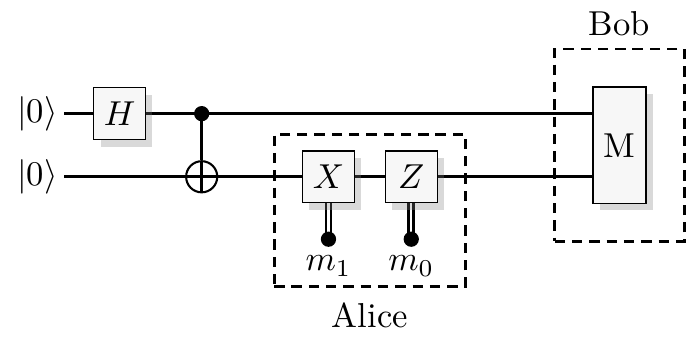}

\captionof{figure}{Protocol for superdense coding. Alice can convey two
bits of information $m_1$ and $m_0$ to Bob by only interacting with one qubit of
a correlated pair.}
\label{fig:superdense_coding}
\end{figure}

First create the state $\ket{\Psi^+}$, give one of the two qubits to Alice and the other qubit to Bob. 
Depending on which of the four messages that Alice wants to send to Bob, she applies $I$, $X$, $Z$, or $Y$ to her qubit, creating
\begin{equation}
\begin{split}
I\otimes I\tfrac1{\sqrt{2}}({\ket{00}+\ket{11}})&=\tfrac1{\sqrt{2}}({\ket{00}+\ket{11}})\\
X\otimes I\tfrac1{\sqrt{2}}({\ket{00}+\ket{11}})&=\tfrac1{\sqrt{2}}({\ket{10}+\ket{01}})\\
Z\otimes I\tfrac1{\sqrt{2}}({\ket{00}+\ket{11}})&=\tfrac1{\sqrt{2}}({\ket{00}-\ket{11}})\\
Y\otimes I\tfrac1{\sqrt{2}}({\ket{00}+\ket{11}})&=\tfrac1{\sqrt{2}}({\ket{10}-\ket{01}}).
\end{split}
\end{equation}
She then sends her qubit to Bob. A Bell state measurement will allow Bob to
perfectly distinguish the state of the pair and deduce Alice's message.

In QSL the state simulating $\ket{\Psi^+}$ is $(R,R')\;(R,R')$, and when Alice applies ${\I, \X,\Z}$ or $\Y$ the result is
\begin{equation}
\begin{split}
\I\times\I\;(R,R')\;(R,R')&=(R,R')\;(R,R'),\\
\X\times\I\;(R,R')\;(R,R')&=(R,\overline{R'})\;(R,R'),\\
\Z\times\I\;(R,R')\;(R,R')&=(\overline{R},R')\;(R,R'),\text{ or}\\
\Y\times\I\;(R,R')\;(R,R')&=(\overline{R},\overline{R'})\;(R,R')
\end{split}
\end{equation}
respectively. 
As previously stated, a Bell state measurement returns the XOR (the correlation) between the two computational bits and the two phase bits. 
This measurement on the above states will yield $00$, $01$, $10$, or $11$, respectively. Thus, we also have superdense coding in QSL.

The reason that Alice can convey a 2 bit message to Bob by sending him one elementary QSL system, is that each QSL system contains 2 bits. 
It is true that Bob cannot access both bit values, since he is restricted to only retrieving 1 bit of information from each system. 
He can retrieve both bits of the message only because there is a second system, initially highly correlated with the first, and where the correlations are manipulated by Alice.

\subsection{Higher Number of Elementary Systems}

Going to higher number of elementary systems, QSL departs from Spekkens' toy
model to become a less restrictive theory. Systems still compose under the
Cartesian product, and the following notation is used
\begin{equation}
\begin{split}
(x_{n-1},p_{n-1})&\;\ldots\;(x_1,p_1)\;(x_0,p_0)\\
&=(\sum 2^ix_i,\sum 2^ip_i)
=(x,p).
\end{split}
\end{equation}
Sometimes we will use a separation between different registers with the
notation $(x,p)\;(x',p')$, where $x$, $p$, $x'$ and $p'$ will be integers modulo a power-of-two.

\subsubsection{Teleportation}

Before we introduce the new transformations, let us take one more example with
more than two qubits.

Teleportation is the name of a protocol where, given some pre-shared entanglement, Bob can recreate a qubit state $\ket{\psi}$ made available to Alice, from information from Alice on the result of a Bell-state measurement performed by her.
\begin{figure}
	\centering
	\includegraphics[scale=1]{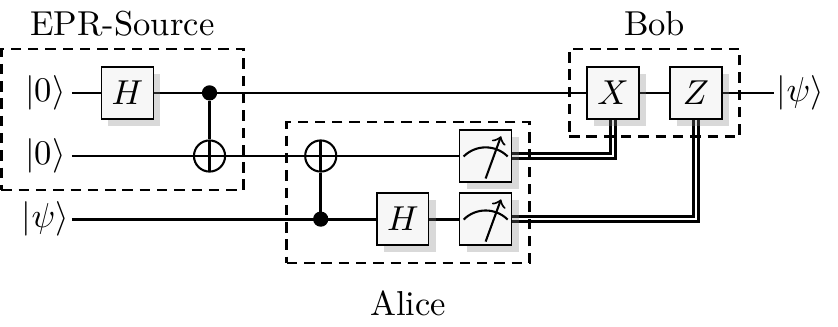}
	\captionof{figure}{The teleportation protocol.}
	\label{fig:teleportation}
\end{figure}
The protocol is shown in \cref{fig:teleportation}, and proceeds as follows.
Alice and Bob share a qubit pair in the Bell state $\ket{\Psi^+}$ created by the  \textit{H} and a \textit{CNOT} gates at the left of \cref{fig:teleportation},
\begin{equation}
\ket{00}\ket{\psi}\to\frac{\ket{00}+\ket{11}}{\sqrt{2}}\ket{\psi}.
\end{equation}
Alice then correlates the state $\ket{\psi}=a\ket{0}+b\ket{1}$, that she wants
to teleport to Bob, with her part of the pair. This is done using a Bell state measurement built from the next \textit{CNOT} and  \textit{H} gates.
\begin{equation}
\begin{split}
&\frac{a(\ket{00}+\ket{11})\ket{0}+b(\ket{00}+\ket{11})\ket{1}}{\sqrt{2}}\\\to
&\frac{a(\ket{00}+\ket{11})\ket{0}+b(\ket{01}+\ket{10})\ket{1}}{\sqrt{2}}.
\end{split}
\end{equation}
After the Hadamard, the state can be written as
\begin{equation}
\begin{split}
&\frac{(a\ket{0}+b\ket{1})\ket{00}+(a\ket{0}-b\ket{1})\ket{01}}{\sqrt{2}}\\
+&\frac{(a\ket{1}+b\ket{0})\ket{10}+(a\ket{1}-b\ket{0})\ket{11}}{\sqrt{2}}.
\end{split}
\label{eq:teleport}
\end{equation}
Alice measures her two qubits and sends the result to Bob. We see from
expression (\ref{eq:teleport}) that if the least significant bit that Bob
receives is set, then he needs to apply $Z$ to his qubit, and if the most
significant is set apply an $X$. Doing so, he retrieves the state
$a\ket{0}+b\ket{1}=\ket{\psi}$.

The same protocol (\cref{fig:teleportation}) with QSL gives
\begin{equation}
\begin{split}
&(0,R)\;(0,R')\;(b_x,b_p)\\
\xrightarrow{{\H\times \I \times \I}}\ &(R,0)\;(0,R')\;(b_x,b_p)\\
\xrightarrow{{\CNOT\times \I}}\ &(R,R')\;(R,R')\;(b_x,b_p)\\
\xrightarrow{{\I\times \CNOT}}\ &(R,R')\;(R\oplus b_x,R')\;(b_x,b_p\oplus R')\\
\xrightarrow{{\I\times \I\times \H}} \ &(R,R')\;(R\oplus b_x,R')\;(b_p\oplus R',b_x).
\end{split}
\end{equation}
By the two measurements Alice now retrieves the values of $R\oplus b_x$ and
$R'\oplus b_p$, and sends them to Bob. 
If the first bit is set he performs an $\X$, i.e., adds $R\oplus b_x$ to the computational bit modulo 2.
If the latter is set he performs a $\Z$, i.e., adds $R'\oplus b_p$ to the phase bit, also modulo 2. 
The state he ends up with is
\begin{equation}
(R\oplus (R\oplus b_x),R'\oplus (R'\oplus b_p))=(b_x,b_p),
\end{equation}
which is the state Alice was provided with.
Thus, Alice and Bob cooperate to ``teleport'' the state of Alice's input system to Bob. 
Just as in the quantum protocol, Alice does not at any point retrieve any information on the state provided to her. 
That is, Alice and Bob cooperate to update the state of Bob's system to correspond to the state of Alice's input system, independent of what the state of that input system is, just as in the quantum protocol.

In classical communication this is not described as teleportation, but rather as encryption and decryption through the One-Time-Pad, where the random numbers $R$ and $R'$ act as secret shared key between Alice and Bob. 
Therefore, in QSL (or Spekkens' toy model), the teleportation protocol is equivalent to (two uses of) the One-Time-Pad. 

We should note that quantum teleportation goes beyond this, since we are restricted to simulate a finite subset of quantum states being teleported, rather than the continuum of states in quantum theory.

\subsubsection{Transformations}

Here we introduce a transformation unavailable in Spekkens' model, and as we
will see it take us out to a less restricted model.
The QSL Toffoli construction is shown in \cref{fig:QSL_3_qubit_gates}A., and
produces the mapping
\begin{equation}
\begin{split}
&\big(x_2,p_2\big)\\&\big(x_1,p_1\big)\\&\big(x_0,p_0\big)
\end{split}\quad\mapsto\quad
\begin{split}
&\big(x_2,p_2\oplus p_0x_1\big)\\
&\big(x_1,p_1 \oplus p_0 x_2\big)\\
&\big(x_0 \oplus x_2 x_1,p_0\big).
\end{split}
\end{equation}
It is constructed in this way to uphold quantum gate identities like those
in \cref{fig:Toffoli_identities}A. and \cref{fig:Toffoli_identities}B. If the
input of one of the two control qubits is initiated in $\ket{1}$, the effect
is that of a \textit{CNOT} over two other systems. Also, if the target qubit is
initiated in $\ket{-}$, the effect is that of a $CZ$ over the two control
qubits.

There is another identity using two \textit{CNOT}s to produce a controlled-SWAP, called Fredkin gate (see \cref{fig:Toffoli_identities}C). The QSL
analogue of this is shown in \cref{fig:QSL_3_qubit_gates}B. and produces the
mapping
\begin{equation}
\begin{split}
&\big(x_2,p_2\big)\\&\big(x_1,p_1\big)\\&\big(x_0,p_0\big)
\end{split}\mapsto
\begin{split}
&\big(x_2,p_2\oplus (x_1\oplus x_0)(p_1\oplus p_0)\big)\\
&\big(x_1\oplus x_2(x_1\oplus x_0),p_1\oplus x_2(p_1\oplus p_0)\big)\\
&\big(x_0\oplus x_2(x_1\oplus x_0),p_0\oplus x_2(p_1\oplus p_0)\big).
\end{split}
\end{equation}

\subsection{Properties and Relations to other Theories} 
Simulated quantum phenomena in QSL have an efficient simulation on a classical
probabilistic Turing machine. 
It uses two classical bits for each elementary
system, and all QSL-gates are constructed from a constant number of classical
reversible gates.  
We therefore have the following simple lemma.

\begin{lemma}\label{lemma:efficiency}
Any quantum circuits, constructed from gates also present in QSL, have a
classical simulation that requires at most a constant overhead in resources.
It can in turn be simulated in polynomial time, in the size of the circuit, on a
classical probabilistic Turing machine.
\end{lemma}

\begin{figure}
\centering
\raisebox{2.4cm}{\textbf{A.}}\includegraphics[scale=1]{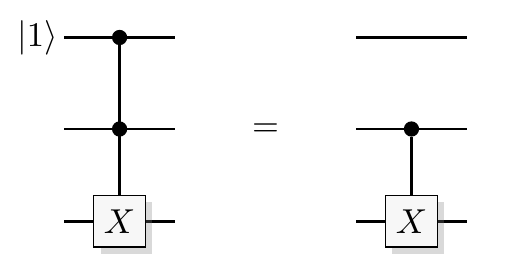}\\\bigskip
\raisebox{2.4cm}{\textbf{B.}}\includegraphics[scale=1]{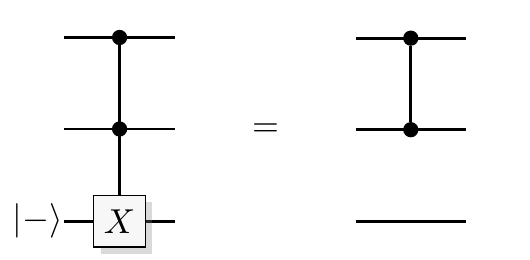}\\\bigskip
\raisebox{2.4cm}{\textbf{C.}}\includegraphics[scale=1]{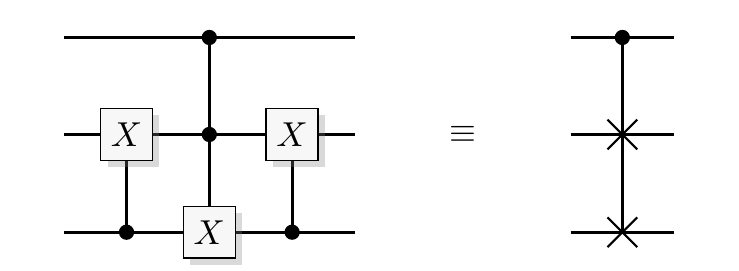}

\captionof{figure}{Quantum gate identities. \textbf{A.} Toffoli gate with one
of the controlling qubits initiated in $\ket{1}$ results in a \textit{CNOT}
over the other two qubits. \textbf{B.} Toffoli gate with the target qubit
initiated in $\ket{-}$ results in a $CZ$ over the cotrol qubits. \textbf{C.}
Identity connecting the Toffoli and Fredkin gate.}
\label{fig:Toffoli_identities}
\end{figure}	
\begin{figure}
\centering
\raisebox{3cm}{\textbf{A.}}\includegraphics[scale=1]{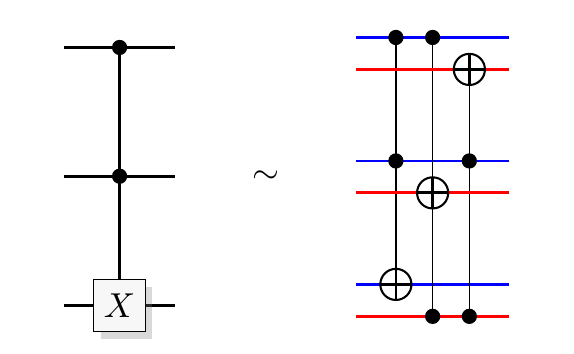}\\\bigskip
\raisebox{3cm}{\textbf{B.}}\includegraphics[scale=1]{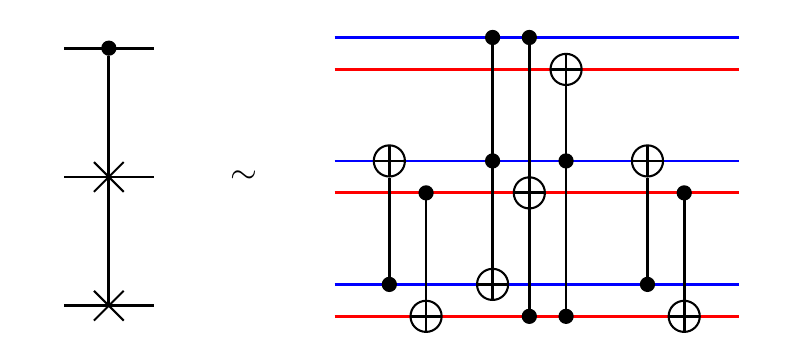}

\captionof{figure}{QSL construcion of the Toffoli gate (\textbf{A.}) and
Fredkin (\textbf{B.}), constructed to uphold the identities in
\cref{fig:Toffoli_identities}.}
\label{fig:QSL_3_qubit_gates}
\end{figure}

\subsubsection{The relation to stabilizer quantum mechanics, locality and contextuality}

Another theory that is computationally tractable is the \emph{stabilizer subtheory}. 
As the name suggests it is a subtheory of quantum theory where one is restricted to  transformations from the Clifford group (generated by the Hadamard, phase-gate, and CNOT), and Pauli group measurements.  Being restricted to these operations, one can only reach the stabilizer states. Classical controls are also allowed. This subtheory has an efficient classical simulation and this is shown by the Gottesman-Knill theorem~\cite{Gottesman1998}. 
An in-depth discussion of its resource requirements can be found in \textcite{Aaronson2004}.

It turns out that Spekkens' model is closely related to the stabilizer subtheory~\cite{Spekkens2016}. 
\textcite{Pusey2012} showed that the number of states in Spekkens' model is the same as the stabilizer states, and it has operations that are similar to the generators of the Clifford group.
However, it is not a restricted version of quantum theory  --- as we saw in \cref{sec:QSL_Elementary_Systems} there are transformations that
correspond to anti-unitary transformations.

QSL contains Spekkens' model and is therefore not a restricted version of quantum theory. It is per construction completely
local, since all information propagates through local interactions, while the
stabilizer subtheory is not. 
At this point, it is important to note that both QSL and Spekkens' model are both also non-contextual, in contrast to stabilizer QM.
Consider an observable $A$ that is jointly measurable with observables $B$ and $C$, where $B$ and $C$ might not be jointly measurable.
Then, we can measure $A$ together with $B$, or $A$ together with $C$.
In quantum theory, $A$ would commute with both $B$ and $C$.
Even so, any attempt to assign values to the outcomes from measuring $A$, $B$, and $C$ can force the value of $A$ to depend on the \emph{context} of it being measured together with $B$ or with $C$, resulting in a contextual model \cite{Kochen1967}. 
\begin{definition}
	A non-contextual model is a model where measurement outcomes do not depend on
	the context of the measurement.
\end{definition}

Using the above definition, we arrive at the following theorem.

\begin{theorem}\label{thm:non-contextual}
	QSL is a non-contextual model.
\end{theorem}
\begin{proof}
	Per construction, QSL has a simultaneous value assignment to all observable quantities, and these values do not change dependent on the measurement or measurement context that we choose to use to retrieve them. 
  Thus, measurements outcomes in QSL 	does not depend on the context of the measurement.
\end{proof}

\subsubsection{QSL extends the state space of Spekkens' model}

QSL also allows for a strictly larger set of states
than Spekkens' model and the stabilizer subtheory. 
To see why, consider the construction in
\cref{fig:GHZ_with_Toffoli} that produces a simulation of the GHZ-state
\begin{equation}
\frac{\ket{000}+\ket{111}}{\sqrt{2}}.
\end{equation}
If the least significant system is measured in the computational basis, the
value of $x$ is retrieved. If the second system is measured in the phase basis,
the value of $y\oplus x  z$ is retrieved. This information is enough to
completely specify the state of the most significant system, i.e., we know the
state of both the computational and phase bit, resulting in zero uncertainty. 

\begin{figure}
\centering
\includegraphics[scale=.8]{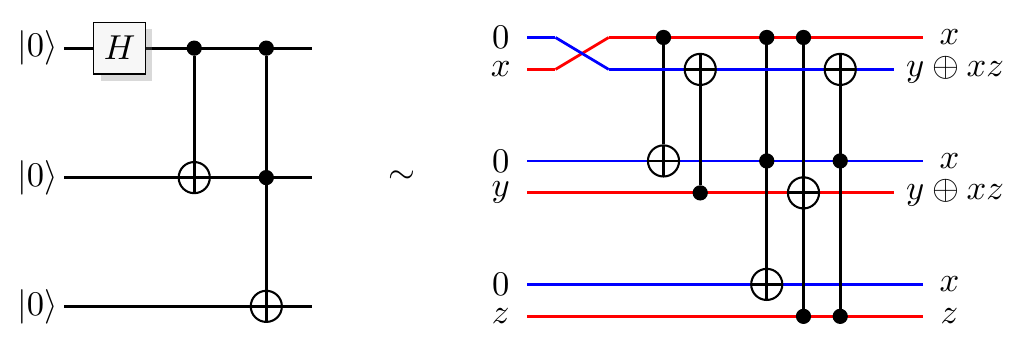}

\captionof{figure}{One construction that produces a GHZ-state, and the
matching QSL construction.}
\label{fig:GHZ_with_Toffoli}
\end{figure}	

We can repeat this scheme retrieving the basic bit values of a fourth system while reusing the two least significant systems as auxiliaries. By induction, with $n$ repetitions we can learn the basic bit values of $n$ QSL-systems, using only two auxiliary systems.   
Using the available reversible transformations, or permutations, we can create any value of the basic bit values we desire. 
The output of this procedure is not a valid (epistemic) state in Spekkens' model, and this has consequences for the set of pure states in QSL, because pure states are now states where both bit values are known. 

This also enables a larger mixed-state space.
It is true that Spekkens' original model \cite{Spekkens2007} does not allow for general mixtures of states since knowledge is defined in terms of partitions of the phase space corresponding to the support of the probability distributions used, and not the distributions themselves. 
For example, the only single-system mixed state allowed in Ref~\cite{Spekkens2007} is the completely mixed state, see \cref{eq:spekkensstates}.
However, it is simple to create any classical mixture of the six available pure states, corresponding to the octahedron in \Cref{fig:state_space}.
For example, to create a classical mixture of two pure states with probability $p$ for one of them, all that is needed is a number of auxiliary systems that is linear in the number of bits of $p$, all in the completely mixed state.
Measurement of each system gives a fair coin toss, and simple binary classical comparison with the number $p$ gives a binary output which is 1 with probability $p$. 
This can then be used to choose what state to prepare, resulting in the desired classical mixture of the two states.
The addition in QSL of the QSL-Toffoli enables mixtures within the larger tetrahedron, but note that the states outside the octahedron does not correspond to quantum states, at least not in a simple manner. 

Compare this to how with stabilizer
quantum theory we start with quantum theory and restrict it to a finite set of
states, transformations, and measurement. Adding the quantum Toffoli to
stabilizer quantum theory we asymptotically recover quantum theory. For one and
two elementary systems, QSL and Spekkens' theory are equivalent, and can be
derived by imposing restrictions on a classical statistical theory over the
phase space described by $\mathbb{Z}_2^{2n}$  (see \cite{Spekkens2016}). Adding
the QSL-Toffoli to Spekkens' model we can reach the ontic states and any
mixtures of these. Thus, asymptotically we recover the unrestricted classical
theory over $\mathbb{Z}_2^{2n}$.

\begin{figure}
\centering
\includegraphics[scale=.95]{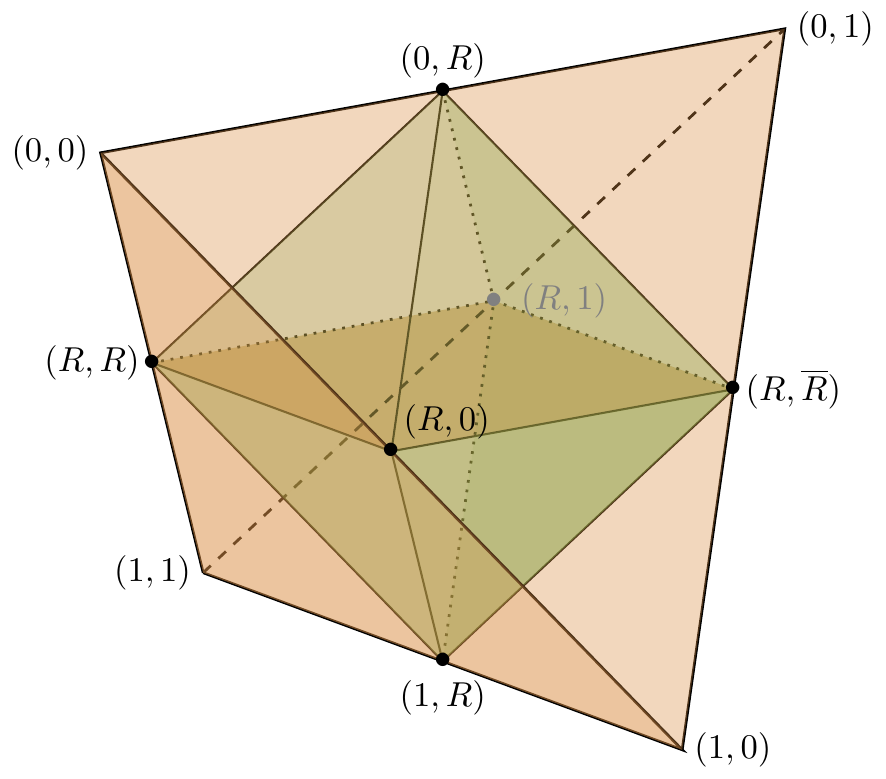}

\captionof{figure}{Representation of the state space of a single elementary system, with a geometry of a tetrahedron inscribing an octahedron. To aid with the visualization: if we take the vertices of the tetrahedron and fold them to one of the closest vertices on the octahedron, we recover the octahedron.}
\label{fig:state_space}
\end{figure}	

\subsubsection{QSL is an example of a Generalized Probability Theory}

Common for all these theories, quantum theory, the stabilizer subtheory, Spekkens' model, and QSL is that they are \emph{generalized probabilistic theories} (GPTs). 
That is, if we define a probabilistic theory $D$ aiming to describe the effects of measurement outcomes from these theories, it is necessarily non-Kolmogorovian. 
To see why let us consider $X=1$ and $Z=1$ as events in $D$. 
The observables do not commute in any of the four theories and are not simultaneously measurable. 
This means that $\{X=1\}\cap\{Z=1\}$ is not well-defined, and therefore, $D$ does not form a $\sigma$-algebra as required by Kolmogorov's third axiom. 
For another and more comprehensive account of this see~\textcite{Kleinmann2014}.

A $\sigma$-algebra is a collection of sets closed under the operations of complement, and countable intersection and union~\cite{Kallenberg1997}. 
If these sets --- connected to events --- are measurable, they can be used to form a logic. 
With the complement, intersection, and union as the negation, conjunction, and disjunction respectively (NOT, AND, and OR), we recover a Boolean algebra~\cite[297]{Judson2016}; the logic underpinning classical probability theory. 
However, since measurable quantities in the four theories do not form a $\sigma$-algebra --- some quantities are not simultaneously measurable --- they cannot form a Boolean algebra.
Systems that are described by these theories do not obey classical logic.

A hands-on example of non-classical logic, closely related to Spekkens' model and QSL, is owing to \textcite[21]{Cohen1989}. 
He considers a firefly in a box which is either lit up or not.
The box can be viewed from the front side, and from one of the adjacent sides.
On the front side of the box there are two windows, one to the right and one to the left. 
Also, to the side of the box there are two windows, one to the right and one to the left. 
Now, if the firefly is lit up, a single observer (without any depth perception) looking into either the front or side windows will only be able to tell whether the firefly was on the left or right side, or close to or far from the front side of the box. 
We can think of the box as divided into four partitions and a single observer can only distinguish between two of them at a time. 
Cohen then writes
\begin{quote}
``If we believe that [this] is the best possible characterization of the firefly system, then we believe our firefly in a box is a nonclassical physical system, because [looking at both front and side windows] cannot be performed simultaneously. 
If we believe that we can [look into both front and side windows] simultaneously, perhaps by positioning two observers, one at each window, [...] then we believe our system is classical and [that it has a better characterization].'' --- \textcite[25]{Cohen1989}
\end{quote} 

The connections to Spekkens' model and QSL is clear. 
The four partitions of the box relate to Spekkens' four ontic states, and we can relate the two windows to the two classical bits, the computational and phase bit, of a QSL-system. 
Note that if the firefly does not light up, the observer has no idea as to where it is, and the system is in the maximally mixed state. 
Cohen shows that the system is classical by finding a refined description, just like Spekkens' model and QSL is constructed in a way permitting a finer description, but without that finer description, all three operate per definition according to a non-classical logic.

To summarize, Spekkens' and QSL are not restricted quantum theories, like the
stabilizer subtheory, they are restricted \emph{classical} theories. They are
generalized probabilistic theories and systems behave according to a
non-classical logic, because of the restriction that turn them into generalized probabilistic theories. 

We will now switch our attention from protocols describing quantum phenomena,
to protocols for solving computational problems. First out is the
\textsc{Bernstein-Vazirani} problem, which builds on an important primitive
in quantum computation known as \textsc{Fourier sampling}.

\section{The \textsc{Bernstein-Vazirani} Problem}\label{sec:Bernstein-Vazirani}

In 1993 Ethan Bernstein and Umesh Vazirani~\cite{Bernstein1993,Bernstein1997}
devised the first decision problem that showed a significant oracle separation
between the quantum and classical (probabilistic) computational models. In
other words, they provided an oracle algorithm   separating the computational
complexity class \textbf{BPP} from the matching class for quantum machines ---
\textbf{BQP}. 
\textbf{BPP} contains all decision problems that can be solved in polynomial time with an error probability bounded away from $1/2$ on a probabilistic Turing machine, while \textbf{BQP} contains all decision problems that can be solved in polynomial time with an error probability bounded away from $1/2$ on a quantum Turing machine~\cite{Bernstein1997}. 

\subsection{Problem Formulation}

The problem they solved is called \textsc{Recursive Fourier Sampling}, but here we will only consider the base problem that we call the \textsc{Bernstein-Vazirani} problem, where there is less than the above-mentioned advantage, see below.
Consider that we are given access to an oracle computing a Boolean
function, $f:\{0,1\}^n \to\{0,1\}$, promised to be of the particular linear
form $f(x)=s\cdot x$, where $x,s \in \{0,1\}^n$. The task is now to find the ``secret'' string~$s$.

\begin{figure}
\centering
\includegraphics[scale=1]{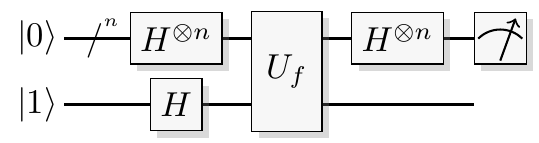}

\captionof{figure}{Quantum circuit for the \textsc{Bernstein-Vazirani} algorithm}
\label{fig:B-V}
\end{figure}	

\subsection{Classical Algorithm}

In the case where the oracle only gives us access to the function values, we query the
function for all inputs $x$ with a Hamming weight of $1$, i.e., all inputs
\begin{equation}
\tilde{x}_i=
\begin{cases}
x_j=1, &\text{for }j=i \\ 0, &\text{otherwise,} 
\end{cases}
\end{equation} where $j\in\{1,2,\ldots,n\}$ is the bit-index. Then the
function will answer with one new bit of information about $s$ for each query,
$f(\tilde{x_i})=s_i$. After $n$ queries we have retrieved $s$. 

\subsection{Quantum Algorithm}

If instead the oracle is given as a unitary transformation, adding the
result of the query to the query-qubit in the computational basis, specifically
$ \ket{x}\ket{y}\mapsto \ket{x}\ket{y\oplus f(x)} $, then a quantum algorithm
can solve the problem with a single query.

\begin{tcolorbox}
	\begin{algorithm}[\textcite{Bernstein1993,Bernstein1997}]
		Proceed with the following steps.
		\begin{enumerate}
			\item Prepare an $n$-qubit
			query register and an additional answer qubit in the state $\ket{0^n}\ket 1$.
			\item Apply the Walsh-Hadamard transform to the query register and
			the output qubit.
			\item Apply the oracle.
			\item Apply the Walsh-Hadamard transform to the query register.
			\item Measure the query register in the computational basis.
		\end{enumerate}
		The measurement in the last step will reveal the secret string $s$.
	\end{algorithm}
\end{tcolorbox}

From step 1 and 2 we have
\begin{equation}
\ket{0^n}\ket{1}\xrightarrow{H^{\otimes n}\otimes H} \frac{1}{\sqrt{2^{n}}}\sum_x \ket{x}\ket{-}.
\end{equation}
From here on, if nothing else is stated, summation indexes run over the
variable's whole domain. Applying the oracle (step 3) we get
\begin{equation}
\begin{split}
&\,\frac{1}{\sqrt{2^{n+1}}}\sum_x \ket x(\ket{ 0 \oplus f(x)} -\ket {1 \oplus
f(x)}) \\ =& \frac{1}{\sqrt{2^{n+1}}} \sum_x (-1)^{f(x)}\ket x(\ket 0 -\ket
1) \ ,
\end{split}
\end{equation}
and in step 4 the inverse Hadamard transform on the query register gives
\begin{equation}
\begin{split}
&\frac{1}{2^n\sqrt{2}} \sum_{x,z} (-1)^{f(x)-x\cdot z}\ket z(\ket 0 -\ket
1)\\
= \ &\frac{1}{2^n\sqrt{2}} \sum_{x,z} (-1)^{(s-z)\cdot x}\ket z(\ket 0 -\ket
1)\\
= \ &\frac{1}{\sqrt{2}}\ket s(\ket 0 - \ket 1) ,
\end{split}	
\label{eq:FourierSampling}
\end{equation}
where the last identity is given by evaluating the sum over $x$ to obtain $
2^n $ if $ z=s $ and zero otherwise. By measuring the query register (step 5) we
retrieve $s$ by calling the quantum oracle only once. This is compared to the linear number
of times in the case when we only have access to the function output.

\begin{figure*}[t]
\centering\hspace{-1cm}
\includegraphics[scale=1]{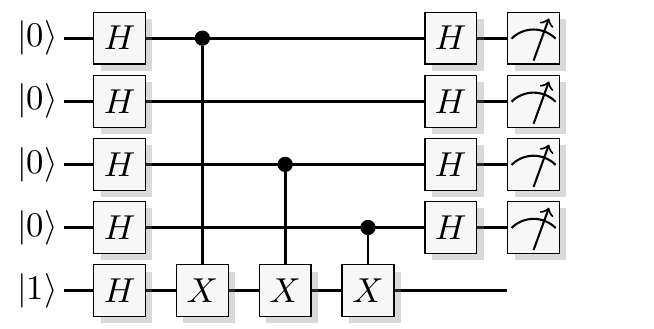}\hspace{3cm}
\includegraphics[scale=.65]{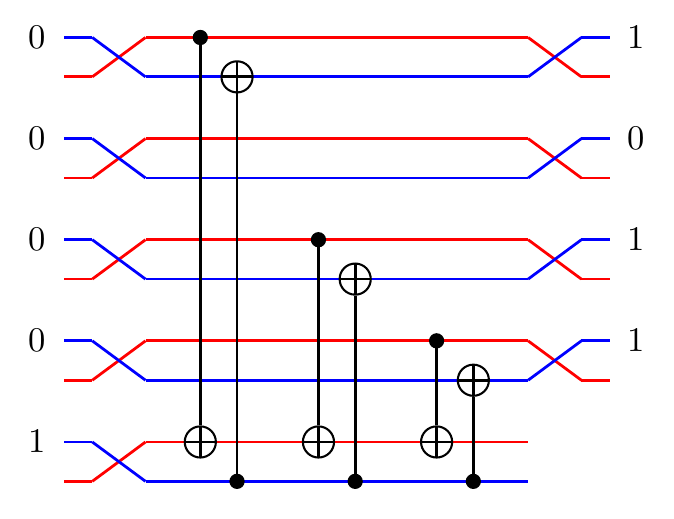}

\caption{Example of a quantum and the corresponding QSL algorithm for solving the \textsc{Bernstein-Vazirani} problem when the secret string is $ s=(1011) $. The unsigned red
inputs and outputs of the QSL-circuit are all uniformly distributed random
bits.}
\label{fig:BV_example}
\end{figure*}

\subsection{QSL Simulation}

To simulate this with QSL the first thing we need is to define the 
oracle. We know that it is promised to perform $f(x)=x\cdot s$ over the
computational basis, and such oracle can be constructed by only using $ CNOT $
gates. With their targets at the answer-qubit, and the controls on the qubits in
the query register that corresponds to where the bit values $s_i$ in the secret
string are set. This will produce the desired function. Observe that this allows
us to construct the whole algorithm using only Clifford group operations, and it
is well-known that there is an efficient classical simulation via the
Gottesman-Knill theorem. However, the Stabilizer subtheory is a nonlocal
contextual model, while QSL is local and non-contextual. Therefore, if this
construction works in QSL we can rule out these from being required properties.

First, prepare $n$ QSL-bits for the query register all in the
$(0,R)$ state, and one QSL-bit for the answer register initiated to $(1,R)$
\begin{equation}
(0^n,R)\;(1,R')
\end{equation} 
where $R$ is a uniformly distributed random bit string, and $R'$ is a uniformly
distributed random bit. Applying the QSL-gate simulating the Hadamard to all
QSL-bits gives
\begin{equation}
(R,0^n)\;(R',1)\ .
\label{eq:102}
\end{equation}

In general, simulating the above construction of the oracle will result in the
following map 
\begin{equation}
(x,p)\;(a,b) \mapsto \big(x,p\oplus bs\big)\;\big(a\oplus f(x),b\big).
\end{equation}

For the state in expression (\ref{eq:102}) the oracle will add $f(R)$ (modulo 2) to the
computational bit of the answer register. Since the phase bit $b$ of the
answer register is set, the oracle will also have the effect of flipping the
phase bits of every QSL-system that acts as a control in the query register.
Since the phase bits in the query register are all zero, flipping each
phase bit for which $s_i=1$. This will induce $s$ into the phase, giving the state
\begin{equation}
\big(R,s\big)\;\big(R'\oplus f(R),1\big).
\end{equation}
The Walsh-Hadamard transform again swaps all computational and phase bits
of the query register
\begin{equation}
\big(s,R\big)\;\big(R'\oplus f(R),1\big).
\end{equation}
Now, measuring the computational bits of the query register will reveal the
secret string $s$, at the cost of only one oracle call. \cref{fig:BV_example}
shows an example for when the secret string is $s=(1011)$.

This is a special case of the more general procedure known as \textsc{Fourier
Sampling}, where all Boolean functions are allowed, that is, we are not
restricted to functions of the form $f(x)=x\cdot s$. 
In \textsc{Fourier Sampling} the outcomes of a computational basis measurement
is weighted by the absolute square of the Fourier (Walsh-Hadamard) coefficients.
In the quantum algorithm, the Fourier-basis contains the Fourier transform of the map, giving direct access to sampling from that distribution. For a classical reversible oracle that only gives access to the function values, the
values of the Fourier transform is not directly accessible. In other words,
having access to the function encoded in a unitary operator enables us to sample
from some distributions related to the function, not only the function itself.
To sample from these distributions only having access to the function output
might be hard. A discussion about this in relation to the
\textsc{Bernstein-Vazirani} problem can be found in \parencite{Wojcik2006}. The \textsc{Deutsch-Jozsa} problem in the next section is also a special case of \textsc{Fourier
Sampling}.

\section{The \textsc{Deutsch-Jozsa} Problem}\label{sec:Deutsch-Jozsa} 

The \textsc{Deutsch-Jozsa} algorithm~\cite{Deutsch1992} was the first oracle
algorithm that suggested there could be a substantial advantage of doing
information processing on quantum systems. It is a generalization of the
\textsc{Deutsch} algorithm (see \cref{sec:Interference}), and have been used in
many experimental demonstrations of quantum computing. For a detailed account see~\cite{Perez-Garcia2016} and citations therein.

\subsection{Problem Formulation} We are going to use two different problem
formulations, the first is due to \textcite{Cleve1998}.
\begin{definition}\label{def:D-J1}
Consider that you are given access to an oracle encoding a Boolean function, guaranteed to be either constant or balanced.
The problem is to determine whether the function is constant or balanced. 
\end{definition}
A \emph{constant} Boolean function is one that always returns 1, or always
returns 0. While a \emph{balanced} function returns an equal number of 1s and 0s, i.e.,
the string of values for all $2^n$ possible input values will have a Hamming weight of $2^{n-1}$.
To distinguish this string of all output states $(f(2^n-1)\ldots f(1)f(0))$,
that completely characterize the function, from the bit-string that we usually
call output, we will call it a \emph{function string}.

The second problem formulation we are going to use is the original formulation by \textcite{Deutsch1992}.
\begin{definition}\label{def:D-J2}
Consider that you are given access to an oracle implementing a Boolean
function. The problem is to determine whether
\begin{itemize}
\item [(i)] the function is \textbf{not} constant, or 
\item [(ii)] the function is \textbf{not} balanced.
\end{itemize}
\end{definition}

Here the function is not guaranteed to be of one or the other kind, it can be
any Boolean function, and one of these statements can always be found true. This is therefore a decision problem in contrast to the promise problem in \cref{def:D-J1}.

For completeness, we will go through the algorithms for solving this problem
having access to an oracle implementing the function over bits, qubits, and
then QSL-bits. We will also consider both definitions of the problem.

\subsection{Deterministic and Probabilistic Algorithms}	

If we get access to an oracle that computes the function $f:\{0,1\}^n
\to \{0,1\}$, then we can query the function $2^{n-1}+1$ times and decide both
problems. If all these outputs are 0 (or 1) then it cannot be balanced, and we
have solved the problem described in \Cref{def:D-J2}. The problem as described in
\Cref{def:D-J1} is also solved by the same algorithm, the only difference is the promise.

This algorithm solves the problem (according to both definitions) with
certainty, but with a number of queries exponential in $n$. This is viewed
as evidence that, relative to the oracle, these problems are not in
$\mathbf{P}$ (the class of problems solvable in polynomial time).

However, if we only require the algorithm to return the solution with an error
probability bounded away from $1/2$, they can be solved with only a few
(constant number of) queries. An algorithm that fails with at most probability $1/4$ is to query the function three times, and answer  constant (or not balanced in the decision problem) if all three outputs are equal, otherwise balanced (or not constant in the decision problem). The error probability can be calculated as follows.

Since balanced functions have an equal amount of 1s
and 0s in the function string, choosing inputs at random we will see a 0 or a 1 at the output
with equal probability. So, if the function is balanced, and we
query the function three times, the output (000) or (111) both occur with
probability $1/8$. Thus, there is then a 1/4 probability of wrongfully guess the
function to be constant (not balanced), and a corresponding success probability of 3/4,
independent of the problem size.
If the function is constant (not balanced), the algorithm will succeed with unit
probability. For the explicit analysis see \cite{Deutsch1992}. This also shows that,
relative to the same oracle, this problem is in \textbf{BPP},
according to both definitions.

\subsection{Quantum Algorithm}

Here we instead are given access to a quantum oracle, namely, one that implements the
function as a unitary over the computational basis. Specifically,
\begin{equation}
U_f\ket{x}\ket{y} = \ket{x}\ket{y \oplus f(x)}.
\end{equation}
The algorithm is as follows (see also \cref{fig:D-J}).
\begin{tcolorbox}
\begin{algorithm}[\textcite{Deutsch1992}] Proceed with the following steps.
\begin{enumerate}
\item Prepare an $n$-qubit query register in the state $\ket{0}$, and an
output qubit in $\ket{1}$.

\item Apply the Walsh-Hadamard transform to the query register and
the output qubit.

\item Apply the oracle.

\item Apply the Walsh-Hadamard transform to the query register.

\item Test if the output state of the query register is $\ket{0}$ by a measurement of the observable $\outer{0}{0}\otimes I_2$.
\end{enumerate}
If the test is positive, output ``constant,'' otherwise ``balanced''.
\end{algorithm}
\end{tcolorbox} 
\noindent
In the algorithm, step 1 and 2 result in
\begin{equation}
\ket{0^n}\ket{1}\xrightarrow{H^{\otimes n}\otimes H} \frac{1}{\sqrt{2^{n}}}\sum_x \ket{x}\ket{-} .
\end{equation}
The oracle (step 3) transforms this into
\begin{equation}
\begin{split}
&\frac{1}{\sqrt{2^{n+1}}}\sum_x \ket{x}\Big(\ket{f(x)}-\big|\overline{f(x)}\big\rangle\Big)
\\&\qquad\qquad
= \frac{1}{\sqrt{2^{n}}}\sum_x (-1)^{f(x)}\ket{x}\ket{-} .
\end{split}
\end{equation}
After the final Walsh-Hadamard transform of the query register (step 4) this becomes
\begin{equation}
\frac{1}{2^{n}}\sum_{x,z} (-1)^{f(x)+ x\cdot z}\ket{z}\ket{-}.
\end{equation}
The query register part of this state is measured to find the value of the Hermitian
projector $\outer{0}{0}$, which can be seen as a test of whether the system
is in the state $\ket{0}$, in other words, it tests whether ``the
state of the query register is unchanged after applying the circuit''. This
test will be positive with probability
\begin{equation}
\left|\frac{1}{2^n}\sum_x(-1)^{f(x)}\right|^2=
\begin{cases}
1,\quad \text{if }f(x)\text{ is constant} \\
0,\quad \text{if }f(x)\text{ is balanced,}
 \end{cases}
\end{equation} 
and that allows us to solve the problem with only one query, showing that
relative to the oracle this problem according to both definitions is in the
complexity class \textbf{EQP} (problems exactly solvable on a quantum machine
in Polynomial time)\cite{Bernstein1993}.
\begin{figure}
	\centering
	\includegraphics[scale=1]{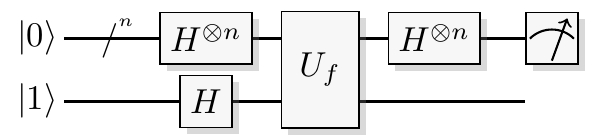}
	
	\captionof{figure}{Circuit representation of the \textsc{Deutsch-Jozsa} algorithm.}
	\label{fig:D-J}
\end{figure}

\subsection{The Problem for Small Input} 

It is well-known that for a small number of input qubits the quantum algorithm
admits an efficient classical simulation. \textcite{Collins1998} pointed out
that using the \textsc{Deutch-Jozsa} algorithm for a meaningful test of quantum
computation requires the number of input qubits to be strictly larger than
two. Similar findings have been made in
\textcite{Calude2007,Abbott2012a}.

The observation of \textcite{Collins1998} was that the algorithm was
completely independent of the answer register, which therefore can be omitted.
The important part of the construction is that it produces the correct
phase imprint in the query register. This is sometimes referred to as a
phase oracle. Please observe that this kind of oracle does not allow for
retrieving function values, and cannot be used to employ the regular solution,
but only makes the quantum algorithm available. 
Further, having access to a unitary implementing a Boolean function, a phase
oracle can efficiently be constructed~\cite{Bernstein1997}.

We can summarize these results simply by the fact that all balanced and
constant functions over one and two bits have oracles that only use
Clifford group operations (see \cref{fig:D-J_oracle_for_1and2})  and therefore admit a classical simulation in ether
QSL~\cite{Johansson2017} or the stabilizer subtheory.

\begin{figure}
\centering
\includegraphics[scale=1]{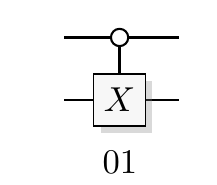}%
\includegraphics[scale=1]{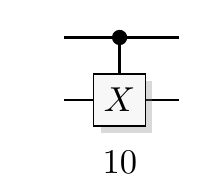}\\
\includegraphics[scale=1]{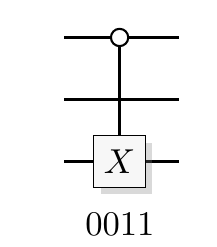}%
\includegraphics[scale=1]{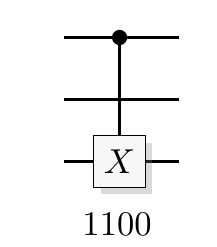}%
\includegraphics[scale=1]{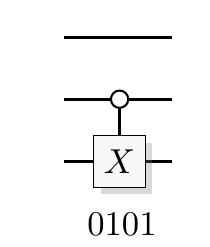}%
\includegraphics[scale=1]{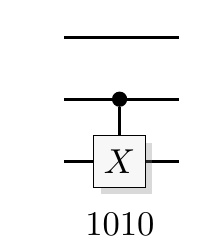}\\
\includegraphics[scale=1]{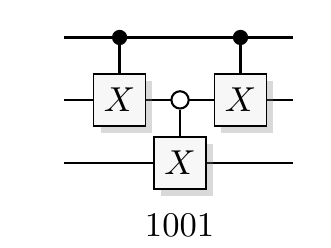}%
\includegraphics[scale=1]{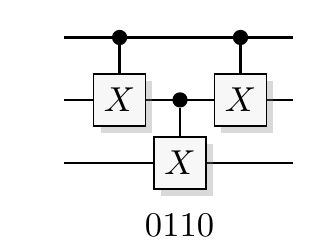}\\
\captionof{figure}{Explicit implementations of all balanced oracles for the \textsc{Deutsch-Jozsa} algorithm for one and two qubit input.}
\label{fig:D-J_oracle_for_1and2}
\end{figure}

For three qubits of input, and not constructing phase oracles, all 72 balanced
and constant function can be implemented in a reversible circuit using only
one Toffoli gate and five \textit{CNOT}s (for the explicit implementations see \Cref{sec:B}). Mapping these circuits
into the corresponding QSL-circuits, we obtain an implementation of an oracle
for each of the 72 functions, and using these to simulate the \textsc{Deutsch-Jozsa}
algorithm, QSL solves the problem with unit probability.

In general, there are two constant functions and
\begin{equation}
{2^n\choose{2^{n-1}}}\ge \Big(\frac{2^n}{2^{n-1}}\Big)^{2^{n-1}}=2^{2^{n-1}}
\end{equation} 
balanced functions, because of the simple bound
\begin{equation}
{n\choose k}\ge\Big(\frac nk\Big)^k.
\end{equation} 
This makes it intractable to construct, or even to index, all oracles explicitly. 
Thus, the oracle paradigm unavoidable in this problem setting. This will be discussed further in \cref{sec:Oracles}.

\subsection{QSL Simulation Guaranteed a Constant or Balanced Function}\label{sec:QSL_DJ1}

To approach the problem with QSL we need to determine the effect of an oracle
from QSL-bits onto QSL-bits.
We will start by specifying a valid implementation of a quantum
oracle, and then simply map all the quantum gates of the implementation
into QSL-gates. This will allow us to obtain an expression for the
effect of the QSL oracle, and show that there exists an oracle corresponding to a
simulation of a valid quantum oracle.
We then go on to show that if we are given access to this oracle, that is
sufficient to solve the problem with a single query.

The implementation of the quantum circuit that we are going to use is shown in
\cref{fig:D_J_oracle_1}, and employs two Toffoli gates and two boxes
representing a permutation $\pi \in S_{2^n}$ of the computational basis
states. These can be constructed from $NOT$, \textit{CNOT}, and Toffoli gates, and at
most one auxiliary bit~\cite{Shende2003}.

There is one parameter $\ket{b}$, that determines what type of function the circuit implements. If $\ket{b_0}=\ket{1}$ the function is balanced and if $\ket{b_0}=\ket{0}$ it is constant. 
If the function is constant, the value is $f(x)=b_1$.
The idea is to use a simple construction for one particular balanced function, and then use a permutation of the possible values $x$ to make all balanced functions available. 
If $\ket{b_0}=\ket{1}$ and the permutation is the identity, the Toffoli connected to $\ket{b_0}$ will give the balanced function $f(x)=x_{n-1}$
If the permutation is any other than the identity permutation, the input values are permuted before the function is calculated and written to the answer register. 
This makes all balanced functions available in the used reversible circuitry. 
The inverse permutation is used on the query register to uncompute the permutation, assuring that we get the map $\ket{x}\ket{y} \mapsto \ket{x}\ket{y \oplus f(x)}$.

\begin{figure}
\centering
\hspace*{.5cm}\includegraphics[scale=1]{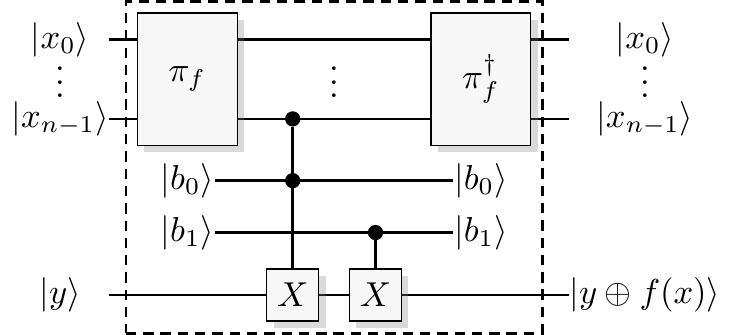}

\captionof{figure}{A specific implementation of a quantum oracle that can
	be used to solve the \textsc{Deutsch-Jozsa} problem in \Cref{def:D-J1}. The
	internal parameter $b_0$ chooses between a constant or balanced function, and if the function is constant, the value is $b_1$. 
	For balanced functions, the idea is that the
	first Toffoli acting as a \textit{CNOT} to produce one specific balanced function (if the permutation $\pi_f$ is the identity), adding the most significant bit to the output, $f(x)=x_{n-1}$. 
	Access to all balanced functions is then obtained by an arbitrary permutation of the output.}
\label{fig:D_J_oracle_1}
\end{figure}

We start by analyzing the effect of the permutation box mapped into QSL. 
Its effect can be separated into two parts, one is the effect on the computational bits and one on the phase bits. 
The effect on the computational bit-string is the same as the effect of the quantum gate $\pi_f$ on the computational basis. 
The effect over the phase bit-string will be another permutation, and since any computational basis permutation can be built from $X$, \textit{CNOT}, and Toffoli gates, the corresponding QSL gates can be used to build a QSL permutation
In general, if the construction uses Toffoli gates, the exact permutation on the phase bit-string will depend not only on which computational basis-string permutation is used, but also on the actual value of the computational bits.
Let us call this permutation $\pi_{f,x}$.
As we shall see, it does not matter precisely which permutation is realized, only that the resulting map is invertible.
The corresponding map to a computational basis permutation $\pi_f$ in QSL will be
\begin{equation}
(x,p) \mapsto \big(\pi_f(x), \pi_{f,x}(p)\big),
\label{eq:phasepermutation}
\end{equation}
and the effect of the inverse on the latter state, after a possible change of the phase bits, will be
\begin{equation}
(\pi_f(x),p') \mapsto \big(x, \pi_{f,x}^{-1}(p')\big).
\end{equation}
Then, the overall effect of the circuit in \cref{fig:D_J_oracle_1} is
\begin{equation}\label{eq:QSL_DJoracle}
\begin{split}
(x,p)\;&(y,r)
\xrightarrow{\pi_f} \ \big(\pi_f(x),\pi_{f,x}(p)\big)\;(y,r)\\
\xrightarrow{{\T_a}}\ &\big(\pi_f(x),\pi_{f,x}(p)+b_0r\delta_{n-1}\big)\;\big(y\oplus b_0{\pi_f(x)_{n-1}},r\big)\\
\xrightarrow{{\T_b}}\ &\big(\pi_f(x),\pi_{f,x}(p)+b_0r\delta_{n-1}\big)\;\big(y\oplus b_0{\pi_f(x)_{n-1}}\oplus b_1,r\big)\\
=\ &\big(\pi_f(x),\pi_{f,x}(p)+b_0r\delta_{n-1}\big)\;\big(y\oplus f(x),r\big)\\
\xrightarrow{\pi^{-1}_f}\ &\Big(x,\pi_{f,x}^{-1}\big(\pi_{f,x}(p)+b_0r\delta_{n-1}\big)\Big)\;\big(y\oplus f(x),r\big),
\end{split}
\end{equation}
where
$\pi_f(x)_{n-1}$ denotes the most significant bit of the input vector after
it has gone through the permutation, and $\delta_{n-1}$ is the bit vector for which the  most significant bit is 1 and all others 0.

If we now get oracle access to this function we can use it to query for particular function values, by preparing the registers in the states corresponding to $\ket{x}\ket{0}$ and applying the oracle. 
We can also use it to determine whether the function is constant or balanced by running the \textsc{Deutsch-Jozsa} algorithm. This gives
\begin{equation}
\begin{split}
(0,X)\;&(1,Y) 
\xrightarrow{\H^{\times (n+1)}}\ (X,0)\;(Y,1)\\
\xrightarrow{\U_f}\ &\Big(X,\pi_{f,X}^{-1}(\pi_{f,X}(0)+b_0\delta_{n-1}\big)\Big)\;(Y\oplus f(X),1)\\
\xrightarrow{\H^{\times n}}\ &\Big(\pi_{f,X}^{-1}\big(\pi_{f,X}(0)+b_0\delta_{n-1}\big),X\Big)\;(Y\oplus f(X),1).
\end{split} 
\end{equation}
If the function is constant, then $b_0=0$ and measurement of the query register will return $\pi_{f,X}^{-1}\big(\pi_{f,X}(0)\big)=0$. If
the function is balanced, then the measurement returns
$\pi_{f,X}^{-1}\big(\pi_{f,X}(0)+\delta_{n-1}\big)\neq0$, which is part from zero since $\pi_{f,X}^{-1}$ is injective. Note that the detailed behavior of the permutation $\pi_{f,X}$ is not important since knowing that it is invertible is enough. We have proven the following theorem.
\begin{theorem}
There is an efficient QSL algorithm that solves the \textsc{Deutsch-Jozsa}
problem by a single query to the oracle. This algorithm has an efficient
classical simulation on a PTM.
\end{theorem}
That the simulation is efficient follows from the fact that the QSL circuit uses a number of gates polynomial in the input size, relative to the oracle, and from
\cref{lemma:efficiency}. 
The QSL algorithm gives the same answers as the quantum algorithm: the zero bitstring if the function is constant, and a nonzero bitstring if the function is balanced. 
It is thus a faithful simulation of the quantum algorithm, in other words, these procedures are operationally the same.

\subsection{QSL Simulation Accepting Arbitrary Boolean Functions}\label{sec:D-J_Arbitrary_Boolean}

To approach the less strict problem formulation (\cref{def:D-J2}) we need a
construction that can produce all Boolean functions. To do that we have chosen
an implementation using a comparator.
A comparator is a device that compares two values, and we have chosen to compare $x+a$ and $2^n$, and output $1$ if $x+a\ge2^n$ and $0$ otherwise. 
This is an example of a function $f(x)$ that output 1 for $a$ of the possible inputs, and 0 otherwise, for each $a\in\{0,\ldots,2^n-1\}$. 
A generic such function can be generated by using the computational basis permutation presented in \cref{sec:QSL_DJ1}. 

Such a comparator can be built in reversible logic by adapting a construction for a
ripple-carry adder, known as a Cuccaro adder \cite{Cuccaro2004,Markov2012}.
The Cuccaro adder uses four reversible logic registers $(c_{\text{in}},x,a,z)$, where $x$ and $a$ are input-strings of equal length, and $c_{\text{in}}$, and $z$ are single-bit registers for carry-in and carry-out respectively. The mapping is
\begin{equation}
(c_{\text{in}},x,a,z)\mapsto (c_{\text{in}},x,x+a,z\oplus c_{\text{out}}),
\end{equation}
but since we want the above comparator, we note that $c_{\text{out}}$ is set if and only if $x+a\ge2^n$ (generating a carry). 
By using the initial value $z=0$ and uncomputing the ripple-carry-adder intermediate values instead of completing the subtraction (see \cite{Markov2012}), we obtain the map
\begin{equation}
(0,x,a,0)\mapsto(0,x,a,x+a\ge2^n).
\end{equation}

An illustrative example of a comparator built to compare two 4-bit integers is
shown in \cref{fig:comparator}.
\begin{figure}
	\centering
	\includegraphics[scale=.9]{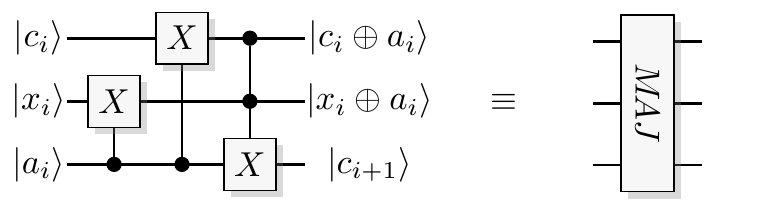}
	
	\captionof{figure}{$M\!AJ$ module where the output carry $c_{i+1}=a_i\oplus(c_i\oplus a_i)(x_i\oplus a_i)=a_ic_i\oplus a_ix_i\oplus c_ix_i$.}
	\label{fig:MAJ}
\end{figure}	

\begin{figure}
	\centering
	\includegraphics[scale=.7]{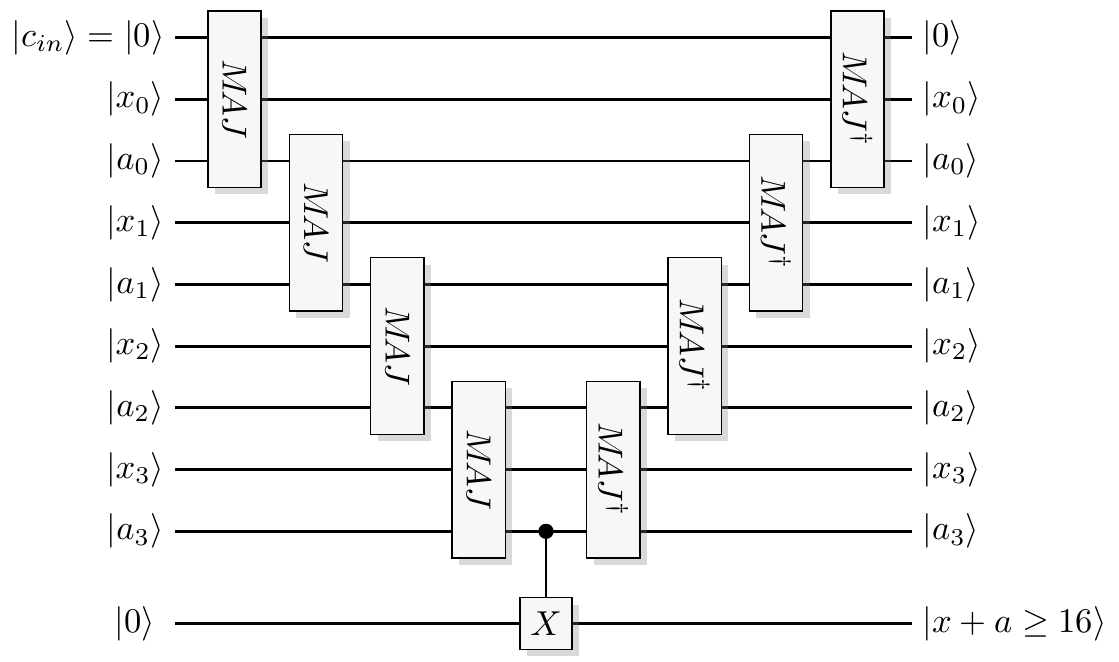}
	
	\captionof{figure}{Quantum circuit of a reversible comparator for comparing the sum of two 4-bit integers $x$ and $a$ to the number $2^4$, signaling if $x+a\ge2^4$ in the answer register.}
	\label{fig:comparator}
\end{figure}	

To enable the function that is constant $1$ one possible solution is to extend the comparator so that it works for $0\le a\le 2^n$.
Inverting the output if the additional bit $a_n$ equals 1 gives the constant function 0 for the value $a=0$ and the constant function 1 for the value $a=2^n$.
The complete construction of the circuit is shown in \cref{fig:all_Bool_oracle}.

\begin{figure}
\centering
\includegraphics[width=\linewidth]{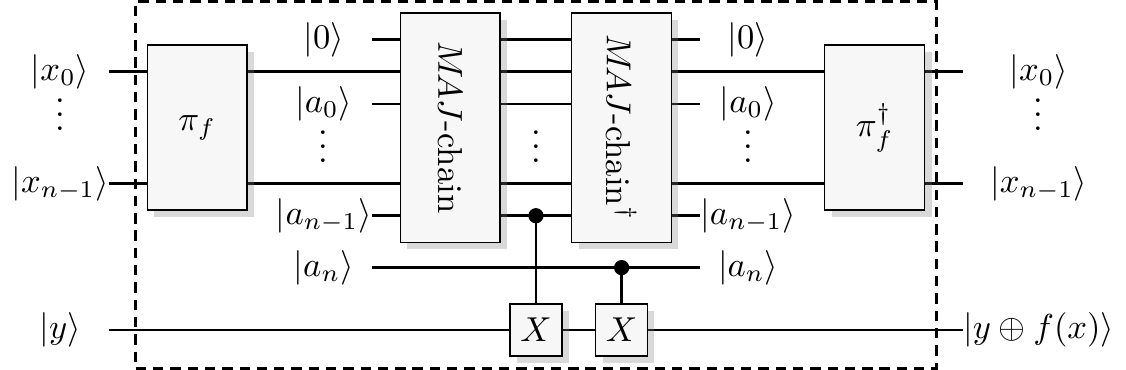}

\captionof{figure}{Construction of a circuit that can give any Boolean function. 
The idea is to construct one specific Boolean function for each number $0\le a\le2^n$, that give the output 1 on $a$ inputs by using a comparator. Then produce the other functions by using a permutation of the input states. 
The comparator is built through two chains of $M\!AJ$ and $M\!AJ^{\dagger}$ gates as in \cref{fig:comparator}.
}
\label{fig:all_Bool_oracle}
\end{figure}	

To verify that this gives the correct answer in the \textsc{Deutsch-Jozsa} algorithm, we need to check the phase kick-back of the gate array. 
For a given bit index, the $\MAJ$ and $\MAJ^{-1}$ gates reverse the phase transformation as well as the computational basis transformation. 
The one difference is the possible phase kick-back $\oplus k_i$ on the $a_i$ register, between the two gates, from the $i+1$st step of the chain or from the oracle target register if $i+1=n$. 
The effect of the three gates in the final $\MAJ^{-1}$ gate is
\begin{equation}
\begin{split}
  &(c_i\oplus a_i,\cdot)\;(x_i\oplus a_i,\cdot)\;(c_{i+1},\cdot\oplus k_i)\\
  &\rightarrow
  (\cdot,\cdot\oplus k_i(x_i\oplus a_i))\;(\cdot,\cdot\oplus k_i(c_i\oplus a_i))\;(\cdot,\cdot\oplus k_i)\\
  &\rightarrow
	(\cdot,\cdot\oplus k_i(x_i\oplus a_i))\;(\cdot,\cdot\oplus k_i(c_i\oplus a_i))\;\\
&\hspace{3.5cm}\times(\cdot,\cdot\oplus k_i\oplus k_i(x_i\oplus a_i))\\
  &\rightarrow
(\cdot,\cdot\oplus k_i(x_i\oplus a_i))\;(\cdot,\cdot\oplus k_i(c_i\oplus a_i))\;\\
&\hspace{3.5cm}\times(\cdot,\cdot\oplus k_i\oplus k_i(x_i\oplus c_i))\\
\end{split}
\end{equation} 
There are two cases we need to check.
If $f$ is constant then $a=0$ or $a=2^n$, and if $f$ is balanced then $a=2^{n-1}$.
In both cases $c_i=a_i=0$ for $0\le i\le n-2$, making $k_i(c_i\oplus a_i)=0$, so the phase kick-back is 0 for all the query register bits except the most significant bit.
For the most significant bit, it is still the case that $c_{n-1}=0$ so that the phase kick-back is $k_{n-1}(c_{n-1}\oplus a_{n-1})=k_{n-1}a_{n-1}$. 
The \textsc{Deutsch-Jozsa} algorithm now gives
\begin{equation}
\begin{split}
	(0,X)\;&(1,Y) 
	\xrightarrow{\H^{\times (n+1)}}\ (X,0)\;(Y,1)\\
	\xrightarrow{\U_f}\ &\Big(X,\pi_{f,X}^{-1}(\pi_{f,X}(0)+a_{n-1}\delta_{n-1}\big)\Big)\;(Y\oplus f(X),1)\\
	\xrightarrow{\H^{\times n}}\ &\Big(\pi_{f,X}^{-1}\big(\pi_{f,X}(0)+a_{n-1}\delta_{n-1}\big),X\Big)\;(Y\oplus f(X),1).
\end{split} 
\end{equation}
This is exactly the same behavior as in the previous section with $a=2^{n-1}b$. 
In fact, $a=2^{n-1}b$ makes the present oracle perform the exact same map as the one in the previous section, both in the computational and phase basis.
The same analysis as in \cref{sec:QSL_DJ1} now follows, and the measurement will reveal that the function was \emph{not balanced} ($a_{n-1}=0$ gives zero output) or \emph{not constant} ($a_{n-1}=1$ gives nonzero output). We have the following theorem.
\begin{theorem}
There is an efficient QSL algorithm that solves the \textbf{original}
\textsc{Deutsch-Jozsa} problem by a single query to the oracle. This
algorithm has an efficient simulation on a classical probabilistic Turing
machine.
\end{theorem}
That the simulation is efficient relative to the oracle follows from \cref{lemma:efficiency}. 

\subsection{Query complexity}

The standard analysis of the \textsc{Deutsch-Jozsa} problem states that if we only have access to the function $f:\{0,1\}^n \to\{0,1\}$, then it cannot be solved with unit probability using less than an exponential number of queries. 
This is usually taken as evidence that the problem is not in \textbf{P}. 
More carefully put, this is evidence that relative to this particular oracle, the problem is not in \textbf{P}. 
If we accept a solution with a bounded error-probability then, relative to the same oracle, the problem is in \textbf{BPP} \cite{Deutsch1992}.

Quantum computation allows us to solve the problem using only $\O(1)$ queries, so that relative to a quantum oracle the problem is in \textbf{BQP}.
However, the quantum oracle is different, since the function is encoded in a much richer framework --- as a transformation of the computational basis of a quantum system composed of qubits. 
Using this richer framework we can choose to extract additional information not available in the regular query model,
and it is this information that allows us to solve the problem using only $\O(1)$ queries. Relative to this \textit{richer} oracle, the problem is in \textbf{BQP}.

Note that, relative to an oracle in QSL, the \textsc{Deutsch-Jozsa} problem can be solved with only $\O(1)$ queries.
In addition, both the oracle and the \textsc{Deutsch-Jozsa} algorithm in QSL can be efficiently simulated on a classical probabilistic Turing machine, so that relative to an oracle in QSL, the \textsc{Deutsch-Jozsa} problem is in \textbf{BPP}.
In this particular case, there is no quantum advantage in terms of query complexity.

Furthermore, the ideal quantum algorithm solves the \textsc{Deutsch-Jozsa} problem with unit probability, so a more precise characterization is to put the problem into \textbf{EQP}. 
The QSL algorithm also provides the solution with unit probability, and in fact it provides it deterministically. 
It is even the case that the randomness present in the framework never contributes to the computation in the \textsc{Deutsch-Jozsa} algorithm, i.e., the random values can be replaced with a fixed value without interfering with the result. 
Combining this with \cref{lemma:efficiency}, we see that, relative to the QSL-oracles, the problem is in \textbf{P}.

\section{Oracles as a Comparison}\label{sec:Oracles}

The fact that the quantum query model possesses \textit{much richer} quantum oracles than standard classical binary-input binary-output function oracles, puts doubt in the standard comparison between the two.
The richness of the quantum oracle really calls for a comparison with a correspondingly rich classical framework in which oracles are able to encode the function in one subpart of the system, akin to the computational basis, and encode some additional function in another subpart of the system, akin to the phase basis. 
Any comparison between classical and quantum query complexity should take place between these two richer frameworks: one quantum, and one classical that allows these richer oracles. 
QSL has exactly the required properties, and can in turn be efficiently simulated on a classical probabilistic Turing machine.
There are several points to make here.

\subsection{The Additional Structure and Constraints}

Only having access to an oracle that computes a function from bits
to bits is insufficient for quantum computation. So what are the
conditions?

One condition is that the oracle needs to accept qubits as inputs. This is
sound but not sufficient, since there are many quantum operations that produce the
correct function map, but will not enable quantum computation. One example is
an implementation of the function over a completely phase-mixing channel.

A more precise condition is that the function needs to be implemented as a
unitary operator, but even this is not enough. The function needs to be implemented as a reversible function with all auxiliary bits cleared and the query register restored.
Even more important, the unitary operator also needs
to preserve the relative phases $c_{x,y}$
\begin{equation}
U_f\Big(\sum_{x,y}c_{x,y}\ket{x}\ket{y}\Big) =\sum_{x,y}c_{x,y}\ket{x}\ket{y\oplus f(x)}
\end{equation}

There are many unitary implementations encoding the function in the
computational basis for which the algorithm does not work. One example is the
unitary given by
\begin{equation}
U'_f\Big(\sum_{x,y}c_{x,y}\ket{x}\ket{y}\Big) =\sum_{x,y}c_{x,y}(-1)^{f(x)}\ket{x}\ket{y\oplus f(x)}
\end{equation}
for which the \textsc{Deutsch-Jozsa} algorithm will answer ``constant'' --- for all Boolean functions --- with unit probability \cite{Johansson2017}.
Another example is the unitary
\begin{equation}
U''_f\Big(\sum_{x,y}c_{x,y}\ket{x}\ket{y}\Big)
=\sum_{x,y}c_{x,y}(-1)^{x_0 +f(x)}\ket{x}\ket{y\oplus f(x)},
\end{equation}
for which the \textsc{Deutsch-Jozsa} algorithm will answer ``balanced'' --- for all Boolean functions --- with unit probability.
There are exponentially many unitaries of this kind.
Still, for all these examples, the function map is available and the oracle can still be used to solve the problem using the classical query algorithm.

A similar observation has been made in \textcite{Machta1998}.
In fact, most unitary implementations of $f$ will not
work for the algorithm since, in general, we can have an addition of
$(-1)^{g(x,y)}$ to the relative phases, where $g(x,y)$ is an arbitrary
function not necessarily related to~$f$.

Underneath the compact requirement that the function should be
implemented as the specific unitary $U_f\ket{x}\ket{y}=\ket{x}\ket{y\oplus f(x)}$
lurks an exponential number of constraints: the preservation of an
exponential number of complex amplitudes. 
This is the source of the richness of quantum oracles: the phase constraints enable access to much more information than can be accessed with the standard classical binary-input binary-output function oracle.
Only under these phase constraints it is possible to retrieve information otherwise not available from the phase degree of freedom, which in turn is useful for a more efficient solution of the problem under study.

\subsection{Is the Black-Box Black?}

The above difference between quantum oracles and classical oracles that only gives us access to function values, puts into question whether they can be used in a justified comparison. 
The broader non-black-box definition offers a good analogy. 
In this case, a quantum oracle should be described as a gray-box model; we have partial knowledge about the underlying model generating the statistics, rather than only getting access to the statistics.

Remember the example in \cref{sec:Preliminaries} where a gray-box model
is used in relation to cryptography. There, in addition to knowing the function map of the protocol, we also know that
it has a physical implementation in some electrical circuit, but we do not
know everything that there is to know about the box. In this model, which
includes a physical description, side channels become available, e.g., as information leakage through the power consumption of the circuit. In cryptography, a
protocol is not considered broken if there is a specific implementation that
leaks the secret, which might otherwise be hard to compute. Such a comparison
would be unjustified.

In quantum query complexity, the oracle also comes with a physical description,
namely, it needs to be implemented as one specific unitary, modulo global phase. Only
then will the additional information be available. It is not available
from where we usually read out the value of the function, i.e., in the computational basis. Instead, it is available as phase information, in another physical degree of freedom, constituting a side channel that allows us to access the extra information that enables the speed-up.
In quantum computing this side effect is systematically used to efficiently solve computational problems that are normally not efficiently solvable. 
This gives a direct comparison between solving computational problems in quantum computers, and breaking a cryptographic device through a side channel.

\subsection{Assumptions in the Use of Oracles}

Let us now briefly return to the definitions of \cref{sec:OracleDefinition}, and
note some differences to the use of the oracle notion in classical algorithms, quantum algorithms, and
QSL algorithms. In query complexity, an oracle is defined by what operation it
performs, having very few other properties, and then all statements that are
made are relative to that oracle. Informally the question is: if the oracle can
be computed in a single time step, what problems could then be solved with a
given amount of resources?

In classical query complexity the oracle computes the function and the standard
definition is that of a black box, i.e., there are no other properties to
declare. Relative to such an oracle the \textsc{Deutsch-Jozsa} problem cannot be
solved deterministically in polynomial time. In reversible classical logic the standard oracle definition is that of a black box that adds the result to a
target bit (see \cref{sec:OracleDefinition}). Also relative to such an oracle
the \textsc{Deutsch-Jozsa} problem cannot be solved deterministically in polynomial time.

In quantum query complexity the oracle is defined as a unitary transformation
computing the function. This gives a not-quite black box, because this adds the property of a specific relation
between the relative phases within the output states. Given such an oracle the \textsc{Deutsch-Jozsa}
problem can be solved with a single query. The justification for the interest in
quantum query complexity is that quantum theory predicts that these
oracles exist and can hypothetically be built.

Finally, in QSL the oracle is defined as a reversible transformation computing
the function, also adding the property of a specific relation between the
relative phases in the phase-basis output. Given such an oracle the \textsc{Deutsch-Jozsa} problem can be
solved with a single query. 
The justification for the interest in QSL query complexity is that QSL predicts that these oracles exist, in fact giving an explicit simulation in a classical probabilistic Turing machine.

In classical versus quantum oracle separations the assumption is that classical
oracles cannot encode any other information than the function. QSL serves as an
example that if we just slightly modify the classical reversible model, then we
do get classically simulatable oracles that systematically encode other information.

\subsection{Systematic Phase Errors}\label{sec:Sys_errors}

The additional information that we can choose to retrieve from a quantum
circuit, and that sometimes can be used to solve computational tasks, is
related to how unitary transformations manage relative phases. The
information in the amplitude and in the phase is mutually unbiased, and
QSL handles this by storing them independently. A function over the
computational bits is always accompanied by a corresponding function on the phase bits.

It would be convenient if, for any Boolean function $f$ over the
computational bits, we could have a unique expression for the effect on the
phase bits. Such an expression would be useful for analyzing the general
behavior of QSL. Unfortunately, different implementation of the same function
will result in different effects on the phase bits.

As an example, there is another procedure for creating the GHZ-state using
two \textit{CNOT}-gates rather than one \textit{CNOT} and one Toffoli-gate, as in
\cref{fig:GHZ_with_Toffoli}. This construction is shown in
\cref{fig:GHZ_with_cnots}. It is clear that these two constructions produce
different maps. They give the same map on the computational bits, but
different on the phase bits. It is not only the maps that are different, also the
probability distributions that we relate to quantum states are as well. To
see this, consider the marginal probability distributions from measuring the
least significant system in the computational basis, and the second system in
the phase basis. Doing this with the construction in
\cref{fig:GHZ_with_Toffoli}, we end up with a state represented by a
probability distribution that has 0 bits of entropy, as we have seen above.
On the other hand, doing the same for the construction in \cref{fig:GHZ_with_cnots}, using
two \textit{CNOT}-gates, the output state will have one bit of entropy.
\begin{figure}
\centering
\includegraphics[scale=0.85]{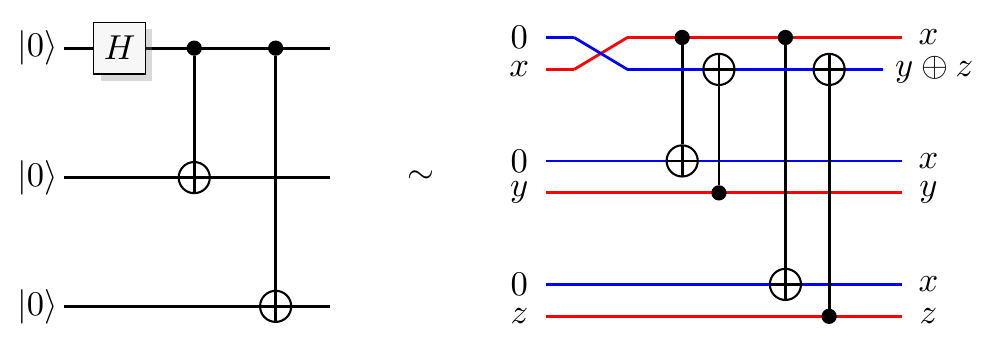}

\captionof{figure}{Another construction that produces a GHZ-state, and the
analogue construction in QSL.}
\label{fig:GHZ_with_cnots}
\end{figure}	

In other words, simulation of different constructions of quantum circuits
that produce the same function in the computational basis, does not produce
the same effect on the phase bits. The major inconvenience this phenomenon
brings is of course that we cannot analyze the general effect of applying a
function, but have to resort to specific constructions. If we find a
construction that lets us solve a problem efficiently, similarly as
in quantum theory, there is no guarantee that other constructions will work.
Conversely, if we fail to find a construction that works, we cannot say that
the quantum algorithm lacks an efficient simulation in QSL. However, this is
similar to what we would expect if we had an implementation of the Toffoli gate with systematic errors, and this will be clarified in what follows. 

The presence of systematic phase errors in a quantum gate array will influence the side channel enabled by the preservation of relative phases. 
Consider the following Toffoli gate that includes a systematic phase error $\theta$,
\begin{equation}
\begin{bmatrix}
1&0&0&0&0&0&0&0\\
0&1&0&0&0&0&0&0 \\
0&0&1&0&0&0&0&0 \\
0&0&0&1&0&0&0&0 \\
0&0&0&0&1&0&0&0 \\
0&0&0&0&0&1&0&0 \\
0&0&0&0&0&0&0&\makebox[3mm]{$e^{i\theta}$} \\
0&0&0&0&0&0&\makebox[0pt]{$e^{i\theta}$}&0 \\
\end{bmatrix}
\end{equation}
or equivalently
\begin{equation}
\sum_{x=0}^{5} \outer{x}{x} + e^{i\theta}\left(\outer{6}{7} + \outer{7}{6}\right).
\end{equation}
This unitary will be denoted $\tilde{T}$. Our aim is to
analyze four different constructions of the balanced function computing the
majority of three bits; the function with function string $ (11101000) $. The four constructions are shown in
\cref{fig:two_majority}. 

The construction in \cref{fig:two_majority}A is straightforward, using four 3-$\tilde{T}$ gates with their target in the answer register. 
If we construct the 3-$\tilde{T}$ gates as in \cref{fig:3-Toffoli_sym_error}, then the resulting operator will be
\begin{equation}
3\hypen\tilde{T}=\sum_{X=0}^{13} \outer{x}{x}+e^{i\theta}\left(\outer{14}{15} + \outer{15}{14}\right).
\end{equation}
The unitary implementation of \cref{fig:two_majority}A becomes $\ket{x}\ket{y}
\mapsto (e^{i\theta})^{f(x)}\ket{0}\ket{y \oplus f(x)}$, and the \textsc{Deutsch-Jozsa}
algorithm will wrongfully answer constant with probability
\begin{equation}
\left|\frac{1}{8}\sum_{x=0}^8 (e^{i\theta})^{f(x)}(-1)^{f(x)}
\right|^2=\left|\frac{4-4e^{i\theta}}{8}\right|^2=\sin^2\left(\frac{\theta}{2}\right).
\end{equation} 
With $\theta=\pi/3$ the error probability becomes $1/4$.

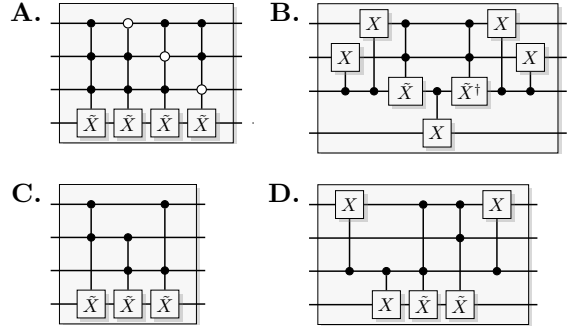
\begin{figure}
\centering
\begin{tikzpicture}
\node(a){\includegraphics[scale=0.7]{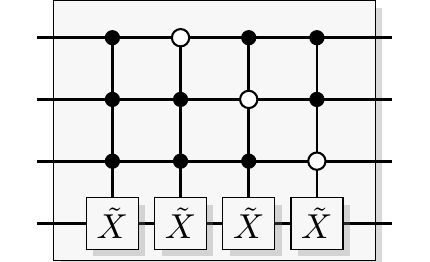}};
\draw (a.north west) +(4mm,0) node[anchor=north east] {\textbf{A.}};
\draw (a.north east) +(1mm,0) node[anchor=north west] (b){\includegraphics[scale=0.7]{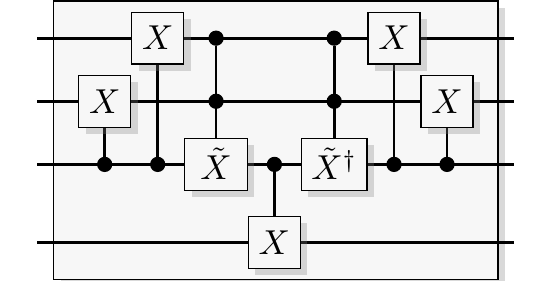}};
\draw (b.north west) +(4mm,0) node[anchor=north east] {\textbf{B.}};
\draw (a.north west) +(0,-24mm) node[anchor=north west] (c){\includegraphics[scale=0.7]{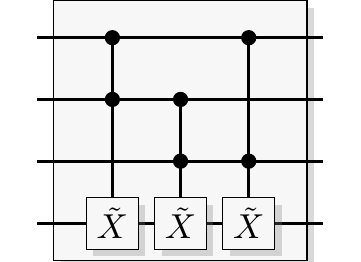}};
\draw (c.north west) +(4mm,0) node[anchor=north east] {\textbf{C.}};
\draw (b.north west) +(0,-24mm) node[anchor=north west] (d){\includegraphics[scale=0.7]{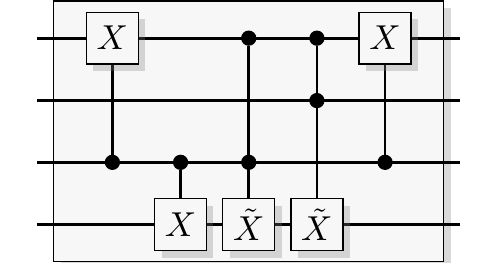}};
\draw (d.north west) +(4mm,0) node[anchor=north east] {\textbf{D.}};
\end{tikzpicture}
\captionof{figure}{Four different quantum circuit implementations of the same balanced function,
computing the majority $M\!AJ$ (see \cref{fig:MAJ}) of three bits. \textbf{A.} Straightforward
construction using four 3-$\tilde{T}$ gates. \textbf{B.} The majority  is
computed in-place in the query register, the answer is copied out to the
answer register, and the intermediate step is then uncomputed by applying
$M\!AJ^\dagger$ to the query register. \textbf{C.} is an optimization of \textbf{A.} tuned for as few gates as possible, and \textbf{D.} is another optimized for few Toffolis followed by few $CNOT$s. It was generated together with all 72 possible functions found in \cref{sec:B}.  } \label{fig:two_majority}
\end{figure}

\begin{figure}
\centering
\includegraphics[scale=.9]{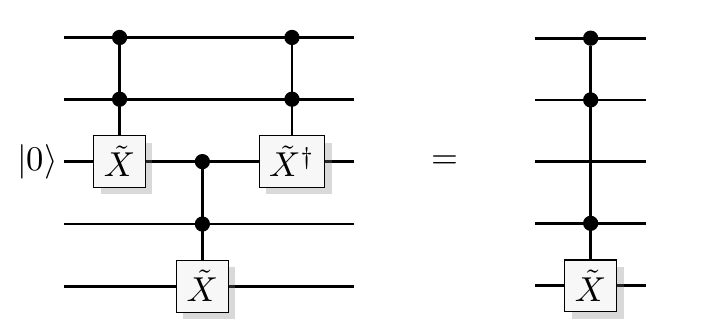}

\captionof{figure}{Scheme for constructing a 3-$\tilde{T}$ from three $\tilde{T}$ gates with an systematic error and an ancillary qubit initiated in
$\ket{0}$. }
\label{fig:3-Toffoli_sym_error}
\end{figure}

The construction in
\cref{fig:two_majority}B uses an intermediate step to compute the majority
in-place, using the $M\!AJ$-gate shown in \cref{fig:MAJ}, adding the result to the output, and then uncomputing the intermediate result by
applying $M\!AJ^\dagger$. With this construction, the \textsc{Deutsch-Jozsa} algorithm will (correctly) answer balanced with unit probability even though there is a phase error in the $\tilde{T}$ gate.

This reproduces the behavior of QSL, where using the QSL-Toffoli gives an error probability $ 1/4 $ when the oracle is
constructed as in \cref{fig:two_majority}A, and zero error probability if constructed as in \cref{fig:two_majority}B. The calculations can be found in \cref{sec:c}.

With the quantum circuit in \cref{fig:two_majority}C the error probability becomes 
\begin{equation}
\left|\frac{4-3e^{i\theta}-e^{i3\theta}}{8}\right|^2,
\end{equation}
and for \cref{fig:two_majority}D
\begin{equation}
\left|\frac{1-e^{i2\theta}}{8}\right|^2,
\end{equation}
and again, if $\theta=\pi/3$ these become $19/64$ and $1/64$ respectively. 
In QSL, only the three random bits going into the query register is taking part in the simulation.
Therefore, any probability distribution describing the outcome can only be resolved into fractions of $8$. 
The best approximations we can obtain for $19/64$ and $1/64$ is then $1/4$ and $0$ respectively. 
This is exactly the error probability that we see in QSL (see \cref{sec:c}).

Perhaps it is important to stress that the QSL-Toffoli is not a quantum Toffoli with a systematic error. 
It is a simulation of a quantum Toffoli that uses classical bits in the simulation. 
What we have seen here is that the QSL-Toffoli behaves similarly to a quantum Toffoli gate with systematic error, giving similar variations in output statistics for different constructions of the same function.

\subsection{Starting with Something Else Than Access to an Oracle}

It may also be argued that if we start from something else than access to an oracle, for instance a circuit implementation, it might be hard to translate that into QSL.
The question is why this would be considered an argument at all.
In the problems under study, there are simply too many possible functions to convey an explicit construction to the solver of the problem. 
For example, in the case of an explicit function from the \textsc{Deutsch-Jozsa} problem, the number of balanced functions grows as $ \sqrt{2^{2^n}} $, so that simply \textit{indexing} the explicit constructions would require an exponential amount of information.
How would it be possible to convey the explicit construction to the solver in this case?
If successful, the solver will have received an exponential amount of information in the transfer. 

One possibility to make the problem practical is to restrict the problem to a polynomial-sized subset of all balanced functions, but this would be a severe alteration of the problem formulation, and then there is no reason to believe that the proof of separation still applies.
The simplification to a polynomial-sized subset may enable a polynomial-time solution even in a classical Turing machine, and it is even possible to argue that this is likely the case.

In addition to this, given an explicit construction, for example in the form of a circuit or a procedure, it may well be that the structure of the explicit construction gives the solution away, the most clear example of this in quantum computation is the standard construction of the \textsc{Bernstein-Vazirani} oracle that can be found in \Cref{fig:BV_example}.
In the case that the solution cannot be found by inspection, the extraction of the solution from the structure would constitute a field of research in its own right. 
--- and with the interpretation of the extra information as a side channel, it could even be argued that quantum computation is part of that field of research. 

The above complications is the basic motivation to use the black-box query model in the first place. 
In the remainder of the paper, we will use the query model to investigate how QSL performs in two additional quantum algorithms: \textsc{Grover's} algorithm and \textsc{Simon's} algorithm. 
We will also have a look at a real-world algorithm, \textsc{Shor's} algorithm, where construction of the gate array is polynomial-time, and study the behavior of QSL in this situation.

\section{\textsc{Grover's} Algorithm}\label{sec:Grover's}

Given a Boolean function $f:\{0,1\}^n \to \{0,1\}$, \textsc{Grover's} algorithm~\cite{Grover1996}
is a probabilistic algorithm that returns an element of the preimage of
$f(x)=1$. That is, it returns with high probability a satisfying assignment
to $f$, if there is one. This algorithm can be used to solve the decision problem of answering whether $f$ has a satisfying assignment or not.

If we have no knowledge about the function, and we are only given oracle access
to it, we must resort to using exhaustive search --- simply querying the function for
all inputs until we find a $1$ or not. This will in the worst case require
$\O(2^n)$ queries. 
The best guess a probabilistic algorithm can use, without any knowledge about the
function, is to pick an input uniformly at random. This
will also require $\O(2^n)$ queries to solve with a bounded
error-probability.
Both of these estimates rely on the black-box nature of the function provided, because if the function gate array is available, it may be trivial to find the satisfying assignment, see for example \cref{fig:one-shot_Grover_example}.

The quantum algorithm proceeds as follows, and only requires
$\O(\sqrt{2^n})$ queries to solve the problem with a bounded error-probability.
\begin{tcolorbox}
\begin{algorithm}[\textcite{Grover1996}]
After performing a Walsh-Hadamard transform on the initial state
$\ket{0}^{\otimes n}\ket{1}$, repeat the following steps $\sqrt{2^n}$ times.
\begin{enumerate}
\item Apply the oracle
\item Perform a Walsh-Hadamard transform of the query register
\item Apply $(I-2\outer{0}{0})$ to the query register. This is an
$n$-controlled-$Z$ with all controls inverted (including the target $Z$).
\item Perform a  Walsh-Hadamard transform of the query register
\end{enumerate}
Then measure the query register, which will give you the satisfying
assignment with high probability (see further \cite{Nielsen2010}). 
One repetition of the steps 1-4 is known as the Grover operator.
\end{algorithm}
\end{tcolorbox}

\subsection{Problem Formulation} 

Here we will restrict the problem to the
worst case scenario, where there is only one satisfying assumption.
\begin{definition}
Assume that you are given access to an oracle for the function
$f:\{0,1\}^n\to\{0,1\}$ and that it only has one satisfying assignment
$x^*:f(x^*)=1$. Find $x^*$.
\end{definition}
This is the original problem formulation used by Grover. 
As usual, for the quantum case the quantum oracle is assumed to be encoded as a unitary function from qubits to qubits, and not as a Boolean function from bits to bits.

\subsection{One-Shot Grover} 
We will start with a small example which
is sometimes called \textsc{One-shot \textsc{Grover}}. This is the case where there
are two qubits in the query register\footnote{In general, the requirement for a
"one-shot" Grover is that $1/4$ of the input states are satisfying
assignments.} and the Grover operator only need to be applied once (see
\cref{fig:one-shot_Grover}).

\begin{figure}
\centering
\includegraphics[scale=1]{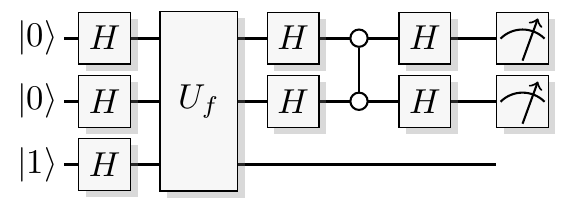}

\captionof{figure}{Circuit for the one-shot Grover instance.}
\label{fig:one-shot_Grover}
\end{figure}	

Since $U_f\ket{x}\ket{y}=\ket{x}\ket{y\oplus f(x)}$, it is straightforward to check that 
\begin{equation}
U_f\ket{x+}=\ket{x+},
\end{equation}
and
\begin{equation}
U_f\ket{x-}=
\begin{cases}
-\ket{x-}, & \text{ if }x=x^*\\
+\ket{x-}, & \text{ otherwise.}\\
\end{cases}
\end{equation}
This is another example of phase kick-back. From linearity we obtain
\begin{equation}
U_f=I\otimes\ket{+}\bra{+}+(I-2\ket{x^*}\bra{x^*})\otimes\ket{-}\bra{-}.
\end{equation}
The combination of steps 2, 3, and 4 is sometimes called ``inversion over the mean'', and can be written
\begin{equation}
  A=H^{\otimes 2}(I-2\outer{00}{00})H^{\otimes 2}=(I-2\outer{++}{++})
\end{equation}
The quantum algorithm proceeds as follows
\begin{equation}
\begin{split}
\ket{00}\ket{1} \xrightarrow{H^{\otimes 3}}\; &\ket{++}\ket{-}\\
\xrightarrow{U_f}\; &\left(\ket{++}-\ket{x^*}\right)\ket{-}\\
\xrightarrow{A\otimes I}\; &
\big(-\ket{++}-\left(\ket{x^*}-\ket{++}\right)\big)\ket{-}\\
=\; & -\ket{x^*}\ket{-},
\end{split}
\end{equation}
so that a measurement of the query register will reveal $x^*$ with unit
probability.

To see how this works in QSL we need to find an expression for the oracle. For
$n=2$ there are four functions with a single satisfying assignment
\begin{equation}
\begin{split}
&f(x)=\overline{x_1}\,\overline{x_0},\\
&f(x)=\overline{x_1}x_0,\\
&f(x)=x_1\overline{x_0},\ \text{and}\\
&f(x)=x_1x_0.\\
\end{split}
\end{equation}
The canonical constructions for these are to use the four different Toffoli gates
with different combinations of regular and inverted controls. Let us take the
example where $f(x)=\overline{x_1}x_0$. Here $x^*=(01)$ and the resulting
quantum circuit is shown in \cref{fig:one-shot_Grover_example}.
\begin{figure}
\centering
\includegraphics[scale=1]{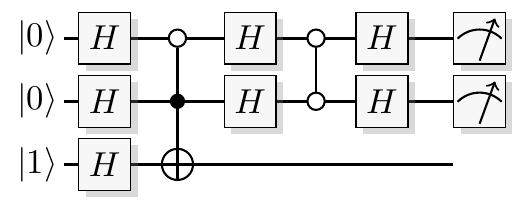}

\captionof{figure}{Example circuit of the one-shot Grover where $f(x)=\overline{x_1}x_0$.}
\label{fig:one-shot_Grover_example}
\end{figure}	 	
The operation $\mathcal{A}$ (``inversion over the mean'') would be
\begin{equation}
\begin{split}
&(x_1,p_1)\;(x_2,p_2)\\
\xrightarrow{\H^{\times 2}}\;&(p_1,x_1)\;(p_2,x_2)\\ 
\xrightarrow{\CZ_{00}}\;&
(p_1,x_1\oplus \overline{p_2})\;(p_2,x_2\oplus \overline{p_1})\\
\xrightarrow{\H^{\times 2}}\;&(x_1\oplus \overline{p_2},p_1)\;(x_2\oplus \overline{p_1},p_2).\\
\end{split}
\end{equation}
The complete algorithm becomes
\begin{equation}
\begin{split}
&(0,r_1)\;(0,r_0)\;(1,r_t)\\
\xrightarrow{\H^{\times 3}}\;&(r_1,0)\;(r_0,0)\;(r_t,1)\\ 
\xrightarrow{\O_f}\;&(r_1,r_0)\;(r_0,\overline{r_1})\;\big(r_t\oplus f(r),1\big)\\ 
\xrightarrow{\mathcal{A}\times \I}\;&
(r_1\oplus \overline{\overline{r_1}},r_0)\;(r_0\oplus \overline{r_0},\overline{r_1})\;\big(r_t\oplus f(r),1\big)\\=\;&
(0,r_0)\;(1,\overline{r_1})\;\big(r_t\oplus f(r),1\big)
\end{split}
\end{equation}
Measurement of the query register will reveal $x^*=(01)$ (see
\cref{fig:QSL_one-shot_Grover}). This will be true also for the other three
functions as well, and we will return to show this later on.

\begin{figure}
\centering
\includegraphics[width=\linewidth]{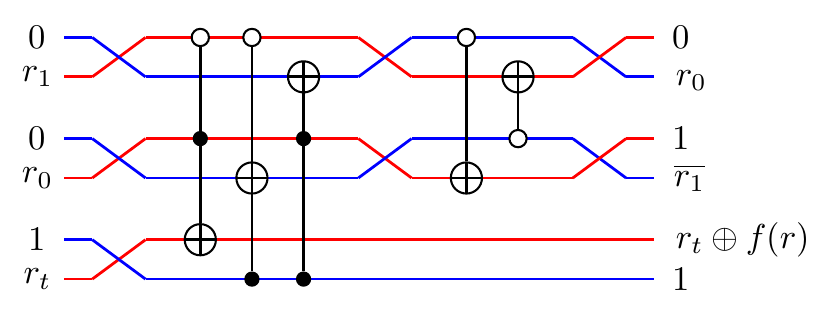}

\captionof{figure}{QSL simulation of the one-shot Grover instance from
\cref{fig:one-shot_Grover_example}.}
\label{fig:QSL_one-shot_Grover}
\end{figure}

\subsection{The $n$-Toffoli} 

To extend this to larger systems we need a description of the simulation of an $n$-Toffoli. 
For a quantum 3-qubit Toffoli we have the identity of \cref{fig:3-Toffoli_identity}.
\begin{figure}
\centering
\includegraphics[scale=1]{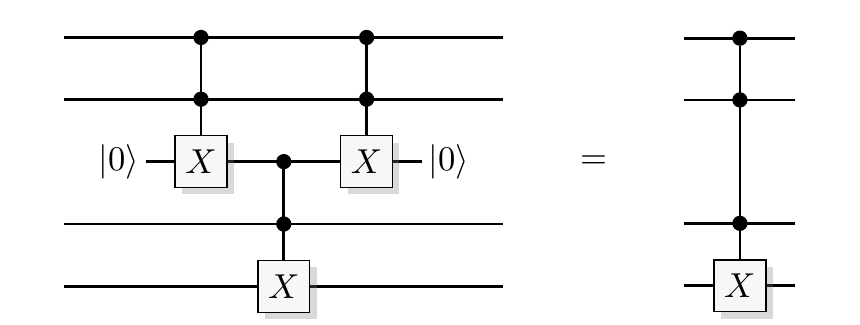}

\captionof{figure}{Scheme for construction a 3-Toffoli from three regular
Toffoli gates and an ancillary qubit initiated in $\ket{0}$.}
\label{fig:3-Toffoli_identity}
\end{figure}	
We see that the corresponding QSL map is
\begin{equation}
\begin{aligned}
&(x_2,p_2)&\quad&(x_2,p_2 \oplus t_px_0x_1)\\
&(x_{1},p_{1})&\quad&(x_{1},p_{1}\oplus t_px_0x_2)\\
&(0,a_1)&\mapsto\quad&(0,a_1)\\
&(x_0,p_0)&\quad&(x_0,p_0\oplus t_px_1x_2)\\
&(t_x,t_p)&\quad&(t_x\oplus x_0x_1x_2,t_p),
\end{aligned}
\end{equation}
which in QSL gives the relation shown in \cref{fig:QSL_3-Toffoli_identity}.
\begin{figure}
\centering
\includegraphics[scale=1]{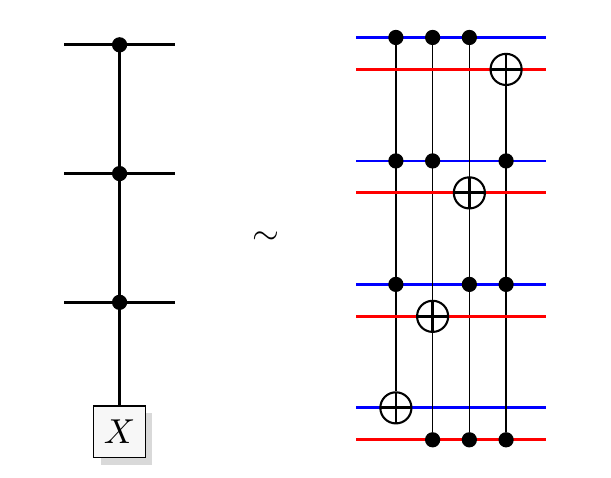}

\captionof{figure}{QSL identity for the 3-Toffoli}
\label{fig:QSL_3-Toffoli_identity}
\end{figure}	

Now using the same identity but by extending the 3-Toffoli to a 4-Toffoli (\cref{fig:4-Toffoli_identity}) we get the map
\begin{equation}
\begin{aligned}
&(x_3,p_3)\\
&(x_2,p_2)\\
&(0,a_1)\\
&(x_1,p_1)\\
&(x_0,p_0)\\
&(t_x,t_p)
\end{aligned}\quad\mapsto\quad
\begin{aligned}
&(x_3,p_3 \oplus t_px_0x_1x_2)\\
&(x_2,p_2\oplus t_px_0x_1x_3)\\
&(0,a_1)\\
&(x_1,p_1\oplus t_px_0x_2x_3)\\
&(x_0,p_0\oplus t_px_1x_2x_3)\\
&(t_x\oplus x_0x_1x_2x_3,t_p).
\end{aligned}
\end{equation}
\begin{figure}
\centering
\includegraphics[scale=1]{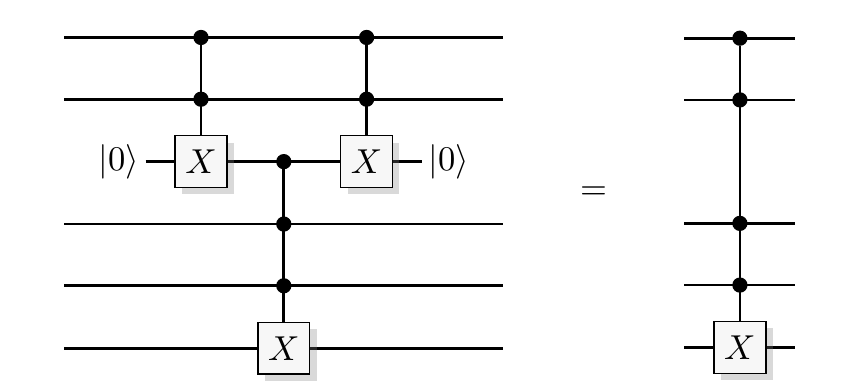}

\captionof{figure}{Scheme for construction a 4-Toffoli from one 3-Toffoli,
two regular Toffoli gates, and an ancillary qubit initiated in $\ket{0}$.}
\label{fig:4-Toffoli_identity}
\end{figure}	
By induction, we see that in QSL the map corresponding to an $n$-Toffoli is 
\begin{equation}
\begin{aligned}
&(x_n,p_n)&\quad&\Big(x_n,p_n \oplus t_p\big(\prod_{k\neq n}x_k\big)\Big)\\
&\quad \vdots&&\quad \vdots \\
&(x_1,p_1)&\raisebox{3ex}{$\mapsto$}\quad&\Big(x_1,p_1\oplus t_p\big(\prod_{k\neq 1}x_k\big)\Big)\\
&(x_0,p_0)&\quad&\Big(x_0,p_0\oplus t_p\big(\prod_{k\neq 0}x_k\big)\Big)\\
&(t_x,t_p)&\quad&\Big(t_x\oplus \big(\prod_{k}x_k\big),t_p\Big)
\end{aligned}
\end{equation} 

Using inverted controls by putting $\X$ before and after the control that is to be inverted, we see that this will not influence the value of the computational bits in the query register. 
The computational bit of the target and the phase bits of the query register will be affected, as the added $\X$ will induce the bit-complement for the corresponding $x_i$. 
With $f$ as the function with only one satisfying assignment $x^*$, 
\begin{equation}
	f(x)=\prod_k\big(x_k\oplus\overline{x^*_k}\big),
\end{equation}
built using a single $n$-Toffoli with some controls inverted, gives the QSL mapping
\begin{equation}
\begin{aligned}
\end{aligned}
(x,p)\;(t_x,t_p)\mapsto\big(x,p\oplus t_p\tilde{f}(x)\big)\;\big(t_x\oplus f(x),t_p\big).
\end{equation}
Here 
\begin{equation}
	\tilde{f}_i(x)=\prod_{k\neq i}\big(x_k\oplus\overline{x^*_k}\big)=f\Big(\sum_{k\neq i}2^kx_k+2^ix_i^*\Big),
\end{equation}
so to speak, the function with the  $i^{th}$ bit of the argument set to the correct value of the satisfying assignment $x_i^*$. 

\subsection{A Scaling Algorithm}

When increasing the number of input-bits the algorithm will have us to apply the Grover operator several times in sequence before measuring. 
A QSL simulation of this will not behave as the quantum algorithm, and we have found that the behavior is too far from the quantum behavior to give the quadratic speed-up of \textsc{Grover's} algorithm.
The underlying reason for this lack of speed-up needs more study, but we have found that simulating only the first application of the Grover operator gives a logarithmic speed-up, and this would give the following QSL algorithm.

\begin{tcolorbox}
\begin{algorithm}\label{alg:correct_repeat}
Repeat the following steps $\O\left(\frac{2^n}{n}\right)$ times
\begin{enumerate}
\item Prepare an $n$-QSL-bit query register in the $(0,r)$ state, and an
answer register containing one QSL-bit in the $(1,r_y)$ state.
\item Apply a Walsh-Hadamard transform on the query register
\item Apply the oracle
\item Apply a Walsh-Hadamard transform on the query register
\item Apply $\X^{\times n}$ to the query register.
\item Apply an $n$-controlled $\Z$ to the query register.
\item Apply $\X^{\times n}$ to the query register.
\item Apply a Walsh-Hadamard transform on the query register
\item Measure the query register to obtain a candidate $y$ for $x^*$.
\end{enumerate}
Check the output candidates for $x^*$ using the oracle, and return the candidate that outputs 1, if any.
\end{algorithm}
\end{tcolorbox}

The algorithm applies the Grover operator only once and measures immediately, and then proceeds to repeat this sequence, whereas standard Grover applies the Grover operator repeatedly, followed by a single measurement. Using our description of the $n$-Toffoli, we can now prove the following theorem.

\begin{theorem}
\Cref{alg:correct_repeat} solves Grovers problem with a bounded error
probability using no more than $\O\left(\frac{2^n}{n}\right)$ number of
queries to the oracle.
\end{theorem}

\begin{proof}
Let us start by calculating the probability distribution from which the measurement samples. 

The first four steps of the algorithm produces
\begin{equation}
(0,r)\;(1,r_t)\mapsto\big(\tilde{f}(r),r\big)\;\big(r_t\oplus f(r),1\big).
\end{equation}
Applying steps five, six, and seven will have the equivalent effect on the query register as applying another $n$-Toffoli with all controls inverted, on a target system with the phase bit set. This corresponds to applying another Boolean function
\begin{equation}
	g(x)=\prod_i \overline{x_i}.
\end{equation}
The effect on the query register is
\begin{equation}
\big(\tilde{f}(r),r\big)
\mapsto\Big(\tilde{f}(r),r\oplus \tilde{g}\big(\tilde{f}(r)\big)\Big)
,
\end{equation}
where $ \tilde{g}_i$ is similar to $\tilde{f}_i$, so that $ \tilde{g}_i$ is $g$ with the correct bit value used in position $i$ (here 0, the satisfying assignment of $g$). 
Measuring after the final Walsh-Hadamard transform we sample from the bitwise distribution
\begin{equation}
y_i =r_i\oplus \underbrace{ \Big(\prod_{k \neq i} \overline{\tilde{f}_k(r)}\Big)}_{\tilde{g}_i(\tilde{f}(r))}
\end{equation}
where $r$ is the whole random bit string, $i$ the bit index, and $\tilde{f}_k$
the different functions described above. Let us split the analysis into three
cases
\begin{enumerate}
\item If $r=x^*$, all $\tilde{f}_i(r)=1$, and all $\tilde{g}_i=0$, so that the output is $r=x^*$.

\item If $r_k=x^*_k$ except for a single $r_i\neq x^*_i$ ($r$ has a Hamming
distance of $1$ to $x^*$), then $\tilde{f}_k(r)=\delta_{ik}$ since $\tilde{f}_i(r)$ uses the correct value $x^*_i$ in place of $r_i$, while $\tilde{f}_k(r)=0$ for $k\neq i$ since the wrong bit-value $r_i$ is being used throughout. 
This in turn gives $\tilde{g}_k\big(\tilde{f}(r)\big)=\delta_{ik}$ by the same mechanism (the satisfying assignment of $g$ is all-zero values). Then, the output $r\oplus\tilde{g}\big(\tilde{f}(r)\big)=r\oplus\delta_{ik}=x^*$.

\item When two or more bits are wrong, all $\tilde{f}_k(r)=0$, and
all $\tilde{g}_i=1$. 
This will give as output the bitwise complement $\overline{r}$, which is not equal to $x^*$ except for when \textit{all} bits are wrong, in which case $\overline{r}=x^*$.
\end{enumerate}

In short, we can think of this as a protocol where a random guess $r$ is corrected to $x^*$ if there is a single bit error in $r$, or if $r$ is the bitwise complement of $x^*$. 
The sampling distribution therefore has a probability $p=\frac{n+2}{2^n}$ of obtaining $x^*$, probability $p=0$ for obtaining the bitwise complement $\overline{x^*}$ or a bitstring with Hamming distance of $1$ to $x^*$, and probability $p=\frac{1}{2^n}$ for any other value. 
Repeating the subroutine $\kappa2^n/(n+2)=\O(2^n/n)$ times gives a probability of not obtaining $x^*$ even once less than $e^{-\kappa}$ (Theorem 5 of \cite{Grover1996}).
\end{proof}

An alternative is to repeat the subroutine until the first occurrence of $x^*$. This gives a solution to the problem with error probability 0, but instead gives a random number of trials that obeys the shifted geometric distribution with probability $(n+2)/2^n$ of obtaining $x^*$. Then, the expected number of subroutine calls to the first success is $2^n/(n+2)$, so that the expected number of calls to the oracle including checking if the output $y=x^*$ is $2^{n+1}/(n+2)=\O(2^n/n)$. 

To connect with the fist example of one-shot \textsc{Grover's} for when $n=2$. There is a total of four assignments, one correct, its bitwise complement, and two strings with one bit difference. The QSL simulation will therefore always return the correct assignment for the one-shot \textsc{Grover's}.

\subsection{Comparison with a 3-qubit Experiment} 

For three bit input, running the protocol once, as in the above proposal, is a simulation of the setup used in a recent 3-qubit Grover experiment on trapped ions by \citeauthor{Figgatt2017}~\cite{Figgatt2017}.
The authors use two measures to characterize this experiment, the success probability of the algorithm and the squared statistical overlap (SSO). 
The SSO is the square of the Bhattacharyya coefficient (see \cref{eq:Bhattacharyya}) between the estimated probability distribution that the experiment samples from and the theoretical distribution.

While \citeauthor{Figgatt2017} obtain an average SSO of $83.2(7)\%$, our simulation gives
$78.5\%$, which is slightly lower. 
On the other hand, their average success probability is $38.9(4)\%$, while our simulation of the same circuitry has a success probability of $5/8=62.5\%$, much closer to the theoretical $25/32\approx78.1\%$.
It is important to note that when we compare quantum and classical bounds for
some problem, we compare the \emph{best known} bounds --- of course including
simulation algorithms. 

\subsection{Application to Ciphers} 
\label{sec:GroverCiphers}

An application of \textsc{Grover's} algorithm gives a quadratic speed-up to
breaking most ciphers for which otherwise exhaustive search is the best-known
method. Let us have a closer look at how that works.

Let us suppose that Alice has a classical system that uses a bijective map
described by the function $E(k,m)=c$, that takes as arguments a plaintext
message $m$, a key $k$, and produces a ciphertext $c$. 
Alice uses this function to encrypt information and its inversion (as a map from $m$ to $c$) to decrypt. 
When the same key $k$ is used both for encryption and decryption this is called a symmetric cipher, and when not, an asymmetric cipher.

An attack that uses \textsc{Grover's} algorithm is a \emph{known-plain-text attack}. 
In such an attack, Eve possesses (or can obtain) both the plaintext and the corresponding ciphertext produced with Alice's key. 
Let us for simplicity assume that the length of the known plaintext is longer than the unicity length, meaning it is long enough so that the key is uniquely determined by $E$, $m$, and $c$.

To retrieve $k$, Eve also needs to know $E$. That she has in her possession a copy of Alice's system, or a schematic, can be justified with Kerckhoff's principle stating that a secure cryptographic system should remain secure under the assumption that an adversary knows everything about the cryptographic system, except for the key.

Now assume that it is possible for Eve to efficiently translate Alice's classical device into a quantum device, and that this process does not give her any information about the key. 
Then she can set up her system to encrypt $\ket{m}$ and add $1$ to a second register, conditioned on seeing $\ket{c}$ (see \cref{fig:cipher_to_Grover}). 
This unitary can now be used to find $k$ in $\O(\sqrt{2^n})$ time using \textsc{Grover's} algorithm.

We conjecture that this attack can be done within the QSL framework. However, it remains to be analyzed for which specific cases this does work since the behavior depends on the implementation. Also, note that the advantage from the above algorithm is smaller, only the factor $1/n$ in $\O(2^n/n)$.

\begin{figure}
\centering
	\includegraphics[scale=1]{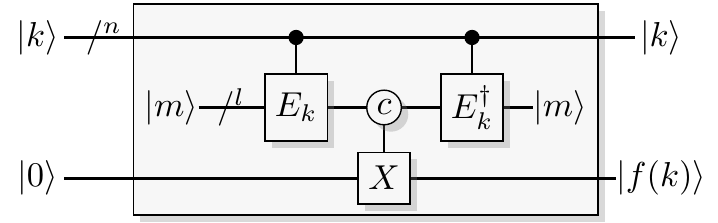}
	\captionof{figure}{Circuit construction for using \textsc{Grover's} algorithm to
		perform a known plaintext attack on a symetric cipher $E_k$, where the key
		$k$ has a bitsize of $n$ and the known plaintext a bitsize of $l$. The
		control labeled $c$ expresses that $X$ is applied only if the output of
		$E_k\!\ket m$ is $\ket{c}$, inverting the target bit. This can be constructed with an $l$-Toffoli with regular
		and inverted controlls.} \label{fig:cipher_to_Grover}
\end{figure}	

\section{\textsc{Simon's} Algorithm}\label{sec:Simon's}

\textsc{Simon's} problem \cite{Simon1994,Simon1997} is the archetype of quantum exponential speed-up, as the first case of a relativized exponential gap between \textbf{BPP} and \textbf{BQP}. Its
method has been reused and generalized many times over, and was the
inspiration for \textsc{Shor's} algorithm~\cite{Shor1994,Shor1999}.

\subsection{Problem Formulation}

We will work from the following problem formulation 
\begin{definition}
Consider that we are given access to an oracle encoding a function
$f:\{0,1\}^n\to\{0,1\}^n$, and promised that it is either one-to-one or
two-to-one and invariant under a non-trivial XOR-mask $s$.
\end{definition}

A function being invariant under a non-trivial XOR-mask $s$ means that the
function $f$ yield the same output on two different inputs, $x$ and $x'$, only
if these two input differs exactly on those bit positions where the bits of the mask
$s$ are $1$. In other words, $f(x)=f(x')$ if and only if $x=x'\oplus s$.

\subsection{Probabilistic Solution}

Assuming that we only have access to an oracle that only gives us the output
of the function $f:\{0,1\}^n\to\{0,1\}^n$. Simon proved the following theorem.

\begin{theorem}[\citeauthor{Simon1994} \cite{Simon1994,Simon1997}]\label{thm:Simon's}
Let $O$ be an oracle constructed as follows: for each $n$, a random $n$-bit
string $s(n)$ and a random bit $b(n)$ are uniformly chosen from $\{0, 1\}^n$
and $\{0, 1\}$,	respectively. If $b(n) = 0$, then the function $f_n : \{0,
1\}^n \to \{0, 1\}^n$ chosen for $O$ to compute on $n$-bit queries is a
random function uniformly distributed over permutations on $\{0, 1\}^n$;
otherwise, it is a random function uniformly distributed over two-to-one
functions such that $f_n(x) = f_n(x \oplus s(n))$ for all $x$, where $\oplus$
denotes bitwise exclusive-or. Then any PTM that queries $O$ no more than
$2^{n/4} $times cannot correctly guess $b(n)$ with probability greater than
$(1/2)+2^{-n/2} $, over choices made in the construction of $O$.
\end{theorem}

The proof can be found in Ref.~\cite{Simon1997}, here we will instead give a simple  motivation as to why this will require an exponential number of queries.
Consider an algorithm that queries the oracle on random inputs. If the
algorithm finds the same function value twice it halts, and we can conclude
that the function is two-to-one and invariant under some XOR-mask $s$. Now, if
the queries are chosen from a uniform distribution (there is no point in
choosing input from another distribution since $s$ is chosen uniformly) the
range of the function can also be seen as a probability distribution. If the
function is one-to-one it will be uniform. If the function is instead two-to-one,
half of the sample space will have $p=0$ and the rest will be uniformly
distributed with $p=\frac{2}{2^n}$. It is therefore unlikely that an algorithm
using less than an exponential number of queries will see the same value
twice.

\subsection{Quantum Algorithm}

Let us now assume that we are given access to an oracle that encodes the function in the unitary transformation 
\begin{equation}
\ket{x}\ket{y}\xrightarrow{U_f}	\ket{x}\ket{y\oplus f(x)}.
\end{equation}

\textsc{Simon's} algorithm can then make use of the subroutine shown in
\cref{fig:Simon's_circuit} and proceed as follows
\begin{equation}
\begin{split}
\ket{0}\ket{0}\xrightarrow{H^{\otimes n}\otimes I^{\otimes n}}\
&\frac{1}{2^{n/2}}\sum_x \ket{x}\ket{0}\\ \xrightarrow{U_f}\
&\frac{1}{2^{n/2}}\sum_x \ket{x}\ket{f(x)}\\ \xrightarrow{H^{\otimes
n}\otimes I^{\otimes n}}\ &\frac1{2^{n-b/2}}\sum_{x,p}
(-1)^{x\cdot p}\ket{p}\ket{f(x)},
\label{eq:simonmap}
\end{split}
\end{equation}
where the scalar product is defined mod 2. The normalization constant on the last row depends on whether $f$ is one-to-one ($b=0$) or two-to-one ($b=1$), and the difference is because of subtractive interference in the two-to-one case, see below.

If the function is one-to-one, measuring the query register will yield an outcome $p$ drawn from a uniform distribution. 
If the function is two-to-one there will be destructive interference so that some terms cancel in the sum, and using $f(x)=f(x\oplus s)$ the output can be written
\begin{equation}
\begin{split}
&\frac1{2^{n-1/2}}\sum_{x,p}(-1)^{x\cdot p}\ket{p}\ket{f(x)}\\
&\;=\frac1{2^{n+1/2}}\sum_{x,p}(-1)^{x\cdot p}\ket{p}\big(\ket{f(x)}+\ket{f(x\oplus s)}\big)\\
&\;=\frac1{2^{n+1/2}}\sum_{x,p}\big((-1)^{x\cdot p}+(-1)^{(x\oplus s)\cdot p}\big)\ket{p}\ket{f(x)}\\
&\;=\frac1{2^{n+1/2}}\sum_{x,p}(-1)^{x\cdot p}\big(1+(-1)^{s\cdot p}\big)\ket{p}\ket{f(x)}\\
&\;=\frac1{2^{n-1/2}}\sum_{x,p:s\cdot p=0}(-1)^{x\cdot p}\ket{p}\ket{f(x)}
\end{split}
\end{equation}
This will yield an outcome drawn from a distribution that is uniform over values $p$ such that $s\cdot p=0$. 

Each call to the subroutine gives independent outcomes so after an expected linear number of rounds, we will have collected $n-1$ linearly independent bit-vectors. 
Solving the resulting linear system of equations will give a non-trivial solution $s^*$ that will be our candidate for $s$. 
If the function is one-to-one, $s^*$ will be a random bit-vector.
By using $s^*$ to query the function twice, for $f(x)$ and $f(x\oplus s^*)$ for some $x$, we obtain the solution to the problem. 
If the function is two-to-one and invariant under the XOR-mask $s^*$, the two queries will yield the same result, otherwise not.
\begin{figure}
\centering
\includegraphics[scale=1]{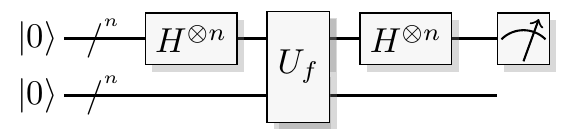}

\captionof{figure}{Quantum circuit used as a subroutine in \textsc{Simon's} algorithm.
The subroutine assumes oracle access to $U_f$. } \label{fig:Simon's_circuit}
\end{figure}	

\begin{tcolorbox}
	\begin{algorithm}[\citeauthor{Simon1994} \cite{Simon1994,Simon1997}]
		Proceed with the following steps.
		\begin{enumerate}
			\item Make $n-1$ repetitions of \Cref{sub:Simon} (also shown in \cref{fig:Simon's_circuit}). This will with high probability 
      yield $n-1$ linearly independent bit-vectors orthogonal to $s$.
			\item Solve the resulting linear system of equations to obtain the candidate solution $s^*$
			\item Query the function for $f(x)$ and $f(x\oplus s^*)$, for some $x$.
		\end{enumerate}
		If the two queries gives the same value, then the function is two-to-one and invariant under the XOR-mask $s=s^*$, and otherwise not.
	\end{algorithm}
\end{tcolorbox}\noindent
\begin{tcolorbox}
	\begin{subroutine}[\citeauthor{Simon1994} \cite{Simon1994,Simon1997}]\label[sub]{sub:Simon} Proceed with the following steps.
		\begin{enumerate}
			\item Prepare two $n$-qubit registers in the state $\ket{0}$.
			
			\item Apply the Walsh-Hadamard transform to the query register.
			
			\item Apply the oracle.
			
			\item Apply the Walsh-Hadamard transform to the query register.
			
			\item Measure the query register to retrieve a bit-vector $p$.
			\end{enumerate}
			Return $p$.
	\end{subroutine}
\end{tcolorbox}

\subsection{QSL Simulation} 
\label{sec:SimonsQSL}

\Cref{thm:Simon's} does of course not apply to the quantum case where the function oracle is given as a unitary map with $n$ qubits as input and $n$ qubits as target, and not a function from $n$ bits to $n$ bits. 
Neither does it apply to the QSL framework since the oracle is a map that uses $n$ QSL-bits as input and $n$ QSL-bits as target, or if one prefers, from $2n$ bits to $2n$ bits in its classical simulation.

The base construct that we will use for \textsc{Simon's} subroutine is that used in \textcite{Tame2014} (see their supplementary material) and consists of a network of \textit{CNOT} gates. 
It is generated by choosing a basis $\{v^{(k)}\}_{k=0}^{n-2}$ for a (perhaps \textit{the}) subspace of $\mathbb{Z}_2^n$ orthogonal to $s$, and storing the scalar product between the input $x$ and $v^{(k)}$ in the $k$-th bit of the output $y$.
This is done by connecting a \textit{CNOT} between the input bit $x_i$ and the output bit $y_k$ for each entry where $v^{(k)}_i=1$, producing the mapping
\begin{equation}
\ket{x}\ket{y}\mapsto \ket{x}\ket{y \oplus f^{(1)}(x)}
\end{equation}
where
\begin{equation}
	f_k^{(1)}(x)=
	\begin{cases}
	0,&k=n-1\\
	x\cdot v^{(k)}=(x \oplus s)\cdot v^{(k)},&\text{otherwise}.
	\end{cases}
	\label{eq:f1}
\end{equation} 
The resulting function $f^{(1)}(x)$ is a two-to-one function invariant under the XOR-mask $s$.
We label the resulting gate array $U_s$.

As an example consider the 3-qubit case where $s=(101)$. the bit-vector basis
can be chosen to $\{(101), (010)\}$. From $v^{(0)}=(101)$ we get two \textit{CNOT}s
with their controls at $\ket{x_0}$ and $\ket{x_2}$, and target at $\ket{y_0}$.
From $v^{(1)}=(010)$ a \textit{CNOT} with control at $\ket{x_1}$ and target at
$\ket{y_1}$ (see \cref{fig:U_S_example}).

\begin{figure}
	\centering
	\includegraphics[scale=1]{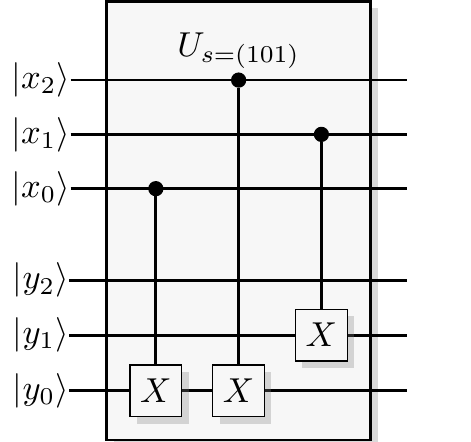}
	
	\captionof{figure}{Example construction for $U_s$, where $s=(101)$ and
		$\{v^{(k)}\}_{k=0}^1=\{(101), (010)\}$.} \label{fig:U_S_example}
\end{figure}	

For the one-to-one maps, we add a map $V_s$ that uses a final vector $v^{(n-1)}$ needed to complete the basis to a basis for the whole $\mathbb{Z}_2^n$. Since we can choose any vector  not orthogonal to $s$, we simply pick a vector with a single bit set out of the bits set in $s$ (which is possible as soon as $s\neq0$). To construct $V_x$, connect a \textit{CNOT} from $x_i$ to $y_{n-1}$ for the single $v^{(n-1)}_i=1$. This gives the map
\begin{equation}
  \ket{x}\ket{f^{(1)}(x)}\mapsto \ket{x}\ket{y \oplus f^{(0)}(x)}
\end{equation}
where 
\begin{equation}
f_k^{(0)}(x)=x\cdot v^{(k)}
\end{equation} 
This equals $f_k^{(1)}$ except for the case $k=n-1$ for which 
\begin{equation}
	f_{n-1}^{(0)}(x\oplus s)=(x\oplus s)\cdot v^{(n-1)}=\overline{f_{n-1}^{(0)}(x)},
\end{equation}
making $f^{(0)}(x)$ a one-to-one function.

The complete construction is shown in \cref{fig:Simon_oracle} and consists of three parts: the $U_s$ map, the $V_s$ map controlled by a bit $b$ that decides if $f$ is one-to-one or two-to-one, and a permutation like the one used in \cref{sec:Deutsch-Jozsa}. 

\begin{figure}
\centering
\includegraphics[scale=1]{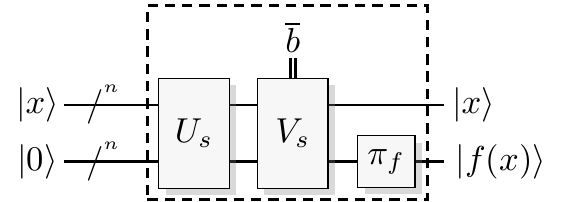}

\captionof{figure}{Oracle construction for \textsc{Simon's} Subroutine.}
\label{fig:Simon_oracle}
\end{figure}

Creating a controlled $V_s$ consists simply of exchanging the single \textit{CNOT} with a single Toffoli with the second inverted control attached to the control bit $b$, giving the map $f^{(1)}(x)\mapsto f^{(b)}(x)$.
To produce all possible one-to-one and two-to-one functions of this type from $f^{(b)}(x)$,  the output needs to be fed into an arbitrary permutation $\pi_f$ (of the number states) which makes $\pi_f(f^{(b)}(x))=f(x)$. Note that if the function is two-to-one, the permutation will preserve the invariance under addition of $s$.
Also note that the use of a generic permutation prohibits a general $\ket{y}$ to be used since $\pi_f\big(y\oplus f'(x)\big)$ does not necessarily equal $y\oplus\pi_f\big(f'(x)\big)$, thus the use of $y=0$ in the construction. 
For a general $y$, see \cref{sec:output_to_the_target_modulo_2}.

The overall implementation of the oracle will give
\begin{equation}
\ket{x}\ket{0}\xrightarrow{U_f}\ket{x}\ket{\pi_f\big(f^{(b)}(x)\big)}=\ket{x}\ket{f(x)}.
\end{equation}

Next we need to find out how this behaves in QSL. Each $\CNOT$ gives a phase kick-back to the controlling QSL-bit, and there is a $\CNOT$ between input QSL-bit $x_i$ and the output QSL-bit $y_k$ for each entry where $v_i^{(k)}=1$. Therefore, for each phase bit $c_k$ set in the output, the corresponding $v^{(k)}$ is added to the phase of the input. The total effect of $\U_s$ is to add
\begin{equation}
g^{(1)}(c)=\sum_{k=0}^{n-2} c_kv^{(k)} \quad (\mathrm{mod}\ 2)
\end{equation}
to the phase of $x$, while $\V_s$ adds $c_{n-1}v^{(n-1)}$ if enabled, creating
\begin{equation}
g^{(0)}(c)=\sum_{k=0}^{n-1} c_kv^{(k)} \quad (\mathrm{mod}\ 2).
\end{equation}
Omitting the control system $(b,a)$ (but including a nonzero target $y$), the effect of the QSL oracle is
\begin{equation}\label{QSL_simons_behavior}
\begin{split}
(x,p)\;&(y,c)\xrightarrow{\U_s}\big(x,a \oplus g^{(1)}(c) \big)\;\big(y\oplus f^{(1)}(x),c\big)\\
\xrightarrow{\CV_s}\;&\big(x,a \oplus g^{(b)}(c)
\big)\;\big(y\oplus f^{(b)}(x),c\big)\\
\xrightarrow{\pi_f}\;&\big(x,a \oplus g^{(b)}(c)\big)\;\Big(\pi_f\big(y\oplus f^{(b)}(x)\big),\pi_{f,y\oplus f^{(b)}(x)}(c)\Big),
\end{split}
\end{equation}
where the phase permutation notation from \cref{eq:phasepermutation} becomes the somewhat cumbersome $\pi_{f,y\oplus f^{(b)}(x)}(c)$.
Using this to simulate the entire subroutine in \cref{fig:Simon's_circuit}, we get
\begin{equation}
\begin{split}
(0,r)&\; (0,r') \xrightarrow{\H\times \I}\ (r,0)\; (0,r')\\
&\xrightarrow{\O_f}\ \big(r,g^{(b)}(r')\big)\; \big(f(r),\pi_{f,f^{(b)}(r)}(r')\big)\\
&\xrightarrow{\H\times \I}\ \big(g^{(b)}(r'),r\big)\; \big(f(r),\pi_{f,f^{(b)}(r)}(r')\big)
\end{split}
\end{equation}
Measuring the query register will yield $g^{(b)}(r')$, that is a linear combination of vectors $v^{(k)}$ with $r'_k$ as coefficients. 
The state preparation in each round of the subroutine creates independent and uniformly distributed random bit vectors $r'$ (and $r$). 
Therefore, if $b=1$ the output will be a random vector uniformly distributed over the subspace orthogonal to $s$, and if $b=0$ the output is a random vector uniformly distributed over the entire $\mathbb{Z}_2^n$. 
Thus, the simulation samples from the same probability distribution as the quantum algorithm, and will also solve the decision problem with the same expected linear number of queries to the oracle as the quantum algorithm. We have now arrived at
\begin{theorem}
There is an efficient QSL algorithm that solves \textsc{Simon's} problem with an expected $\mathcal{O}(n)$ number of queries query to the oracle. 
Relative to the oracle, this QSL algorithm has an efficient classical simulation on a PTM, in accordance with \cref{lemma:efficiency}.
\end{theorem}

This shows that, relative to this oracle, the problem is in \textbf{BPP}.

\subsection{Adding the function output to the target modulo 2} \label{sec:output_to_the_target_modulo_2}

We noted earlier that the oracle in \cref{fig:Simon's_circuit} does not produce a unitary map that adds the function value mod 2 to the output, because of the final number-state permutation. 
To allow a non-zero-initialized answer register so that the map becomes $\ket{x}\ket{y}\mapsto \ket{x}\ket{y\oplus f(x)}$, Bennett's trick of uncomputation \cite{Bennett1973} can be used (see \cref{fig:Simon_oracle2}).
\begin{figure}
\centering
\includegraphics[scale=0.7]{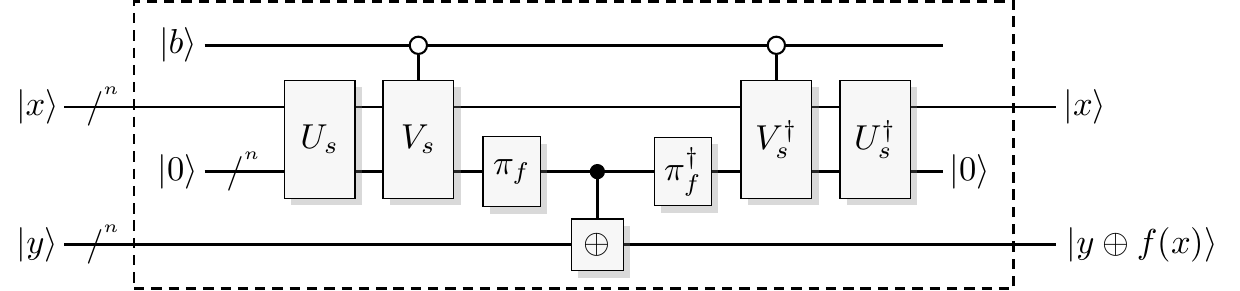}

\captionof{figure}{Oracle construction for Simons algorithm where the output
can be initiated in a non-zero state. The boxed modulo 2 addition denotes an array of \textit{CNOT}s that adds each ancilla bit to the corresponding target bit.}
\label{fig:Simon_oracle2}
\end{figure}	 

In this case the first three gates ($\U_s$, $\CV_s$, and $\pi_f$) performs the map (only writing what happens in the query register and ancilla),
\begin{equation}
\begin{split}
	(x,p)\;(0,r)
	&\rightarrow\big(x,p\oplus g^{(b)}(r)\big)\;\Big(f(x),\pi_{f,f^{(b)}(x)}(r)\Big).
\end{split}
\end{equation}
The $\CNOT$ array will add the value $f(x)$ to the target register  $(y,c)$ modulo 2 and cause a phase kickback from the target register,
\begin{equation}
\begin{split}
\big(x,p&\oplus g^{(b)}(r)\big)\;\Big(f(x),\pi_{f,f^{(b)}(x)}(r)\Big)\\
&\rightarrow
\big(x,p\oplus g^{(b)}(r)\big)\;\Big(f(x),\pi_{f,f^{(b)}(x)}(r)\oplus c\Big).
\end{split}
\end{equation}
Finally, letting $r'$ equal the output of the inverted permutation, 
\begin{equation}
	r'=\pi_{f,f^{(b)}(r)}^{-1}\big(\pi_{f,f^{(b)}(x)}(r)\oplus c\big),
\end{equation} 
the final three gates ($\pi_f^\dagger$, $\CV_s^\dagger$, and $\U_s^\dagger$) will invert the function map and cause another addition of $g^{(b)}$ in the phase of the input system,
\begin{equation}
\begin{split}
\big(x,p&\oplus g^{(b)}(r)\big)\;\big(f(x),\pi_{f,f^{(b)}(x)}(r)\oplus c\big)\\
&\rightarrow
\big(x,p\oplus g^{(b)}(r)\oplus g^{(b)}(r')\big)\;(0,r').
\end{split}
\end{equation}
When used in \textsc{Simon's} algorithm, $p=0$, and $r$ and $c$ are independent uniformly distibuted random values, which make $r'$ independent of $r$ and uniformly distributed. The output of the algorithm is the phase of the query register, which equals 
\begin{equation}
	g^{(b)}(r) \oplus g^{(b)}(r')
\end{equation}
As the two arguments to $g$ are independent and randomly distributed, the sample retrieved from the measurement will be drawn from the uniform distribution over all bit strings if $b=0$, or the strings orthogonal to $s$ if $b=1$. 
This gives the same output distribution as the quantum algorithm, and the theorem holds in this case as well.

\subsection{A deterministic algorithm for \textsc{Simon's} problem} 

An interesting extension of this algorithm is the one devised by \textcite{Brassard1997}, which relative to an oracle decides the problem in exact quantum polynomial time \textbf{EQP}, meaning that the number of calls to the oracle is deterministic and polynomial, and that the error probability is zero.
This algorithm was considered relativized evidence for an exponential separation between \textbf{EQP} and \textbf{BPP}. 
In \cite{Johansson2017} we developed a similar algorithm for the case when the oracle is of the form used in \cref{sec:SimonsQSL}, or in general, when the bias enters into the function evaluation as
\begin{equation}
	\ket{x}\ket{y}\mapsto\ket{x}\ket{f_y(x)}=\ket{x}\ket{\pi_f\Big(y\oplus U_0\big(x\oplus (x\cdot v_s)s\big)\Big)},
	\label{eq:simonsotherbias}
\end{equation}
where $U_0$ is a binary orthogonal map (i.e., one-to-one, and preserves scalar products), $v_s$ is a vector with $v_s\cdot v_s=v_s\cdot s=1$ unless $s=0$ in which case $v_s=0$, and $\pi_f$ is a permutation as before. Note that the construction enables all $f(x)=f_0(x)$, and that 
\begin{equation}
	(x\oplus s)\oplus \big((x\oplus s)\cdot v_s\big)s=x\oplus (x\cdot v_s\big)s,
\end{equation}
so that $f_y(x\oplus s)=f_y(x)$. The difference to the standard construction is that bias enters differently, since the standard construction gives $f'_y(x)=y\oplus f(x)$. 

\begin{figure}
	\centering
	\includegraphics[scale=1]{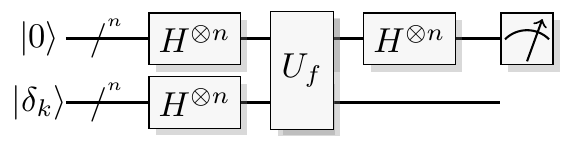}
	\captionof{figure}{Proposed quantum circuit as subroutine in an algorithm solving \textsc{Simon's} problem in polynomial time, both in quantum theory, and in QSL.}
	\label{fig:Simon_circuit_P}
\end{figure}	
 
The algorithm, as compared to the standard algorithm, adds a Walsh-Hadamard transform on the answer register input, and then simply iterates through bit-strings $\delta_k$ that have only the $k$th bit set on the answer register input.
A quantum oracle that obeys \cref{eq:simonsotherbias} will give the following effect in the algorithm, 
\begin{equation}
\begin{split}
	&\ket0\ket{p}\xrightarrow{H^{\otimes n}\otimes H^{\otimes n}}\;\sum_{x,y}(-1)^{y\cdot p}\ket{x}\ket{y}\\
	&\quad\xrightarrow{y\oplus U_0(\ldots)}\sum_{x,y}(-1)^{y\cdot p}\ket{x}\ket{y\oplus U_0\big(x\oplus(x\cdot v_s)s\big)}\\
	&\quad=\sum_{x,y}(-1)^{\big(y\oplus U_0(x\oplus(x\cdot v_s)s)\big)\cdot p}\ket{x}\ket{y}\\
	&\quad=\sum_{x,y}(-1)^{x\cdot \big(U_0^{-1}p\oplus(s\cdot U_0^{-1}p)v_s\big)\oplus(y\cdot p)}\ket{x}\ket{y}\\
	&\quad\xrightarrow{H^{\otimes n}\otimes \pi_f}\ket{U_0^{-1}p\oplus(s\cdot U_0^{-1}p)v_s}\sum_{y}(-1)^{y\cdot p}\ket{\pi_f(y)}.
	\end{split}	
\end{equation}
A measurement on the query register will give the output $U_0^{-1}p\oplus(s\cdot U_0^{-1}p)v_s$, which is orthogonal to $s$ because 
\begin{equation}
	s\cdot \big(U_0^{-1}p\oplus(s\cdot U_0^{-1}p)v_s\big)=
	(s\cdot U_0^{-1}p)\oplus(s\cdot U_0^{-1}p)(s\cdot v_s)=0,
\end{equation}
and the output from letting $p$ iterate over a basis for bitstrings of length $n$ will span the subspace orthogonal to $s$. 
The important difference to the previous analysis is that the present form preserves orthogonality in the phase kick-back, whereas there is no such guarantee needed in \cref{sec:output_to_the_target_modulo_2}.
It is now easy to check if the output subspace has dimension $n$ (the $n$ outputs are linearly independent) or $n-1$. The quantum algorithm solves the problem with error probability zero in linear time, so relative to this oracle, this problem is in \textbf{EQP}.

The behavior will be exactly the same for the QSL simulation, but here we only analyze the explicit choices 
\begin{equation}
	U_0=[v^{(1)}, v^{(2)},\ldots,v^{(n)}]^T
\end{equation}
and if $s\neq0$,
\begin{equation}
	v_s=v^{(n)}.
\end{equation}
This gives us the explicit construction of \cref{sec:SimonsQSL}, for which
\begin{equation}
\begin{split}
	&(0,x)\;(\delta_k,y)\xrightarrow{\H^{\times n}\times \H^{\times n}}(x,0)\;(y,\delta_k)\\
	&\;\;\xrightarrow{\U_s\circ\CV_s}\big(x,v^{(k)}\big)\;\big(y\oplus f^{(b)}(x),\delta_k\big)\\
	&\;\;\xrightarrow{\H^{\times n}\times \pi_f}\big(v^{(k)},x\big)\;\Big(\pi_f\big(y\oplus f^{(b)}(x)\big),\pi_{f,y\oplus f^{(b)}(x)}(\delta_k)\Big),
\end{split}	
\end{equation}
for $k\le n-2$ and 
\begin{equation}
\begin{split}
	&(0,x)\;(\delta_{n-1},y)\xrightarrow{\H^{\times n}\times \H^{\times n}}(x,0)\;(y,\delta_{n-1})\\
	&\quad\xrightarrow{\U_s\circ\CV_s}(x,\overline{b}v^{(n-1)})\;\big(y\oplus f^{(b)}(x),\delta_{n-1}\big)\\
	&\quad\xrightarrow{\H^{\times n}\times \pi_f}(\overline{b}v^{(n-1)},x)\;
	\\&\qquad\qquad\times
	\Big(\pi_f\big(y\oplus f^{(b)}(x)\big),\pi_{f,y\oplus f^{(b)}(x)}(\delta_{n-1})\Big).
\end{split}	
\end{equation}
Again, the iteration  will produce measurement results from the query register that span the subspace of bitstrings $v$ that obey $v\cdot bs=0$. 
The output is exactly the same as in the quantum case.
Thus, relative to this QSL oracle, this problem is in \textbf{P}. 
By modifying the effect of bias on the target register, we have moved from a problem in a probabilistic class to a problem in the deterministic class of the same complexity.

\subsection{Application to Symmetric Ciphers}

\textsc{Simon's} algorithm has been used to attack symmetric-key ciphers, retrieving the key in a linear number of calls to the encryption device.
There are a number of results that lead up to this.

\textcite{Kuwakado2010} showed that, if built as a unitary, a 3-round Feistel
network with internal permutations can be distinguished from a random
permutation using \textsc{Simon's} algorithm. If we look at the 3-round Feistel
network as a block cipher running in the Electronic Code Book (ECB) mode, it
permutes a block of the plaintext into a block of the ciphertext. If this
permutation is indistinguishable from a random permutation, and therefore
unpredictable, the best an adversary running a known-plain-text attack can do,
is to search the whole key space. However, if the permutation can be
distinguished from a random permutation, then we get knowledge about the
system and the key space that we need to search shrinks. The result
of \textcite{Kuwakado2010} show us that after the application of \textsc{Simon's}
algorithm there should be a speed-up of the key recovery algorithm, but not how large the speed-up is.

\textcite{Kuwakado2012} later showed that \textsc{Simon's} algorithm can be used
to retrieve the key from the Evan-Mansour cipher by employing a number of
queries that are linear in the key size. This gives an exponential speed-up to
the best-known classical query bound and a significant speed-up compared to
the method using \textsc{Grover's} algorithm described in \cref{sec:Grover's}.

\textcite{Kaplan2016} extended this to other ciphers, and modes of operations,
and showed that it can give an exponential speed-up some
classical attacks that require finding collisions.

As always, for this to work, the encrypting device needs to be re-cast as a unitary operation and the key needs to be stored in the device.
Note that here, the key is a classical parameter to the encryption device, and not under control of an attacker or entered as a quantum state in contrast to the known-plaintext-attack using \textsc{Grover's} algorithm in \cref{sec:GroverCiphers}.
From a practical point of view, this means that the key or the secret needs to be built into the construction, just as in our previous realization of \textsc{Simon's} oracle, where $s$ is used to build an explicit gate array.
Furthermore, an attacker is able to retrieve the built-in secret because it is built into the dynamics of a physical system. 
Again, when comparing with the previous realization, it is clear that the secret leaks through another degree of freedom in the physical system, and not through the computational basis that holds the cleartext or ciphertext.

The conclusion is that this is a side-channel attack, i.e., an attack made available because of a specific implementation --- not an attack on the protocol itself. 
Cryptographic protocols that are vulnerable to side-channel attacks are in general not considered broken, since all are more or less vulnerable depending on how they are implemented. 

With that said, the results of \cite{Kuwakado2010,Kuwakado2012,Kaplan2016} are important examples of new and innovative ways of performing side-channel attacks, in this case using a quantum-computer implementation of the cipher, and having the owner of the secret hide it as a parameter within the quantum computer.
We agree with the recommendation of \cite{Kuwakado2012}, that classical ciphers should not be implemented in quantum circuits, because of potentially introducing side-channel attacks.
A large part of modern cryptography is to prevent the creation of unintended side channels in new encryption devices.
It is still an open question whether QSL enables a similar side-channel attack on these kinds of systems, but for the same reason as with quantum computation, we would still advise against using QSL systems for encryption or decryption purposes.

\section{Shor's Algorithm Factoring 15}\label{sec:Shor's}

Shor's algorithm \cite{Shor1994,Shor1999} is one of the quantum algorithms that 
solve a computational problem with real-world applications: to efficiently 
find a factor in a composite number $N$. We have seen that 
Deutsch-Jozsa's and Simon's algorithms can be run in polynomial time on a 
classical Turing machine; the natural continuation is to 
attempt integer factorization, whose hardness is one of the most widely 
believed conjectures in computer science.

While large high-fidelity quantum computers are still far away, several
experimental realizations of Shor's algorithm for small numbers have been
presented~\cite{Monz2016,Lucero2012,Martin-Lopez2012,Politi2009,Lu2007,Lanyon2007,Vandersypen2001}.
These are all very impressive demonstration of quantum optimal control, but an
experimental realization of Shor's algorithm with the currently available 
technology is demanding and this has led to the need for
vast simplifications in the algorithm. There are essentially two parameters 
subject to optimization: bit-depth and 
circuit-depth (see Ref.\cite{Markov2012}\ and the citations therein). Also, the 
approximate quantum Fourier transform \cite{Coppersmith2002} is crucial for 
scalability. In this optimization procedure, one has to
be careful not to over-simplify, or make explicit (or implicit) use of 
knowledge about the solution~\cite{Smolin2013}. We will avoid such over-simplifications by using an identical circuit for factoring 15 as in Ref.~\cite{Monz2016}.

Shor's algorithm uses a quantum subroutine that finds the order 
(or period) of an element $a$ in the multiplicative group of integers modulo 
$N$. Here, the order is the smallest integer $r$ such that $a^r=1$ (mod $N$). 
This is sufficient information to find a factor in $N$. The algorithm makes use 
of an input-register quantum state, containing an integer $x$, and an 
output-register in which modular exponentiation $a^x$ (mod $N$) is computed. By 
creating a superposition in the input-register using the quantum Fourier 
transform, performing the calculation, and then inverting the transform, one 
can with high probability retrieve sufficient information to calculate the 
order (see \cref{fig:ShorA}). 

\begin{figure}[t]
	\centering
  \includegraphics[scale=1]{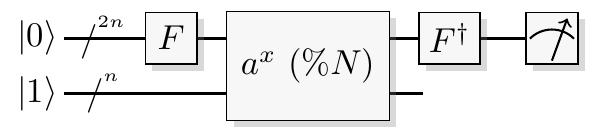}
	\captionof{figure}{Circuit diagram of the quantum subroutine used in Shor's 
		algorithm. A $2n$-qubit register is initiated in the zero-state $\ket{0}$, 
		and an $n$-qubit register in $\ket{1}$. Basis change of the input-register part 
		of the controlled modular exponentiation operator allow for sampling a 
		probability distribution with peaks at $s/r$.}
	\label{fig:ShorA}
\end{figure}

More specifically, this procedure let us 
sample from a probability distribution with peaks at $s/r$, where $s$ is 
uniformly distributed over the integers between $0$ and $r-1$.  
Ideally, the peaks are completely localized to $s/r$ but in most cases there is 
peak broadening due to Fourier leakage, and to ensure that the measurement 
yields a binary fraction sufficiently close to $s/r$, the input-register needs 
to be at least $2n$ qubits in size, where $n$ is the size of the 
output-register which is large enough to perform calculations mod $N$. 
The full procedure to retrieve $r$ is as follows: 

\begin{tcolorbox}
	\begin{algorithm}[\textcite{Shor1994}]
		Proceed with the following steps to retrieve a factor of a composite number with high probability.
		\setlength{\leftmargini}{1.5em}
		\begin{enumerate}
		    \item Pick at random an integer $a\neq \pm1$ modulo $N$. 
				If GCD$(a,N)$ is a nontrivial factor of $N$, we have a solution.
				\item Otherwise generate, setup, and run the quantum \Cref{sub:Shor} (Also shown in 
				\cref{fig:ShorA}) to find a candidate for $s/r$.
				\item Use the continued fraction expansion to retrieve $r$ (or a
				factor in $r$ when $s$ and $r$ has a common factor). 
				\item If $r$ is even, one of GCD$(a^{r/2}\pm 1, N)$ may be a nontrivial 
				factor of $N$. This happens with high probability.
		\end{enumerate}		
	\end{algorithm}
\end{tcolorbox}

\begin{tcolorbox}
	\begin{subroutine}[\textcite{Shor1994}]\label[sub]{sub:Shor} Proceed with the following steps to find a candidate for the quotient $s/r$ where $s$ is random and $r$ is the period of $a^x$ mod $N$.
		\begin{enumerate}
			\item Prepare a $2n$-qubit query register in the state $\ket{0}$ and an $n$-qubit answer register in the state $\ket{1}$. 
			
			\item Apply the quantum Fourier transform to the query register.
			
			\item Apply the specific unitary $U_{a,N}\ket{x}\ket{y}=\ket{x}\ket{y\times a^x \text{ mod }N}$.
			
			\item Apply the inverse quantum Fourier transform to the query register. 
			
			\item Measure the query register to retrieve a candidate for $s/r$.
			\end{enumerate}
			Return that candidate for $s/r$.
	\end{subroutine}
\end{tcolorbox}

For $N=15$ the possible integers that can occur in steps 2-4 are 
$a\in\{2,4,7,8,11,13\}$ (see \cref{fig:ShorC}). 
Note that deliberately choosing $a$ to give a short period that is easy to find, would be another example of over-simplifying the algorithm. 
It is therefore important that the element $a$ is 
chosen randomly~\cite[see e.g.,][]{Smolin2013}. In what follows, we have used 
all alternatives; this is of course only possible because of the small $N$ used.

\begin{figure}[t]
	\centering
	\begin{tikzpicture}[anchor=north west,rounded corners]
	\node (1) [inner sep=-1mm]{\includegraphics[scale=.65]{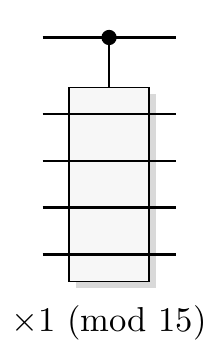}};
	\node at (1.north east) (2)[inner sep=-1mm] {\includegraphics[scale=.65]{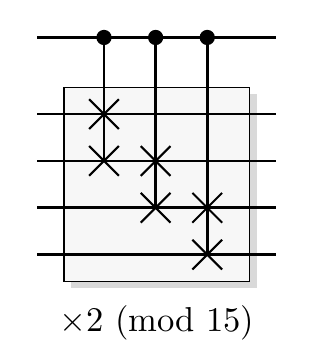}};
	\node at (2.north east) (4)[inner sep=-1mm] {\includegraphics[scale=.65]{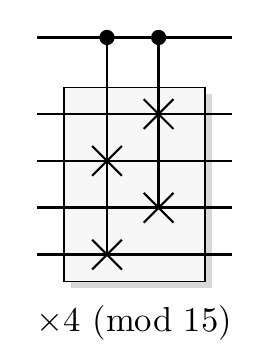}};
	\node at (4.north east) (7)[inner sep=-1mm] {\includegraphics[scale=.65]{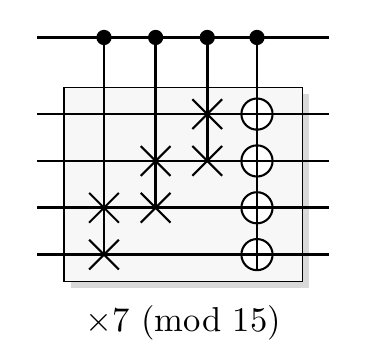}};
	\node at (1.south west) (8)[inner sep=-1mm] {\includegraphics[scale=.65]{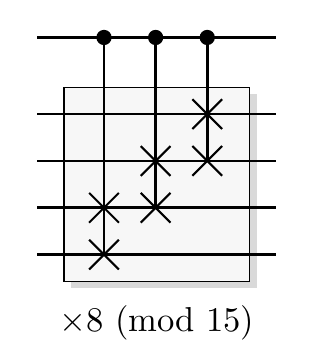}};
	\node at (8.north east) (11)[inner sep=-1mm] {\includegraphics[scale=.65]{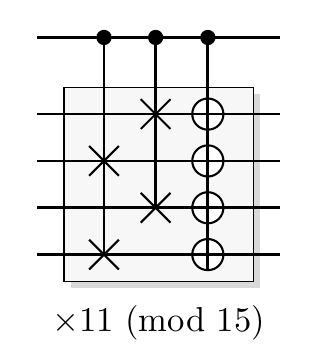}};
	\node at (11.north east) (13)[inner sep=-1mm] {\includegraphics[scale=.65]{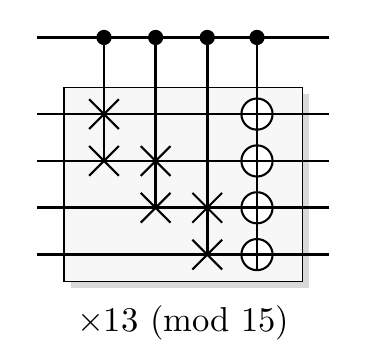}};%
	\end{tikzpicture}
	\captionof{figure}{Controlled modular multipliers that occur in Shor's algorithm.}
	\label{fig:ShorC}
\end{figure}

Some useful simplifications \textit{are} allowed. 
The Fourier transform on the input can be exchanged for the Walsh-Hadamard transform (Hadamard gates), while the inverse Fourier transform can be implemented through 	Hadamards followed by classically controlled single qubit rotations \cite{Griffiths1996}; by advancing the measurement of the controlling qubit and using the outcome as a classical control (see \cref{fig:ShorB}). 
This decouples the $2n$ qubits of the input-register in the 
sense that the procedure of preparation, transformation, and measurement can be 
performed individually on each qubit. It is common to perform these single 
qubit procedures in sequence on one single qubit, a method known as qubit 
recycling, which reduces the overall bit-depth from $3n$ to $n+1$ at the cost 
of circuit-depth.

\begin{figure}[t]
	\centering
  \includegraphics[width=\linewidth]{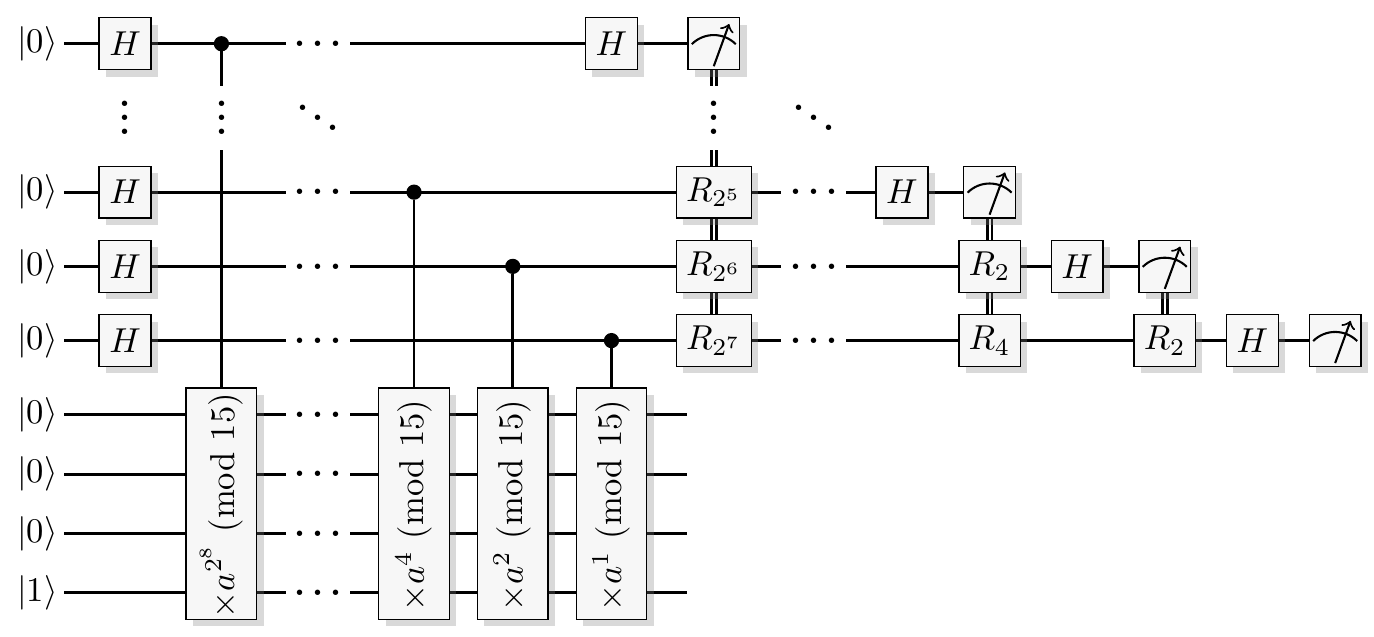}
	\captionof{figure}{Shor's algorithm with semiclassical inverse Fourier 
	transform.  Note that 
    $\times2^2\equiv \times7^2 \equiv \times8^2 \equiv \times13^2 
    \equiv \times4^1$ and $\times4^2\equiv\times11^2\equiv\times1$ (mod $15$), 
    so that many of the controlled multiplications will be identities. 
    Therefore, most rotations $R_{2^k}$ will never be 
	applied (in the ideal situation), in fact only the very last $R_2$ 
	operation can ever occur.}
	\label{fig:ShorB}
\end{figure}

The demonstration of Monz \textit{et al}~\cite{Monz2016} is the most advanced to
date. They refrain from ``compiling'' the circuitry, and use non-Clifford group
operations which is necessary for a scalable version. One
simplification Monz \textit{et al} do use is restricting the resolution of the
input-register from $2n$ qubits to three qubits. This is possible because all
elements in the multiplicative group mod 15 have power-of-two periods, and they
only verify the behavior of the exponentiation until and including the first
that is equivalent to the identity map (see \cref{fig:ShorC}).

It has been argued~\cite{Cao2015} that the three-bit precision of Monz 
\textit{et al} is insufficient since $2n$ bits are required for the algorithm 
to overcome Fourier leakage in general, and to succeed with a bounded error rate 
required for scalability. However, for $N=15$ there is no Fourier leakage 
because of the power-of-two periods and this is clear when building the 
quantum gate array. Therefore, measuring more qubits will only add noise and not 
precision. The distribution is completely described at two bits of precision. 
On the other hand, this also means that the process of generating the circuitry 
solves the factoring problem, since it is enough to know at which point the 
identity emerges from exponentiation. Thus, the experiment \cite{Monz2016} is 
not so much about factoring 15 through \textsc{Shor's} algorithm, but more a 
verification that their quantum (ion-trap) realization is good enough to 
approximate the ideal quantum statistics even when not using a compiled circuit.

We now present an experimental realization of Shor's algorithm, factoring 15, using QSL. 
Please note that the motivation for doing this is not to factor 15, but to 
compare our completely classical construction with the current state-of-the-art quantum implementations, and to show that QSL is applicable outside the oracle paradigm. 
We use basically 
the same algorithmic setup, but employ the semiclassical Fourier transform 
without qubit recycling, and the emerging identity
operations are of course all omitted. Our setup is similar to that of Monz 
\emph{et al.} (see \cref{fig:ShorB}) and uses the multiplication 
operators from Markov et al. \cite{Markov2012}, but avoids precomputing their 
effect on the initial state (see \cref{fig:ShorC}).

\begin{figure}[]
    \begin{center}
        \includegraphics[width=1\linewidth]{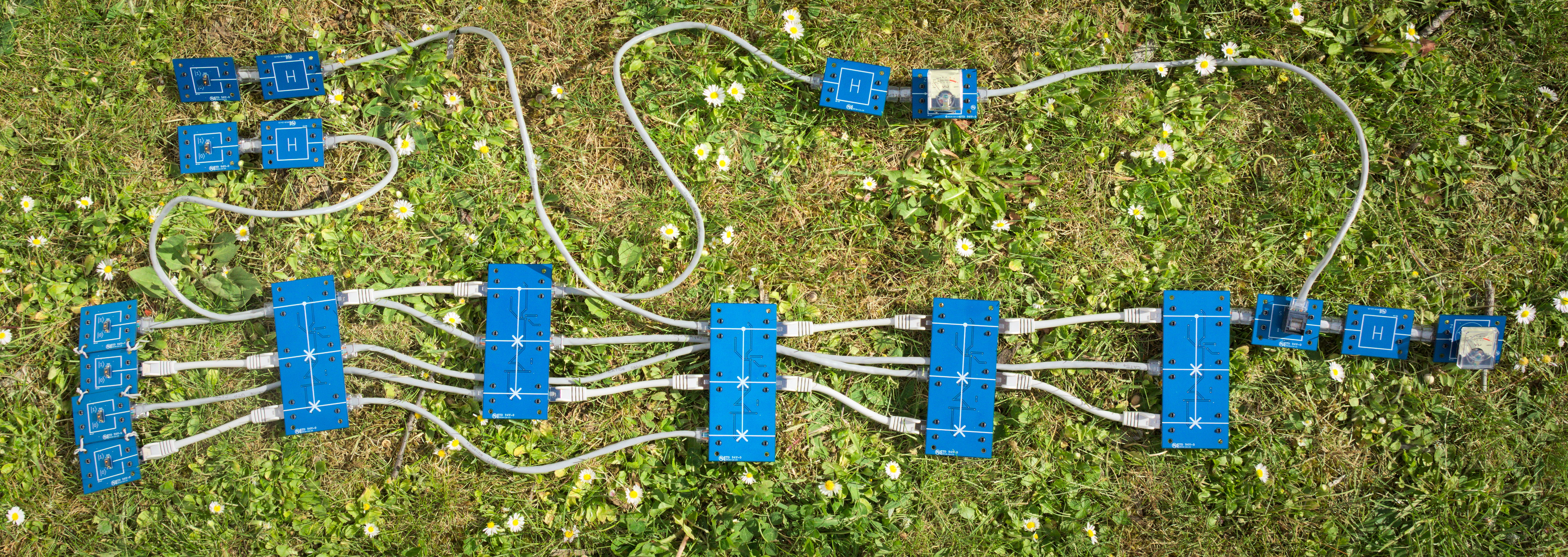}%
    \end{center}
    \captionof{figure}{{QSL realization of Shor's algorithm.} For the case $a=8$ 
    so that the modular multipliers used are, 
    $\times8^2\equiv\times4$ (mod 15), and $\times8$ (mod 15).}
    \label{fig:QSL_lawn}
\end{figure}

\begin{figure}[]
	\begin{center}
		\includegraphics[scale=.39,clip,trim=.5cm 0cm 12.15cm 0cm]{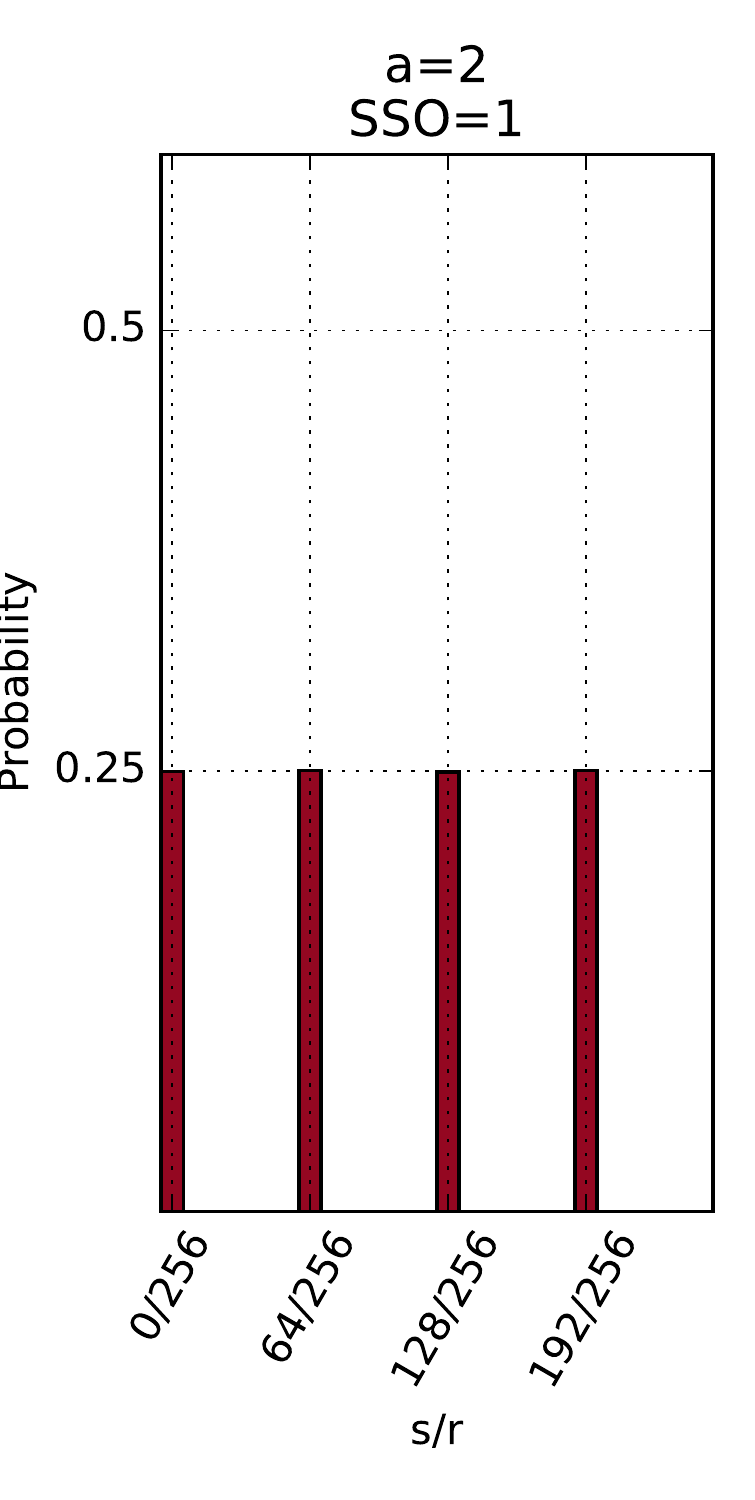}
		\hspace{0cm}\includegraphics[scale=.39,clip,trim=0cm 0cm 12.55cm 
		0cm]{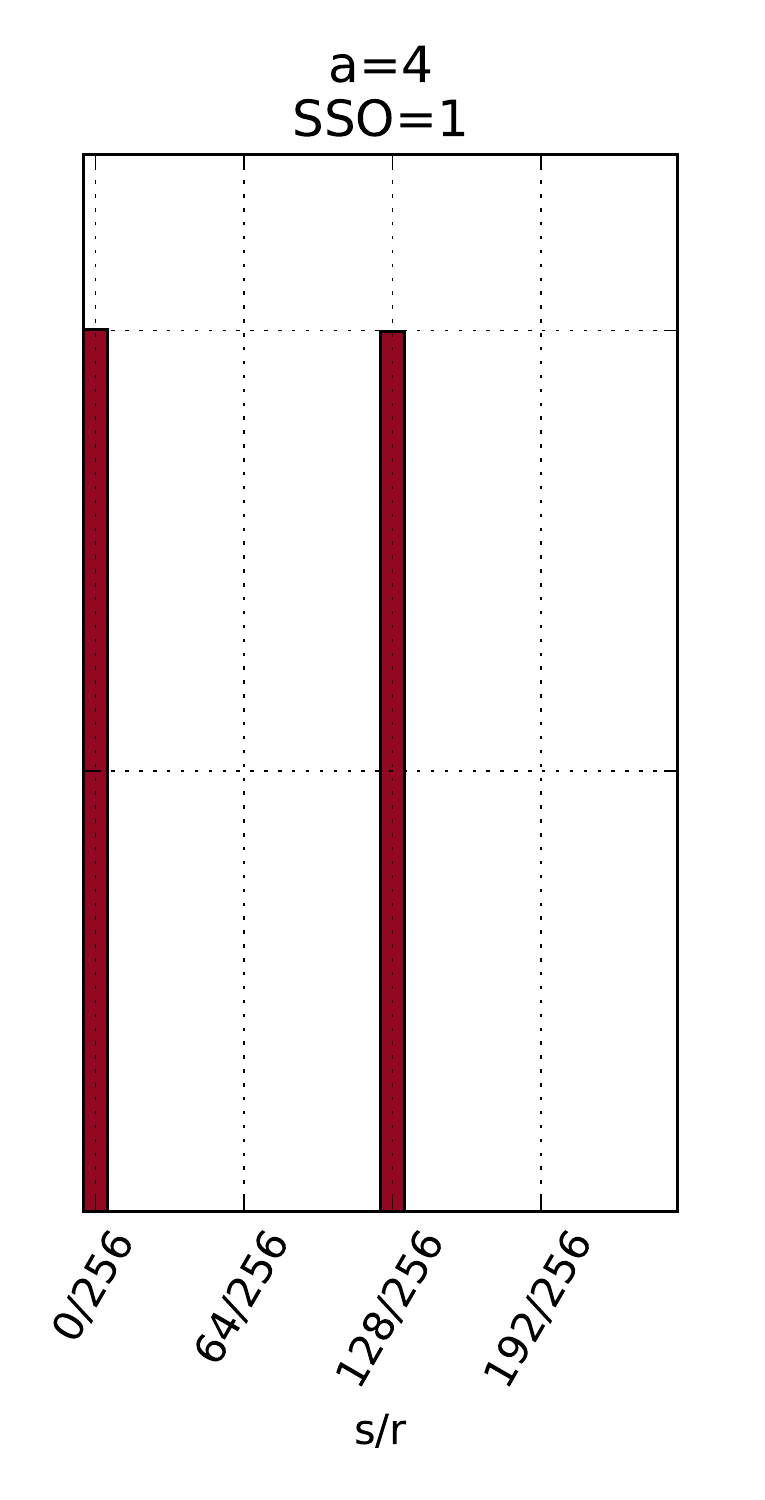}
		\hspace{0cm}\includegraphics[scale=.39,clip,trim=.5cm 0cm 12.55cm 
		0cm]{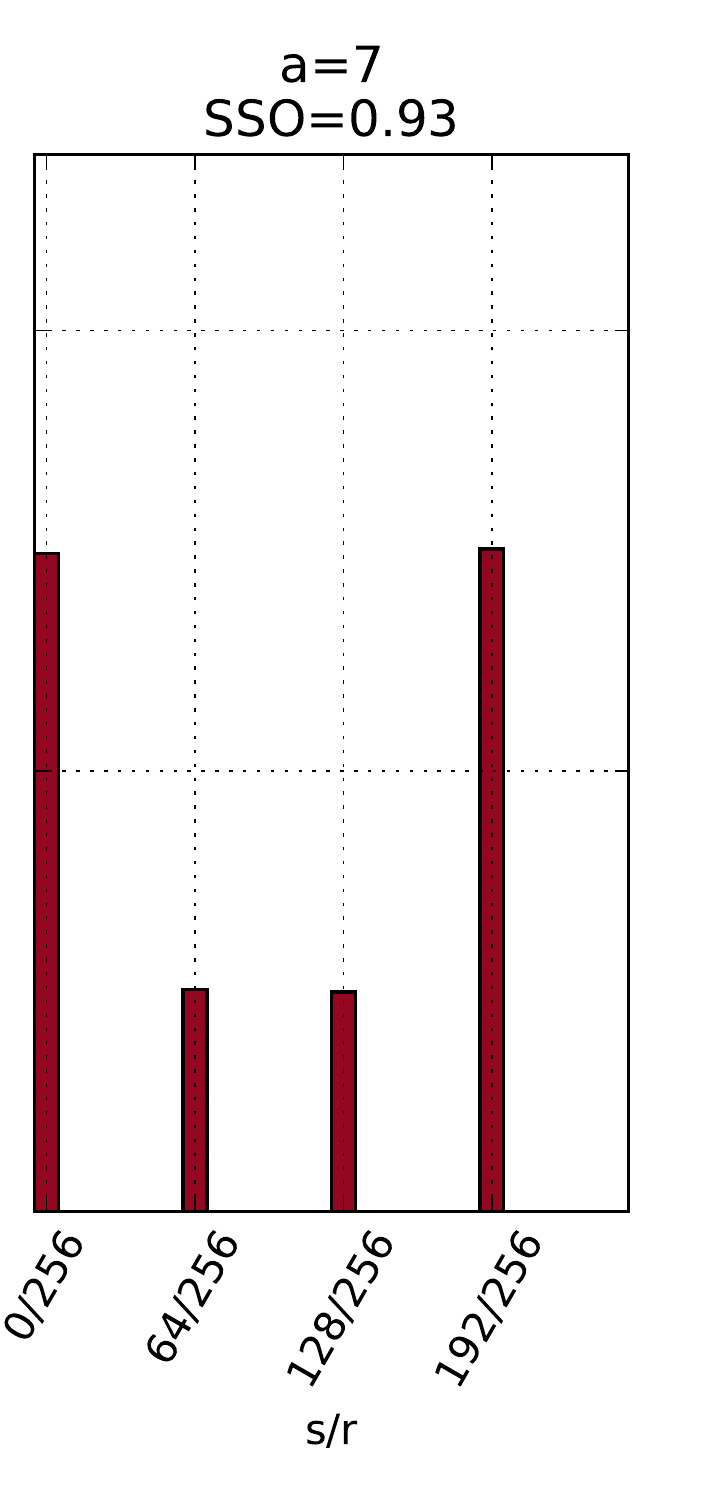}
		\hspace*{.0cm}\includegraphics[scale=.39,clip,trim=.5cm 0cm 12.55cm 
		0cm]{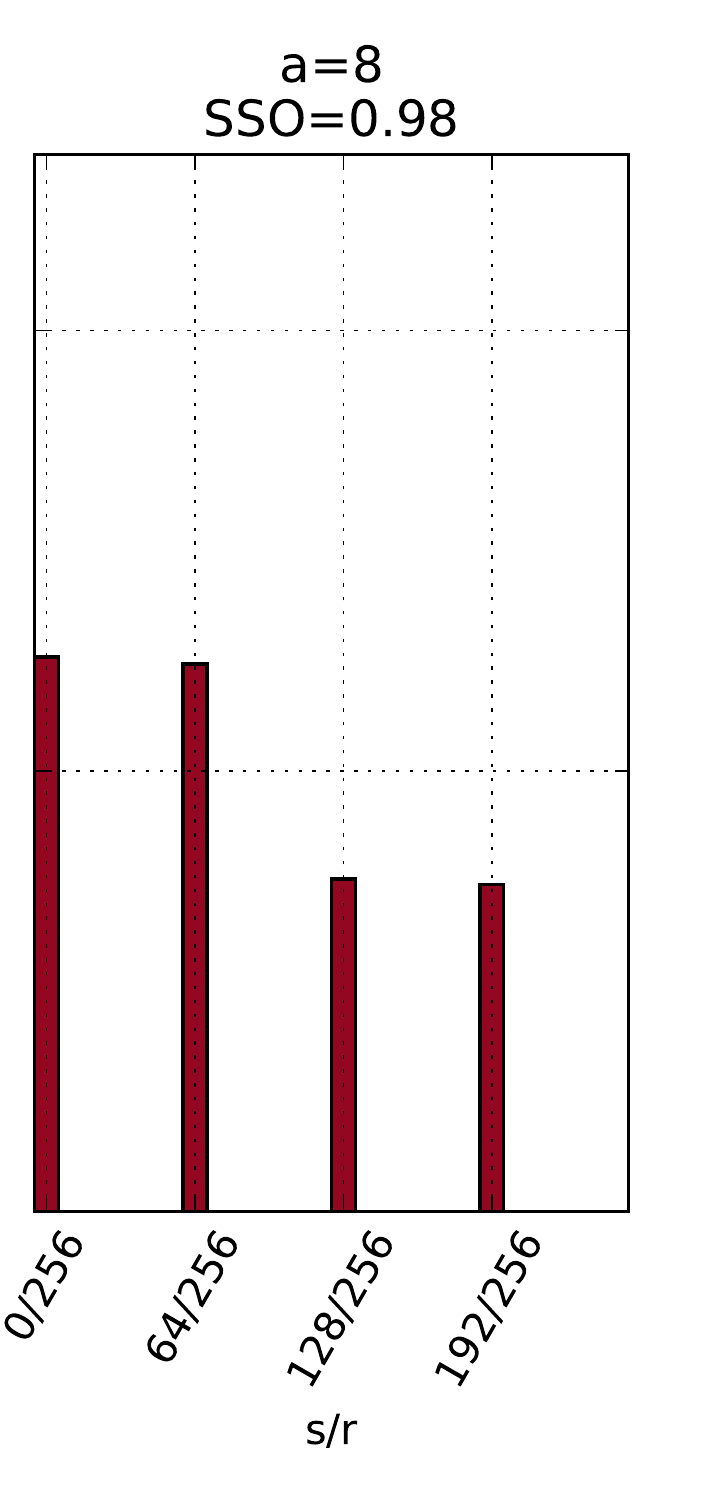}
		\hspace{-0cm}\includegraphics[scale=.39,clip,trim=.5cm 0cm 12.55cm 
		0cm]{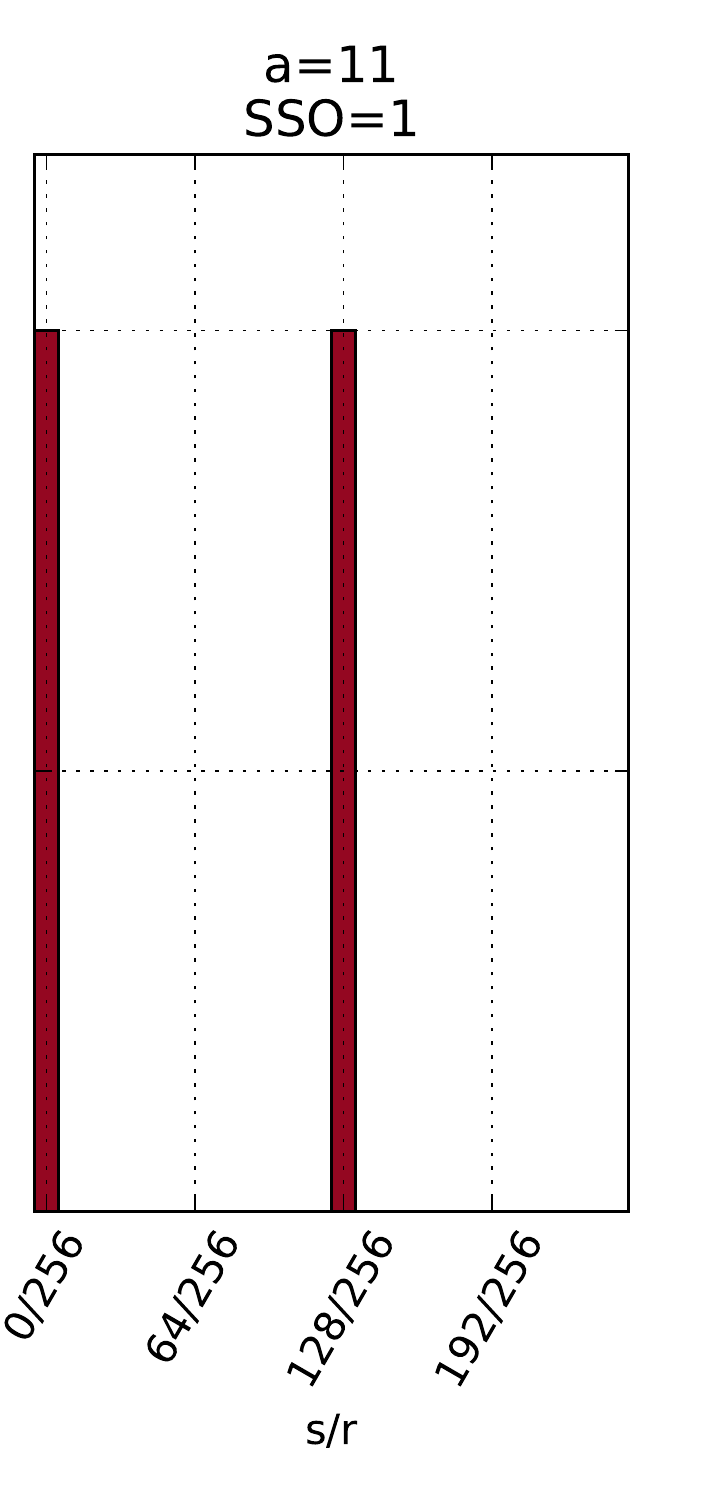}
		\hspace{-0cm}\includegraphics[scale=.39,clip,trim=.5cm 0cm 12.55cm 
		0cm]{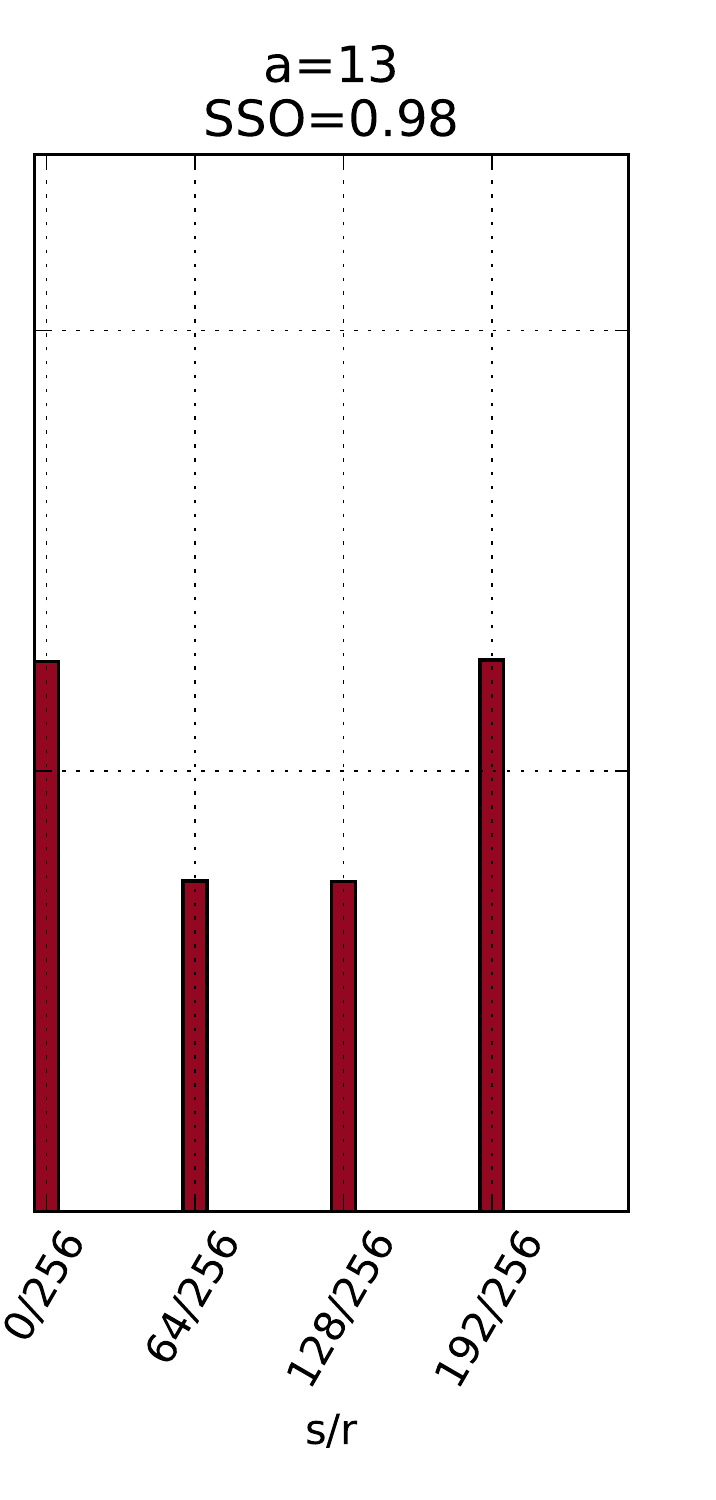}
	\end{center}
	\captionof{figure}{{Estimated output probability.} Point estimates of the 
	output probability distributions of the subroutine, for the non-trivial 
	elements in the multiplicative group of integers mod 15. Each 
	distribution is estimated from $10^6$ samples.}
	\label{fig:distribution}
\end{figure}

We have chosen a hardware 
implementation in 2-complementary reversible pass-transistor
logic, specifically using transmission gates~\cite{Caroe2012}. These are
built using currently available semiconductor technology and operate in the digital regime using TTL levels. The random number generators of the source and measurement devices use high-gain analog amplifiers to sample random bits from noise at a reverse-biased PN junction.

The output probability distributions are estimated from $10^6$ samples for each
element $a\in\{2,4,7,8,11,13\}$ (see \cref{fig:distribution}). We see that
the distributions for $a=7,8,13$ are not uniform over their support, as
predicted by quantum theory, but they are closer to the ideal distributions than any
other implementation that we know of. To analyze their performance, Monz et al.
used the \emph{square statistical overlap} (SSO) as a fidelity measure. Our
implementation reaches an SSO of $\{$0.9999(1), 0.9999(1), 0.933(3), 0.984(2),
0.9999(1), 0.984(2)$\}$ for $a\in\{2,4,7,8,11,13\}$ respectively (statistical
errors as one standard deviation). Notably, our implementation gives the same
probability (0.5) of returning a good candidate for $r$ as the ideal quantum
subroutine.

Like any other physical implementation of quantum gates, our 
simulation suffers from systematic errors. This will result in an error 
propagation that suppresses the amount of useful information that we can 
retrieve when scaling the algorithm to larger numbers --- even though there is 
no practical restriction for doing so. Further work is needed, beyond the scope of this paper, to reduce these systematic errors, perhaps by altering the framework or using error correcting techniques, so that it becomes 
useful for larger instances.

\section{Conclusion} \label{sec:Conclusion}

In this paper, we have looked at the QSL simulation framework, which is
efficiently simulatable on a classical Turing machine and can at the same time
reproduce many quantum phenomena. These phenomena include results from quantum
query complexity in which the query complexity is exactly the same in QSL as in
quantum theory. In the circuit model there is only a constant overhead in
classical resources, which in turn gives a polynomial time simulation in a
classical probabilistic Turing machine.

Using this approach there is no quantum advantage with respect to the QSL
implementation for the \textsc{Deutsch-Jozsa}, \textsc{Bernstein-Vazirani}, or
\textsc{Simon's} problem. For \textsc{Grover} search QSL has a slight
advantage over the classical case, reaching $\O(2^n/n)$ queries compared to $\O(2^n)$ in the classical
case, while the quantum algorithm retains an advantage using $\O(2^{n/2})$ queries. Outside the oracle paradigm, the choices made in the QSL construction gives less of an advantage, but does enable small examples, which we have
illustrated by simulating the setup of two recent experimental realizations of
\textsc{Shor's} and \textsc{Grover's} algorithms.

Regarding \textsc{Deutsch-Jozsa}, \textsc{Bernstein-Vazirani}, and
\textsc{Simon's} problem we can also conclude that no genuinely quantum
resources are required, i.e., that the resources used can be efficiently
simulated on a classical Turing machine. The reason is that the resources needed
are also available in QSL. The inclusion of phase information in QSL gives rise
to phenomena that have previously been suggested resources for the quantum
advantage, including superposition, interference, and some aspects of
entanglement among others. We believe that it is more appropriate to assign the
speed-up in these cases to the ability to store, process and retrieve
information in another information-carrying degree-of-freedom, in this case the 
phase degree-of-freedom.

Functions in quantum computation are realized as dynamics from physical systems
to physical systems, and not as mere function maps from bits to bits. Tracing
the full dynamics of such a physical system may require keeping track of an
exponential amount of information. The available efforts to efficiently trace
the dynamics used in quantum computation usually adopt a top-down approach,
restricting the available dynamics until the necessary information can be
efficiently traced. For instance, the stabilizer subtheory is the restriction to
Pauli measurements and Clifford gates, and this can be traced efficiently. The
approach presented here instead uses a bottom-up approach, that extends the
standard function maps to also encompass the phase information within each
individual qubit. Even though simulating each qubit with a single classical bit
is not enough, we have found that two classical bits per qubit gets us very far.

This clearly shows that comparing bounds from classical query complexity with
those from quantum query complexity does not provide conclusive evidence for
a quantum advantage. A quantum oracle is so vastly different from a classical
oracle that it is questionable whether they are comparable at all. The
comparison between quantum oracle and QSL oracle is much more well-founded,
since both enable a choice between retrieving function value or additional
information related to the function, stored in the phase degree of freedom. All
it takes to enable the same behavior for a QSL oracle as a quantum oracle, in
the above examples, is a single bit of extra (phase) information, i.e., a
constant overhead. This leaves a lot of headroom for improvement, as allowing
for a polynomial overhead will enable a more accurate simulation.

In conclusion, the enabling root cause (resource) for the quantum speed-up in
the mentioned examples is not superposition, interference, entanglement,
contextuality, the continuity of quantum state space, or quantum coherence. It
is the ability to store, process and retrieve information in an additional
information-carrying degree-of-freedom in the physical system being used as
information carrier.

\section*{Acknowledgements}
We acknowledge Markus Grassl, Peter Jonsson, and Matthias Kleinmann for valuable discussions. 

\section*{References}\vspace{-.6cm}
\printbibliography[heading=subbibliography, title={\null}]

\begin{appendix}
\section{Constant and Balanced Functions for Three Bits of Input}\label{sec:B}
Circuit representation of all constant and balanced functions for 3-bit input. 
The values below each gate array are the function values concatenated in the order $f(7)f(6)f(5)f(4)f(3)f(2)f(1)f(0)$.
\subsection*{Toffolis: 0, CNOTS: 0}
\begin{center}
\tikz{\node(a){\includegraphics[scale=0.7]{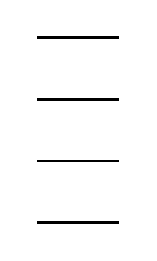}};
	\node[below] at (a.south){00000000};} 
\tikz{\node(a){\includegraphics[scale=0.7]{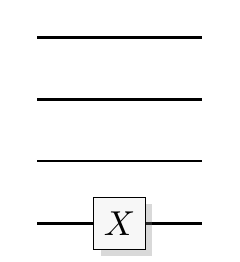}};
	\node[below] at (a.south){11111111};} 
\end{center}
\subsection*{Toffolis: 0, CNOTS: 1}
\begin{center}
\tikz{\node(a){\includegraphics[scale=0.7]{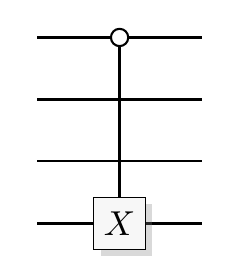}};
	\node[below] at (a.south){00001111};} 
\tikz{\node(a){\includegraphics[scale=0.7]{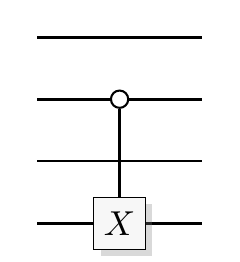}};
	\node[below] at (a.south){00110011};} 
\tikz{\node(a){\includegraphics[scale=0.7]{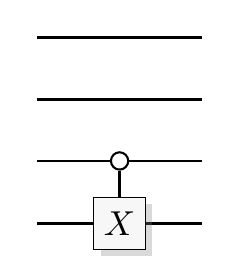}};
	\node[below] at (a.south){01010101};} \\
\tikz{\node(a){\includegraphics[scale=0.7]{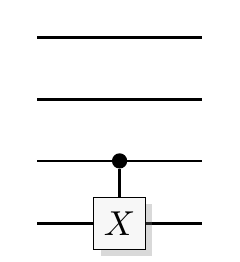}};
	\node[below] at (a.south){10101010};} 
\tikz{\node(a){\includegraphics[scale=0.7]{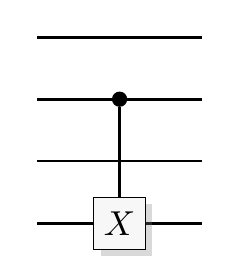}};
	\node[below] at (a.south){11001100};} 
\tikz{\node(a){\includegraphics[scale=0.7]{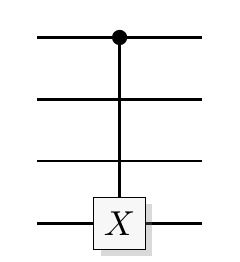}};
	\node[below] at (a.south){11110000};} 
\end{center}
\subsection*{Toffolis: 0, CNOTS: 2}
\begin{center}
\tikz{\node(a){\includegraphics[scale=0.7]{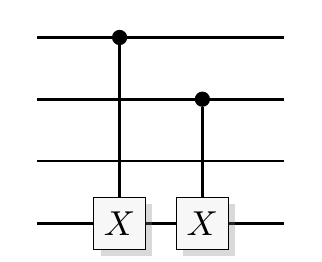}};
	\node[below] at (a.south){00111100};} 
\tikz{\node(a){\includegraphics[scale=0.7]{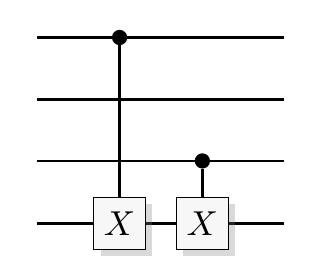}};
	\node[below] at (a.south){01011010};} 
\tikz{\node(a){\includegraphics[scale=0.7]{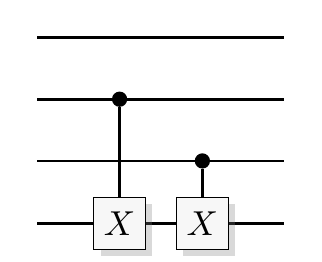}};
	\node[below] at (a.south){01100110};} 
\tikz{\node(a){\includegraphics[scale=0.7]{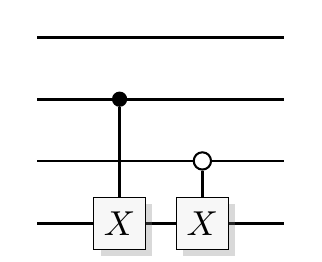}};
	\node[below] at (a.south){10011001};} 
\tikz{\node(a){\includegraphics[scale=0.7]{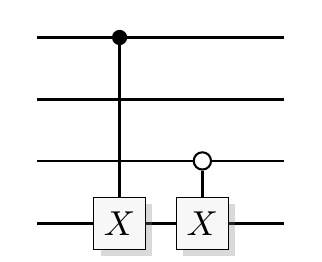}};
	\node[below] at (a.south){10100101};} 
\tikz{\node(a){\includegraphics[scale=0.7]{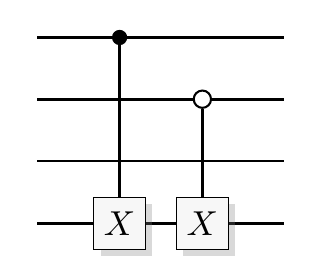}};
	\node[below] at (a.south){11000011};} 
\end{center}
\subsection*{Toffolis: 0, CNOTS: 3}
\begin{center}
\tikz{\node(a){\includegraphics[scale=0.7]{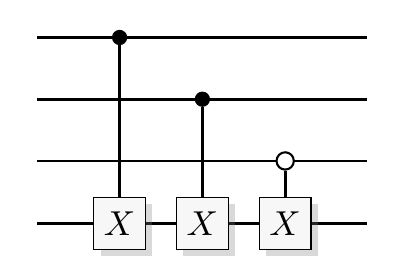}};
	\node[below] at (a.south){01101001};} 
\tikz{\node(a){\includegraphics[scale=0.7]{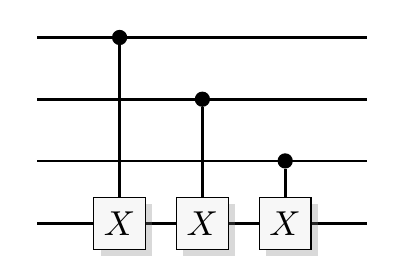}};
	\node[below] at (a.south){10010110};} 
\end{center}
\subsection*{Toffolis: 1, CNOTS: 1}
\begin{center}
\tikz{\node(a){\includegraphics[scale=0.7]{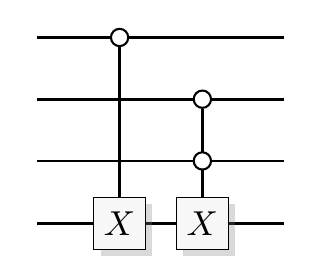}};
	\node[below] at (a.south){00011110};} 
\tikz{\node(a){\includegraphics[scale=0.7]{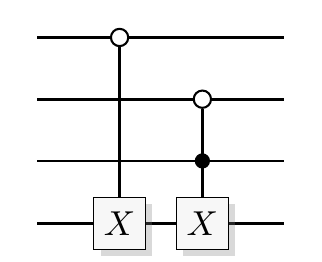}};
	\node[below] at (a.south){00101101};} 
\tikz{\node(a){\includegraphics[scale=0.7]{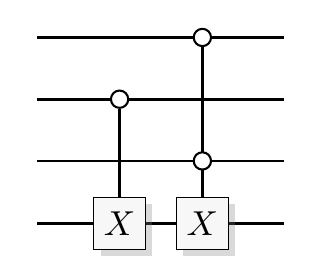}};
	\node[below] at (a.south){00110110};} 
\tikz{\node(a){\includegraphics[scale=0.7]{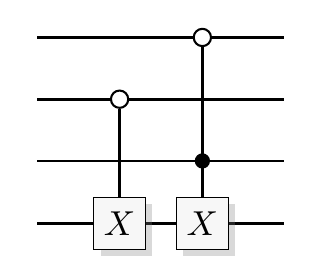}};
	\node[below] at (a.south){00111001};} 
\tikz{\node(a){\includegraphics[scale=0.7]{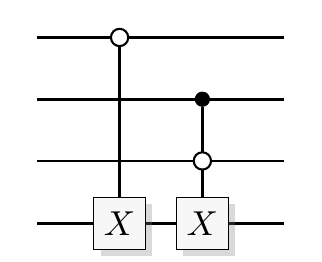}};
	\node[below] at (a.south){01001011};} 
\tikz{\node(a){\includegraphics[scale=0.7]{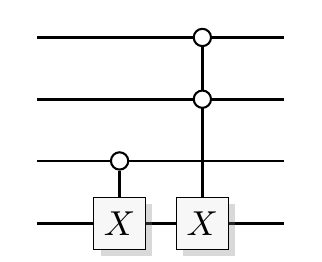}};
	\node[below] at (a.south){01010110};} 
\tikz{\node(a){\includegraphics[scale=0.7]{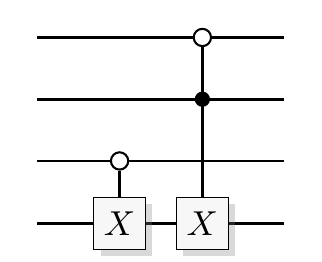}};
	\node[below] at (a.south){01011001};} 
\tikz{\node(a){\includegraphics[scale=0.7]{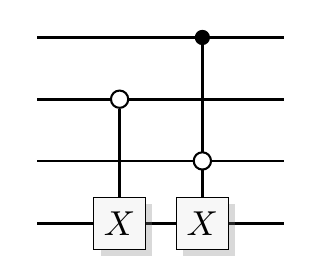}};
	\node[below] at (a.south){01100011};} 
\tikz{\node(a){\includegraphics[scale=0.7]{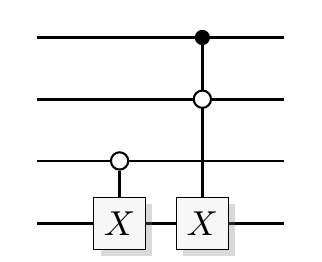}};
	\node[below] at (a.south){01100101};} 
\tikz{\node(a){\includegraphics[scale=0.7]{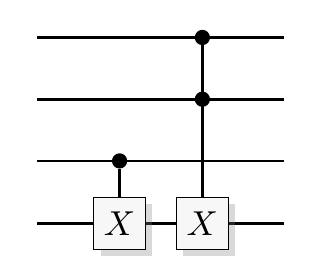}};
	\node[below] at (a.south){01101010};} 
\tikz{\node(a){\includegraphics[scale=0.7]{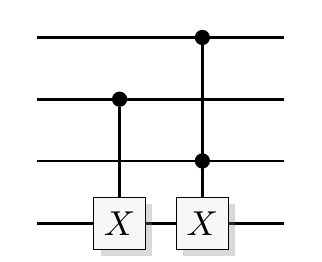}};
	\node[below] at (a.south){01101100};} 
\tikz{\node(a){\includegraphics[scale=0.7]{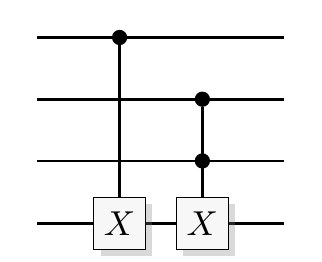}};
	\node[below] at (a.south){01111000};} 
\tikz{\node(a){\includegraphics[scale=0.7]{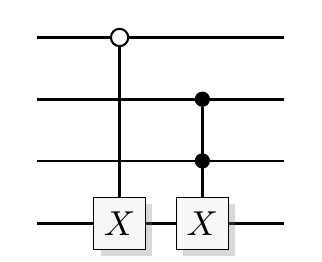}};
	\node[below] at (a.south){10000111};} 
\tikz{\node(a){\includegraphics[scale=0.7]{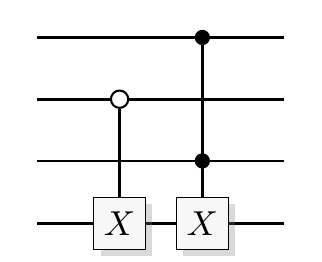}};
	\node[below] at (a.south){10010011};} 
\tikz{\node(a){\includegraphics[scale=0.7]{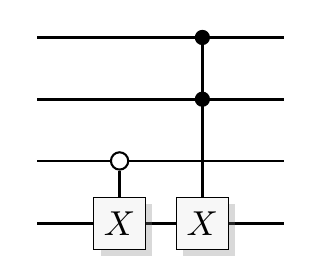}};
	\node[below] at (a.south){10010101};} 
\tikz{\node(a){\includegraphics[scale=0.7]{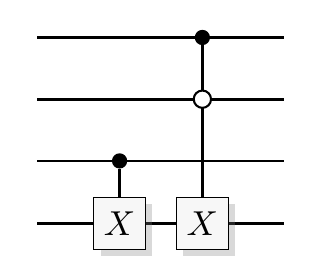}};
	\node[below] at (a.south){10011010};} 
\tikz{\node(a){\includegraphics[scale=0.7]{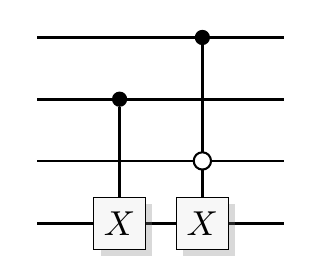}};
	\node[below] at (a.south){10011100};} 
\tikz{\node(a){\includegraphics[scale=0.7]{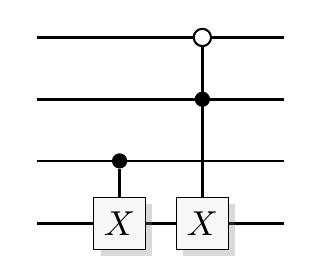}};
	\node[below] at (a.south){10100110};} 
\tikz{\node(a){\includegraphics[scale=0.7]{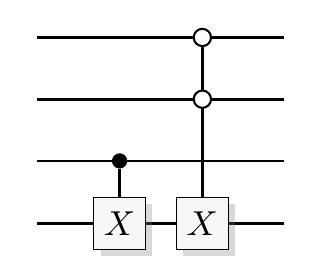}};
	\node[below] at (a.south){10101001};} 
\tikz{\node(a){\includegraphics[scale=0.7]{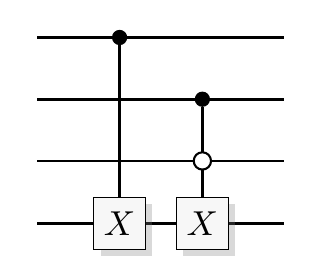}};
	\node[below] at (a.south){10110100};} 
\tikz{\node(a){\includegraphics[scale=0.7]{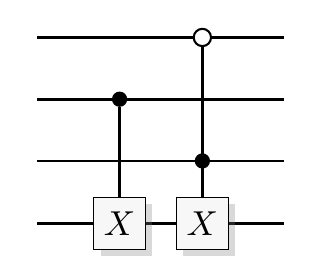}};
	\node[below] at (a.south){11000110};} 
\tikz{\node(a){\includegraphics[scale=0.7]{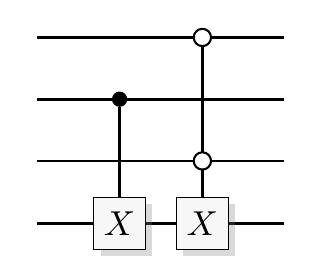}};
	\node[below] at (a.south){11001001};} 
\tikz{\node(a){\includegraphics[scale=0.7]{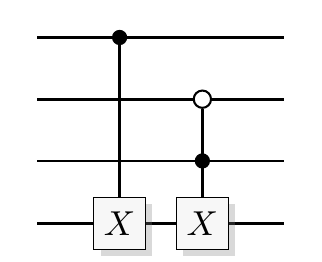}};
	\node[below] at (a.south){11010010};} 
\tikz{\node(a){\includegraphics[scale=0.7]{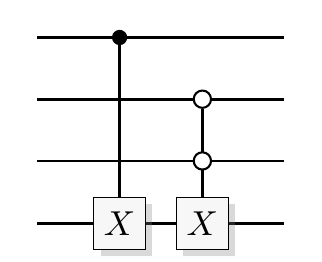}};
	\node[below] at (a.south){11100001};} 
\end{center}
\subsection*{Toffolis: 1, CNOTS: 3}
\begin{center}
\tikz{\node(a){\includegraphics[scale=0.7]{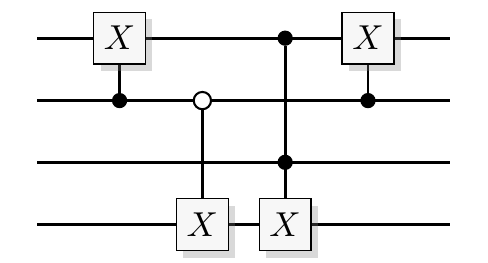}};
	\node[below] at (a.south){00011011};} 
\tikz{\node(a){\includegraphics[scale=0.7]{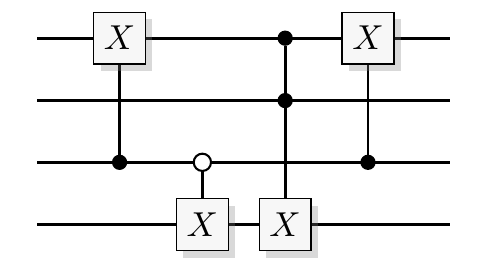}};
	\node[below] at (a.south){00011101};} 
\tikz{\node(a){\includegraphics[scale=0.7]{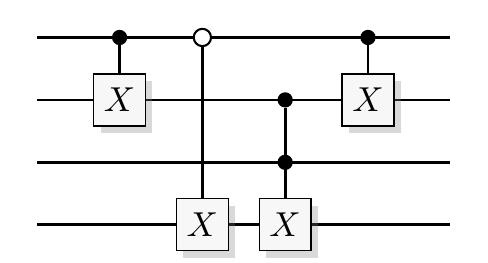}};
	\node[below] at (a.south){00100111};} 
\tikz{\node(a){\includegraphics[scale=0.7]{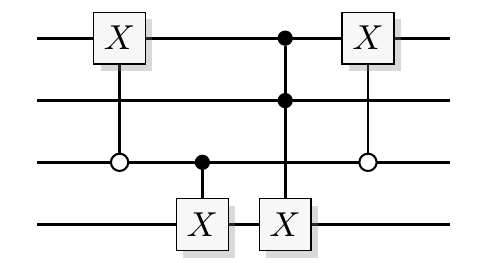}};
	\node[below] at (a.south){00101110};} 
\tikz{\node(a){\includegraphics[scale=0.7]{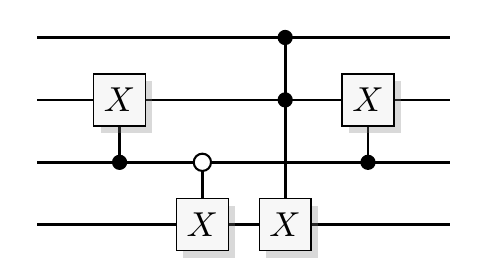}};
	\node[below] at (a.south){00110101};} 
\tikz{\node(a){\includegraphics[scale=0.7]{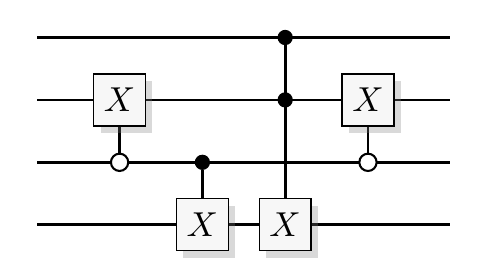}};
	\node[below] at (a.south){00111010};} 
\tikz{\node(a){\includegraphics[scale=0.7]{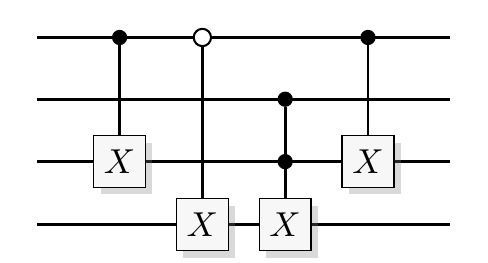}};
	\node[below] at (a.south){01000111};} 
\tikz{\node(a){\includegraphics[scale=0.7]{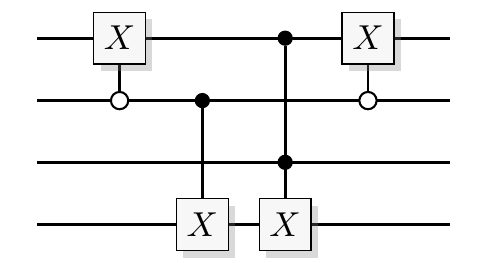}};
	\node[below] at (a.south){01001110};} 
\tikz{\node(a){\includegraphics[scale=0.7]{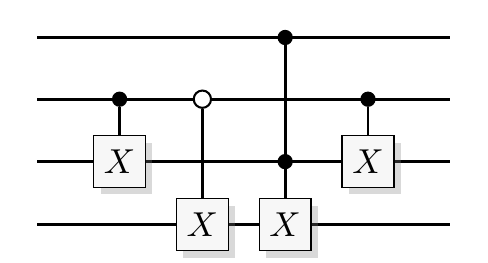}};
	\node[below] at (a.south){01010011};} 
\tikz{\node(a){\includegraphics[scale=0.7]{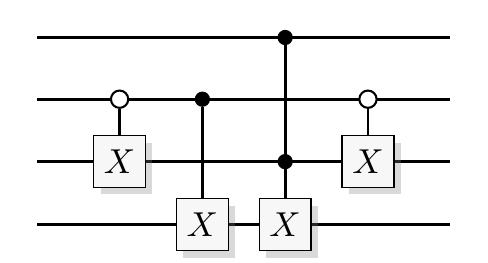}};
	\node[below] at (a.south){01011100};} 
\tikz{\node(a){\includegraphics[scale=0.7]{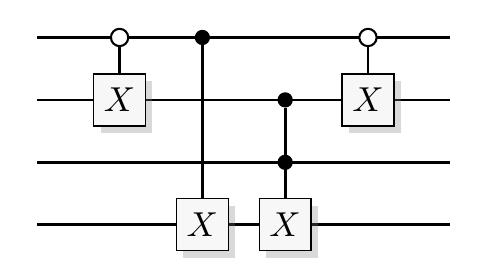}};
	\node[below] at (a.south){01110010};} 
\tikz{\node(a){\includegraphics[scale=0.7]{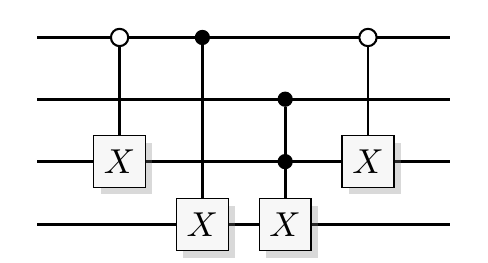}};
	\node[below] at (a.south){01110100};} 
\tikz{\node(a){\includegraphics[scale=0.7]{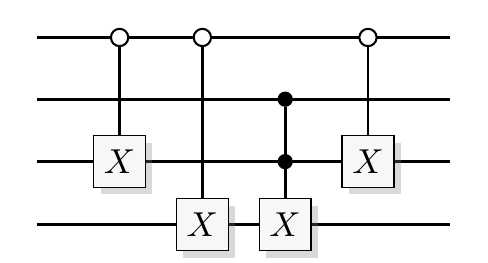}};
	\node[below] at (a.south){10001011};} 
\tikz{\node(a){\includegraphics[scale=0.7]{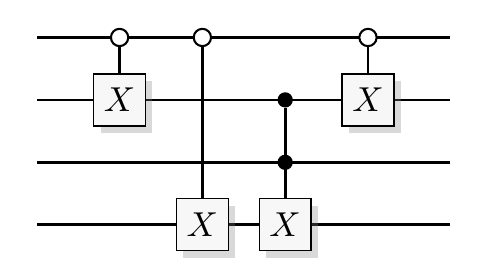}};
	\node[below] at (a.south){10001101};} 
\tikz{\node(a){\includegraphics[scale=0.7]{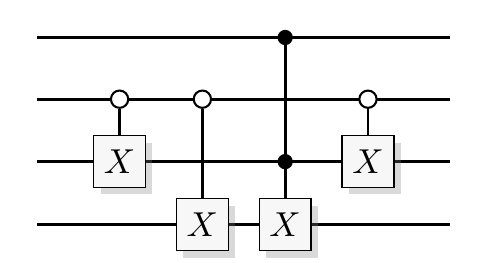}};
	\node[below] at (a.south){10100011};} 
\tikz{\node(a){\includegraphics[scale=0.7]{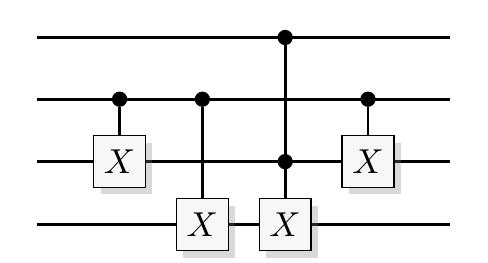}};
	\node[below] at (a.south){10101100};} 
\tikz{\node(a){\includegraphics[scale=0.7]{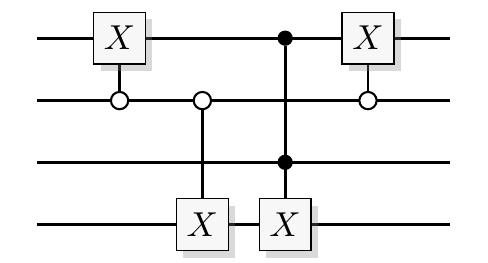}};
	\node[below] at (a.south){10110001};} 
\tikz{\node(a){\includegraphics[scale=0.7]{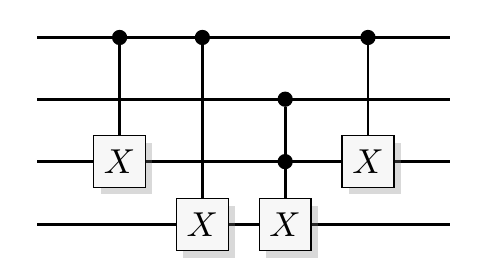}};
	\node[below] at (a.south){10111000};} 
\tikz{\node(a){\includegraphics[scale=0.7]{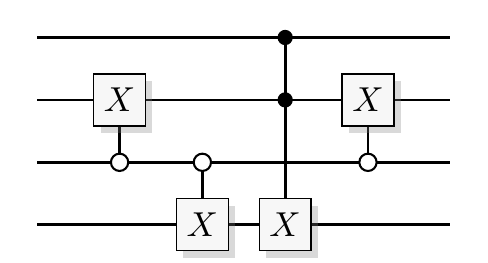}};
	\node[below] at (a.south){11000101};} 
\tikz{\node(a){\includegraphics[scale=0.7]{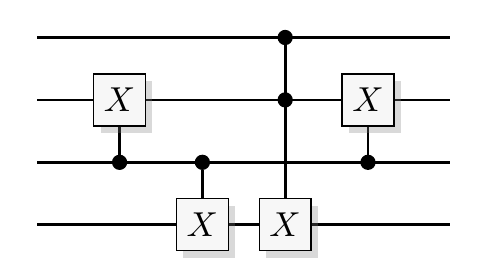}};
	\node[below] at (a.south){11001010};} 
\tikz{\node(a){\includegraphics[scale=0.7]{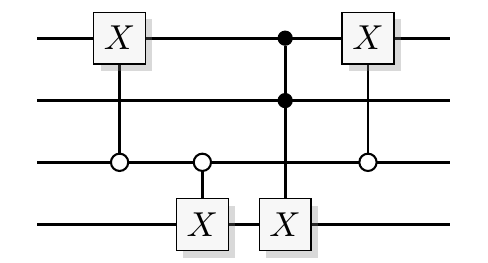}};
	\node[below] at (a.south){11010001};} 
\tikz{\node(a){\includegraphics[scale=0.7]{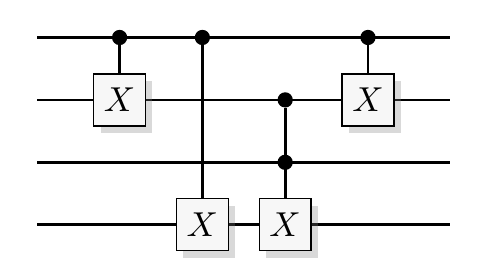}};
	\node[below] at (a.south){11011000};} 
\tikz{\node(a){\includegraphics[scale=0.7]{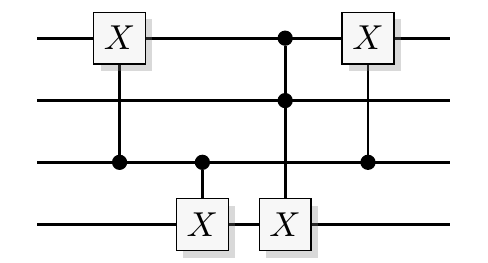}};
	\node[below] at (a.south){11100010};} 
\tikz{\node(a){\includegraphics[scale=0.7]{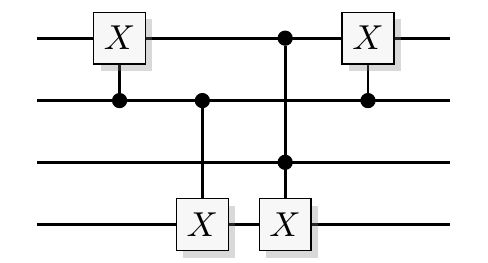}};
	\node[below] at (a.south){11100100};} 
\end{center}
\subsection*{Toffolis: 1, CNOTS: 5}
\begin{center}
\tikz{\node(a){\includegraphics[scale=0.7]{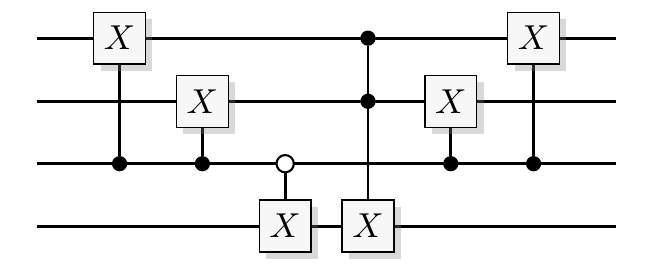}};
	\node[below] at (a.south){00010111};} 
\tikz{\node(a){\includegraphics[scale=0.7]{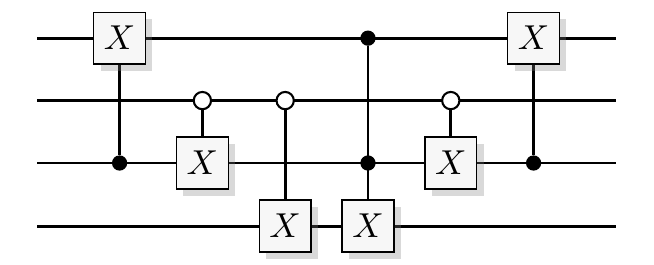}};
	\node[below] at (a.south){00101011};} 
\tikz{\node(a){\includegraphics[scale=0.7]{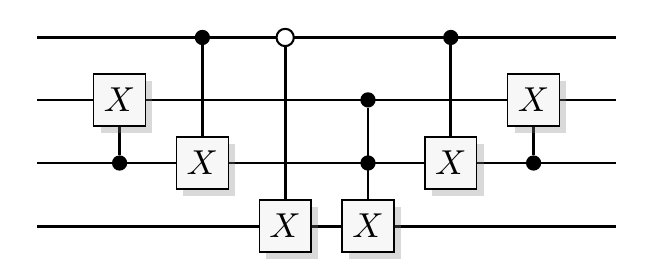}};
	\node[below] at (a.south){01001101};} 
\tikz{\node(a){\includegraphics[scale=0.7]{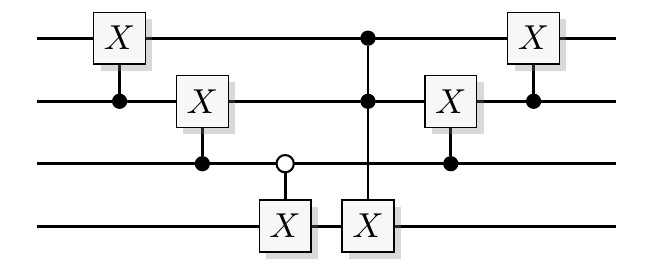}};
	\node[below] at (a.south){01110001};} 
\tikz{\node(a){\includegraphics[scale=0.7]{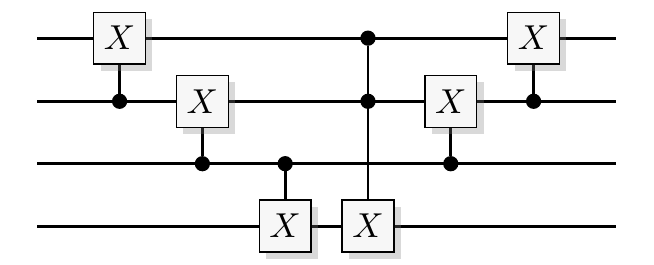}};
	\node[below] at (a.south){10001110};} 
\tikz{\node(a){\includegraphics[scale=0.7]{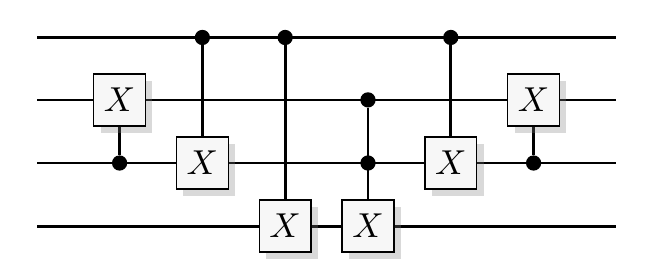}};
	\node[below] at (a.south){10110010};} 
\tikz{\node(a){\includegraphics[scale=0.7]{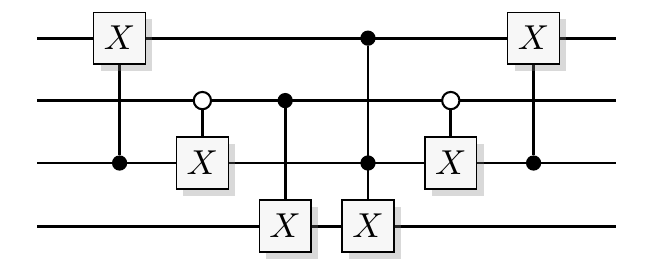}};
	\node[below] at (a.south){11010100};} 
\tikz{\node(a){\includegraphics[scale=0.7]{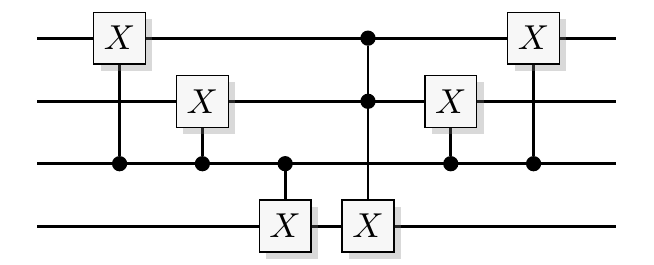}};
	\node[below] at (a.south){11101000};} 
\end{center}

\section{Error Probability for Different Constructions of the Majority
Function}\label{sec:c}
In this section we calculate the error probability for the QSL simulation for
when the four different constructions of the majority function (see
\cref{fig:two_majority}) is used as balanced in the \textsc{Deutsch-Jozsa} algorithm.
We
start with the state after the Walsh-Hadamard transform,
$(r_2,0)\;(r_1,0)\;(r_0,0)\;(r_t,1)$. For clarity we omit the
target
QSL-bit, keeping in mind that its phase bit is set.

For the construction in \cref{fig:two_majority}A as the function, the simulation follows.
\begin{equation}
\begin{aligned}
&(r_2,0)\;(r_1,0)\;(r_0,0)\\
\to\ &(r_2,r_1r_0)\;(r_1,r_2r_0)\;(r_0,r_2r_1)\\
\to\ &(r_2,0)\;(r_1,r_0)\;(r_0,r_1)\\
\to\ &(r_2,\overline{r_1}r_0)\;(r_1,r_0 \oplus r_2r_1)\;(r_0,r_1\oplus r_2\overline{r_1})\\
\to\ &(r_2,r_1\oplus r_0)\;(r_1,r_2\oplus r_0)\;(r_0,r_2\oplus r_1)\\
\end{aligned}
\end{equation}
After the Walsh-Hadamard transform and measurement we see that the simulation will wrongfully answer constant if $r_1\oplus r_0$, $r_2\oplus r_0$, and $r_2\oplus r_1$ all equate to zero. This will happen only if $r_2=r_1=r_0=0$ or $r_2=r_1=r_0=1$, and since $r_2,r_1$ and $r_0$ are equally weighted random bits the error probability becomes $1/4$.

For the construction in \cref{fig:two_majority}B as the function, the simulation follows.
\begin{equation}
\begin{aligned}
&(r_2,0)\;(r_1,0)\;(r_0,0)\\
\to\ &(r_2,0)\;(r_1\oplus r_0,0)\;(r_0,0)\\
\to\ &(r_2\oplus r_0,0)\;(r_1\oplus r_0,0)\;(r_0,0)\\
\to\ &(r_2\oplus r_0,0)\;(r_1\oplus r_0,0)\;(r_0\oplus (r_2\oplus r_0)(r_1\oplus r_0),0)\\
\to\ &(r_2\oplus r_0,0)\;(r_1\oplus r_0,0)\;(r_0\oplus (r_2\oplus r_0)(r_1\oplus r_0),1)\\
\to\ &(r_2\oplus r_0,r_1\oplus r_0)\;(r_1\oplus r_0,r_2\oplus r_0)\;(r_0,1)\\
\to\ &(r_2,r_1\oplus r_0)\;(r_1\oplus r_0,r_2\oplus r_0)\;(r_0,\overline{r_1}\oplus r_0)\\
\to\ &(r_2,r_1\oplus r_0)\;(r_1,r_2\oplus r_0)\;(r_0,r_2\oplus\overline{r_1})\\
\end{aligned}
\end{equation}
After the Walsh-Hadamard transform and measurement we see that the simulation will wrongfully answer constant if $r_1\oplus r_0$, $r_2\oplus r_0$, and $r_2\oplus \overline{r_1}$ all equate to zero, but this cannot happen. If we try to assign values and still obey the constraint we will reach a contradiction. Therefore, the error probability is $0$.

For the construction in \cref{fig:two_majority}C as the function, the simulation follows.
\begin{equation}
\begin{aligned}
&(r_2,0)\;(r_1,0)\;(r_0,0)\\
\to\ &(r_2,r_1)\;(r_1,r_2)\;(r_0,0)\\
\to\ &(r_2,r_1)\;(r_1,r_2\oplus r_0)\;(r_0,r_1)\\
\to\ &(r_2,r_1\oplus r_0)\;(r_1,r_2\oplus r_0)\;(r_0,r_2\oplus r_1)\\
\end{aligned}
\end{equation}
After the Walsh-Hadamard transform and measurement we see that the simulation will wrongfully answer constant if $r_1\oplus r_0$, $r_2\oplus r_0$, and $r_2\oplus r_1$ all equate to zero. Again, this will only happen if $r_2=r_1=r_0=0$ or $r_2=r_1=r_0=1$, and since $r_2,r_1$ and $r_0$ are equally weighted random bits the error probability becomes $1/4$.

For the construction in \cref{fig:two_majority}B as the function, the simulation follows.
\begin{equation}
\begin{aligned}
&(r_2,0)\;(r_1,0)\;(r_0,0)\\
\to\ &(r_2\oplus r_0,0)\;(r_1,0)\;(r_0,0)\\
\to\ &(r_2\oplus r_0,0)\;(r_1,0)\;(r_0,1)\\
\to\ &(r_2\oplus r_0,r_0)\;(r_1,0)\;(r_0,\overline {r_2}\oplus r_0)\\
\to\ &(r_2,r_1\oplus r_0)\;(r_1,r_2\oplus r_0)\;(r_0,\overline{r_2}\oplus r_0)\\
\end{aligned}
\end{equation}
We see that the simulation will wrongfully answer constant if $r_1\oplus r_0$, $r_2\oplus r_0$, and $\overline{r_2}\oplus r_0$ all equate to zero. This is a contradiction since $r_2\oplus r_0\neq \overline{r_2}\oplus r_0$. Therefore, the error probability is $0$.

\end{appendix}
\end{multicols}
\end{document}